\documentclass[11pt,a4paper]{article}
\usepackage{authblk}

\usepackage{algorithm}
\usepackage{algorithmic}

\usepackage{thm-restate}

\usepackage{fullpage}
\usepackage[latin1]{inputenc}
\usepackage{amsthm}
\usepackage{amsmath, amsfonts, amssymb}
\usepackage{dsfont}
\usepackage{xspace}
\usepackage{multirow,longtable}
\usepackage{array}
\usepackage{tabularx}
\usepackage{color}
\usepackage{hyperref}
\usepackage[normalem]{ulem}
\usepackage{tikz,pgf}
\usepackage{subfig}
\usepackage{caption}
\usepackage{pifont}
\usepackage{anyfontsize}
\usetikzlibrary{snakes}
\tikzset{decoration={snake,amplitude=.4mm,segment length=2mm,
                       post length=0mm,pre length=0mm}}

\usepackage{booktabs}
\newcolumntype{C}{>{$\displaystyle}c<{$}}

\usetikzlibrary{decorations.pathreplacing,calc}
\tikzset{%
  middle dotted line/.style={
    decoration={show path construction, 
      lineto code={
          \draw[#1] (\tikzinputsegmentfirst) --($(\tikzinputsegmentfirst)!.3333!(\tikzinputsegmentlast)$);,
          \draw[dotted,#1] ($(\tikzinputsegmentfirst)!.3333!(\tikzinputsegmentlast)$)--($(\tikzinputsegmentfirst)!.6666!(\tikzinputsegmentlast)$);,
          \draw[#1] ($(\tikzinputsegmentfirst)!.6666!(\tikzinputsegmentlast)$)--(\tikzinputsegmentlast);,
      }
    },
    decorate
  },
}

\tikzset{%
  linetwo/.style={
    decoration={show path construction, 
      lineto code={
          \draw[#1] (\tikzinputsegmentfirst) --($(\tikzinputsegmentfirst)!.5!(\tikzinputsegmentlast)$);,
          \draw[snake,#1] ($(\tikzinputsegmentfirst)!.5!(\tikzinputsegmentlast)$)--(\tikzinputsegmentlast);,
          
      }
    },
    decorate
  },
}

\tikzset{%
  linethree/.style={
    decoration={show path construction, 
      lineto code={
          \draw[#1] (\tikzinputsegmentfirst) --($(\tikzinputsegmentfirst)!.3333!(\tikzinputsegmentlast)$);,
          \draw[snake,#1] ($(\tikzinputsegmentfirst)!.3333!(\tikzinputsegmentlast)$)--($(\tikzinputsegmentfirst)!.6666!(\tikzinputsegmentlast)$);,
          \draw[very thick,#1] ($(\tikzinputsegmentfirst)!.6666!(\tikzinputsegmentlast)$)--(\tikzinputsegmentlast);,
          
      }
    },
    decorate
  },
}

\tikzset{%
  linefour/.style={
    decoration={show path construction, 
      lineto code={
          \draw[#1] (\tikzinputsegmentfirst) --($(\tikzinputsegmentfirst)!.25!(\tikzinputsegmentlast)$);,
          \draw[snake,#1] ($(\tikzinputsegmentfirst)!.25!(\tikzinputsegmentlast)$)--($(\tikzinputsegmentfirst)!.5!(\tikzinputsegmentlast)$);,
          \draw[very thick,#1] ($(\tikzinputsegmentfirst)!.5!(\tikzinputsegmentlast)$)--($(\tikzinputsegmentfirst)!.75!(\tikzinputsegmentlast)$);,
          \draw[dashdotted,#1] ($(\tikzinputsegmentfirst)!.75!(\tikzinputsegmentlast)$)--(\tikzinputsegmentlast);,
      }
    },
    decorate
  },
}

\tikzset{%
  linefive/.style={
    decoration={show path construction, 
      lineto code={
          \draw[#1] (\tikzinputsegmentfirst) --($(\tikzinputsegmentfirst)!.2!(\tikzinputsegmentlast)$);,
          \draw[snake,#1] ($(\tikzinputsegmentfirst)!.2!(\tikzinputsegmentlast)$)--($(\tikzinputsegmentfirst)!.4!(\tikzinputsegmentlast)$);,
          \draw[very thick,#1] ($(\tikzinputsegmentfirst)!.4!(\tikzinputsegmentlast)$)--($(\tikzinputsegmentfirst)!.6!(\tikzinputsegmentlast)$);,
          \draw[dashdotted,#1] ($(\tikzinputsegmentfirst)!.6!(\tikzinputsegmentlast)$)--($(\tikzinputsegmentfirst)!.8!(\tikzinputsegmentlast)$);,
          \draw[decorate,decoration={coil,segment length=4pt},#1] ($(\tikzinputsegmentfirst)!.8!(\tikzinputsegmentlast)$)--(\tikzinputsegmentlast);,
      }
    },
    decorate
  },
}

\tikzset{%
  linesix/.style={
    decoration={show path construction, 
      lineto code={
          \draw[#1] (\tikzinputsegmentfirst) --($(\tikzinputsegmentfirst)!.1666!(\tikzinputsegmentlast)$);,
          \draw[snake,#1] ($(\tikzinputsegmentfirst)!.1666!(\tikzinputsegmentlast)$)--($(\tikzinputsegmentfirst)!.3332!(\tikzinputsegmentlast)$);,
          \draw[very thick,#1] ($(\tikzinputsegmentfirst)!.3332!(\tikzinputsegmentlast)$)--($(\tikzinputsegmentfirst)!.4998!(\tikzinputsegmentlast)$);,
          \draw[dashdotted,#1] ($(\tikzinputsegmentfirst)!.4998!(\tikzinputsegmentlast)$)--($(\tikzinputsegmentfirst)!.6664!(\tikzinputsegmentlast)$);,
          
          \draw[decorate,decoration={coil,segment length=4pt},#1] ($(\tikzinputsegmentfirst)!.6664!(\tikzinputsegmentlast)$)--($(\tikzinputsegmentfirst)!.8330!(\tikzinputsegmentlast)$);,
          
          \draw[decorate,decoration={coil,segment length=2pt},#1] ($(\tikzinputsegmentfirst)!.8330!(\tikzinputsegmentlast)$)--(\tikzinputsegmentlast);,
      }
    },
    decorate
  },
}

\tikzset{%
  linedot/.style={
    decoration={show path construction, 
      lineto code={
          \draw[decorate, line width=1pt,decoration={snake,amplitude=.4mm,segment length=2mm,
                       post length=0mm,pre length=0mm, },#1] (\tikzinputsegmentfirst) --($(\tikzinputsegmentfirst)!.3333!(\tikzinputsegmentlast)$);,
          \draw[dotted,#1] ($(\tikzinputsegmentfirst)!.3333!(\tikzinputsegmentlast)$)--($(\tikzinputsegmentfirst)!.6666!(\tikzinputsegmentlast)$);,
          \draw[decorate, line width=1pt,decoration={snake,amplitude=.4mm,segment length=2mm,
                       post length=0mm,pre length=0mm, },#1] ($(\tikzinputsegmentfirst)!.6666!(\tikzinputsegmentlast)$)--(\tikzinputsegmentlast);,
          
      }
    },
    decorate
  },
}
\usepackage{comment}
\usepackage{todonotes}

\newtheorem{theorem}{Theorem}
\newtheorem{proposition}{Proposition}
\newtheorem{lemma}{Lemma}
\newtheorem{claim}{Claim}
\newtheorem{question}{Question}

\newtheorem{definition}{Definition}
\newtheorem{corollary}{Corollary}

\newcommand{\qedclaim}{\hfill $\diamond$ \medskip}
\newenvironment{proofclaim}{\noindent{\it Proof of Claim.}}{\qedclaim}

\newcommand{\FPT}{$\mathsf{FPT}$}
\begin{document}


\title{
Further results on the Hunters and Rabbit game through monotonicity\footnote{A Preliminary version of this paper appeared in the proceedings of MFCS 2023.}}

\author[1]{Thomas Dissaux}
\author[1,2]{Foivos Fioravantes}
\author[3]{Harmender Galhawat}
\author[1]{Nicolas Nisse}
\affil[1]{Universit\'e C\^{o}te d'Azur, Inria, CNRS, I3S, France}
\affil[2]{Department of Theoretical Computer Science, Faculty of Information Technology, Czech Technical University in Prague, Prague, Czech Republic}
\affil[3]{Ben-Gurion University of the Negev, Beersheba, Israel}

\maketitle

\begin{abstract}
The \textsc{Hunters and Rabbit} game is played on a graph $G$ where the Hunter player shoots at $k$ vertices  in every round while the Rabbit player occupies an unknown vertex and, if it is not shot, must move to a neighbouring vertex after each round. The Rabbit player wins if it can ensure that its position is never shot. The Hunter player wins otherwise. The hunter number $h(G)$ of a graph $G$ is the minimum integer $k$ such that the Hunter player has a winning strategy (i.e., allowing him to win whatever be the strategy of the Rabbit player). This game has been studied in several graph classes, in particular in bipartite graphs (grids, trees, hypercubes...), but the computational complexity of computing $h(G)$ remains open in general graphs and even in more restricted graph classes such as trees. To progress further in this study, we propose a notion of monotonicity (a well-studied and useful property in classical pursuit-evasion games such as graph searching games) for the \textsc{Hunters and Rabbit} game imposing that, roughly, a vertex that has already been shot ``must not host the rabbit anymore''. This allows us to obtain new results in various graph classes.

More precisely, let the monotone hunter number $mh(G)$ of a graph $G$ be the minimum integer $k$ such that the Hunter player has a monotone winning strategy. We show that $pw(G) \leq mh(G) \leq pw(G)+1$ for any graph $G$ with pathwidth $pw(G)$, which implies that computing $mh(G)$, or even approximating $mh(G)$ up to an additive constant, is \textsf{NP}-hard. Then, we show that $mh(G)$ can be computed in polynomial time in split graphs, interval graphs, cographs and trees. These results go through structural characterisations which allow us to relate the monotone hunter number with the pathwidth in some of these graph classes. In all cases, this allows us to specify the hunter number or to show that there may be an arbitrary gap between $h$ and $mh$, i.e., that monotonicity does not help. In particular, we show that, for every $k\geq 3$, there exists a tree $T$ with $h(T)=2$ and $mh(T)=k$. We conclude by proving that computing $h$ (resp., $mh$) is \FPT~parameterised by the minimum size of a vertex cover.

\end{abstract}




\section{Introduction}

The \textsc{Hunters and Rabbit} game is played on a graph $G$ and with a fixed integer $k$ (the number of hunters), where the Hunter player shoots at $k$ vertices in every round while the Rabbit player occupies an unknown vertex and, if it is not shot, must move to a neighbouring vertex after each round. The Rabbit player wins if it can ensure that its position is never shot. The Hunter player wins otherwise.
The \textsc{Hunters and Rabbit} game was first introduced in~\cite{FPrincess}, in the case $k=1$, where it was shown that the Hunter player wins in a tree $T$ if and only if $T$ does not contain as subgraph any tree obtained from a star with $3$ leaves by subdividing each edges twice.
This result was also observed in~\cite{H14}, where the authors also consider the minimum number of rounds needed for the Hunter player to win. The version where $k>1$ was first considered in~\cite{AbramovskayaFGP16}. Observe that, if $k=|V(G)|-1$, the Hunter player can win in any connected graph $G$ (in two rounds) by shooting twice a subset of $k$ vertices of $G$. Hence, let the \textit{hunter number} of $G$, denoted by $h(G)$, be the minimum integer $k$ such that $k$ hunters can win in $G$ whatever be the rabbit strategy.
%
The exact value of $h(G)$ 
has been determined for several specific families of graphs $G$. For any $n\geq 2$, $h(P_n)=1$ where $P_n$ is the path with $n$ vertices~\cite{AbramovskayaFGP16} (because the rabbit is forced to move at every round, $h(P_1)=0$). For any $n\geq 3$, $h(C_n)=2$ and $h(K_n)=n-1$, where $C_n$ and $K_n$ are the cycle and complete graph on $n$ vertices respectively~\cite{AbramovskayaFGP16}. Moreover, $h(G_{n\times m})=\lfloor \frac{\min\{n,m\}}{2}\rfloor +1$~\cite{AbramovskayaFGP16} and $h(Q^n)=1+\Sigma_{i=0}^{n-2}{i \choose \lfloor i/2 \rfloor}$~\cite{BOLKEMA2019360}, where $G_{n\times m}$ is the $n\times m$ grid and $Q^n$ is the hypercube with dimension $n$. By taking advantage of the bipartiteness of trees,  it was proven that, for any tree $T$, $h(T)\leq \lceil \frac{1}{2}\log_2(|V(T)|)\rceil$~\cite{gruslys2015catching}. Surprisingly, the computational complexity of the problem that takes a graph $G$ and an integer $k$ as inputs and aims at deciding whether $h(G) \leq k$ is still open, even if $G$ is restricted to be a tree.

In this paper, we progress further in this research direction by exhibiting new classes of graphs $G$ where $h(G)$ can be determined in polynomial time. We also define some {\it monotone} variants of the game which allow us to get new results on the initial game. 
\medskip

\noindent{\bf Graph searching games.}
The \textsc{Hunters and Rabbit} game takes place in the larger class of Graph Searching games initially introduced in~\cite{Br67,parsons1}. In these pursuit-evasion games, one player plays with a team of searchers (also called cops, hunters, etc.) that must track a fugitive (or robber, rabbit, etc.) moving in a graph. There are many games that can fall under this framework, each one specifying its own rules on, for example, the available moves of the searchers, the speed of the fugitive, whether the fugitive is visible or not, and so on. Several variations of graph searching games have been studied in the literature due to their numerous applications in artificial intelligence~\cite{app1}, robot motion planning~\cite{appRobotics}, constraint satisfaction problems and database theory~\cite{appDB3}, and distributed computing~\cite{DBLP:series/lncs/Nisse19}.
Graph Searching games have mostly been studied for their significant implications in graph theory and algorithms. In particular, many variants of these games provide  algorithmic interpretations of several width measures of graphs like treewidth~\cite{seymour}, pathwidth~\cite{parsons1}, tree-depth~\cite{depth}, hypertree-width~\cite{adler}, cycle-rank~\cite{depth}, and directed tree-width~\cite{dtwidth}. The connection between Graph Searching games and structural parameters, such as the treewidth or the pathwidth, is based on the notion of \textit{monotonicity}~\cite{BienstockS91,seymour,DBLP:journals/tcs/MazoitN08,DBLP:journals/dc/IlcinkasNS09}. 
 In short, a searchers' strategy is {\it monotone} if it ensures that the fugitive can never ``recontaminate'' a vertex, i.e., it can never access a vertex that has already been 
 ``visited'' (or ``searched'') by a searcher. The main question is then, given a game, whether ``recontamination does not help in this game''~\cite{LP93}, i.e., whether there always exists, in this game, an optimal (in terms of number of searchers) monotone winning strategy for the searchers. In particular, the monotonicity played a central role in the proof that the minimum number of searchers to capture an invisible (resp., visible) fugitive in the node-searching game played in a graph $G$ equals its pathwidth plus one~\cite{BienstockS91} (resp., treewidth plus one~\cite{seymour}). 

 Not surprisingly, the \textsc{Hunters and Rabbit} game has also a close relationship with the pathwidth of graphs. Precisely, the hunter number of any graph is at most its pathwidth plus one~\cite{AbramovskayaFGP16}. In this paper, we investigate further this relationship and, for this purpose, we define a notion of monotonicity adapted to the \textsc{Hunters and Rabbit} game and study the monotone variant of the game.

\medskip

\noindent{\bf Our contribution.}
In Section~\ref{sec:preliminaries}, we first give the main notation and definitions used throughout this paper, and we prove (or recall from previous works) several basic properties of the hunter number of graphs. In Section~\ref{sec:monotonicity}, we introduce the notion of monotonicity for the \textsc{Hunters and Rabbit} game. As discussed in Section~\ref{sec:monotonicity}, some peculiar behaviours of the \textsc{Hunters and Rabbit} game makes the definition of monotone hunter strategies not as straightforward as in classical Graph Searching games. We then prove, in Section~\ref{ssec:propMonot}, some technical properties (used later) of the monotone hunter number $mh(G)$ of a graph $G$, i.e., the minimum number of hunters needed by a monotone  strategy to win against the rabbit whatever it does in $G$. In Section~\ref{sec:Monotone&Pathwidth}, we prove that $mh(G) \in \{pw(G),pw(G)+1\}$ in any graph $G$. This result has interesting implications. Along with implying that it is \textsf{NP}-hard to compute $mh(G)$ for a graph $G$, it also implies that it is \textsf{NP}-hard to approximate $mh(G)$ up to an additive error of $|V(G)|^\varepsilon$, for $0<\varepsilon<1$. 
On the positive side, we give polynomial-time algorithms to determine $h(G)$ and/or $mh(G)$ in particular graph classes $G$ in Section~\ref{sec:classes}. Precisely, in Section~\ref{sec:split}, we show that $\omega(G)\leq h(G)\leq mh(G) \leq \omega(G)+1$ in any split graph $G$ with maximum clique of size $\omega(G)$ and precisely characterise when each bound is reached. We also precisely characterise $mh(G)$ for any interval graph $G$. In Section~\ref{ssec:cograph}, we design a linear-time algorithm that computes $mh(G)$ for any cograph $G$ and give bounds for $h(G)$ in that case. In Section~\ref{ssec:trees}, we adapt the Parsons' Lemma~\cite{parsons1} to the case of the monotone \textsc{Hunters and Rabbit} game which leads to a polynomial-time algorithm that computes $mh(T)$ for any tree $T$. In Section~\ref{sec:bip}, we investigate the monotonicity property in the case of the ``bipartite'' variant of the \textsc{Hunters and Rabbit} game (see~\cite{AbramovskayaFGP16,gruslys2015catching}). In particular, this allows us to show that, for any $k\in \mathbb{N}$, there exist trees $T$ such that $h(T)=2$ and $mh(T)\geq k$. That is, ``recontamination helps a lot'' in the \textsc{Hunters and Rabbit} game.
Finally, in Section~\ref{sec:kernel}, we show as a general positive result that the problem of deciding if $h(G)\leq k$, for some given integer $k$, is in \FPT~when parameterised by the vertex cover number of $G$. This is done through kernelisation. We close our study by providing directions for further research in Section~\ref{sec:futureDirections}.

\section{Preliminaries}\label{sec:preliminaries}


Unless mentioned otherwise, in this paper we will always deal with graphs $G=(V,E)$ that are non empty, finite, undirected, connected and simple. For any two adjacent vertices $x,y \in V$, let $xy \in E$ denote the edge between $x$ and $y$. Given a set $S \subseteq V$, let $G[S]$ denote the subgraph of $G$ induced by (the vertices in) $S$ and let $G\setminus S$ denote the subgraph $G[V \setminus S]$. For any $v \in V$ and $X \subseteq V$, let $N_X(v) = \{u \in X \mid uv \in E\}$ be the \textit{open neighbourhood} of $v$ in $X$ and let the \textit{closed neighbourhood} of $v$ in $X$ be $N_X[v]=(N_X(v) \cup \{v\}) \cap X$. If $X=V$, we simply write $N(v)$ and $N[v]$ respectively. For any $S \subseteq V$, let $N(S) = \bigcup_{v\in S} N(v)\setminus S$ and $N[S]=N(S)\cup S$. The degree $d(v)=|N(v)|$ is the number of neighbours of $v$ and let $\delta(G)=\min_{v \in V} d(v)$. An {\it independent set} of a graph $G=(V,E)$ is a subset $I$  of $V$ such that, for every $u,v \in I$, $uv \notin E$. A graph is {\it bipartite} if its vertex-set can be partitioned into two independent sets.

\paragraph{\textsc{Hunters and Rabbit} game.} The \textsc{Hunters and Rabbit} game is played between two players, Hunter and Rabbit, on a non empty, finite, undirected, connected and simple graph $G=(V,E)$. Let $k \in \mathbb{N}^*$. The Hunter player controls $k$ hunters and the Rabbit player controls a single rabbit. First, the Rabbit player places the rabbit at a vertex $r_0 \in V$. The rabbit is \textit{invisible}, that is, the position of the rabbit is not known to the hunters.
Then, the game proceeds in \textit{rounds}. In each round $i \geq 1$, first, the Hunter player selects a non empty subset $S_i \subseteq V$ of at most $k$ vertices of $G$ (we say that the vertices in $S_i$ are \textit{shot} at round $i$). If the current position $r_{i-1}$ of the rabbit is shot, i.e., if $r_{i-1} \in S_i$ (we say that the rabbit is shot), then the Hunter player wins, and the game stops. Otherwise, the rabbit must move from its current position $r_{i-1}$ to a vertex $r_i \in N(r_{i-1})$, and the next round starts. The Rabbit wins if it avoids being shot forever. 

 
 A {\it hunter strategy} in $G=(V,E)$ is a finite sequence ${\cal S}=(S_1,\dots,S_{\ell})$ of non empty subsets of vertices of $G$. Let $h({\cal S}):=\max_{1 \leq i \leq \ell} |S_i|$ and let us say that $\cal S$ {\it uses} $h({\cal S})$ hunters.  
 A {\it rabbit trajectory in $G$ starting from $W \subseteq V$} ($W$ will always be assumed non empty) is any walk $(r_0,\dots,r_{\ell})$ starting from $W$, {\it i.e.}, $r_0 \in W$ and $r_i \in N(r_{i-1})$ for every $1 \leq i \leq \ell$. A hunter strategy is {\it winning with respect to $W$} if, for every rabbit trajectory $(r_0,\dots,r_{\ell})$ starting from $W$, there exists $0 \leq j <\ell$ such that $r_j \in S_{j+1}$, that is, the rabbit is eventually shot whatever be its trajectory starting from $W$. Given a hunter strategy ${\cal S}=(S_1,\dots,S_{\ell})$, a rabbit trajectory $(r_0,\dots,r_{\ell})$ starting from $W$ is {\it winning against $\cal S$} if $r_{i} \notin S_{i+1}$ for every $0 \leq i < \ell$. A {\it winning hunter strategy} is any winning hunter strategy with respect to $V$ and a {\it rabbit trajectory} is any rabbit trajectory starting from $V$.
 
 The {\it hunter number of $G=(V,E)$ with respect to $W \subseteq V$}, denoted by $h_W(G)$, is the minimum integer $k$ such that there exists a winning hunter strategy with respect to $W$ and using $k$ hunters. Let $h(G)=h_V(G)$ be the \textit{hunter number} of $G$. Note that, for technical reasons, for a single vertex graph $G$, we set $h(G)=0$. This goes in accordance with ``the locating part'' of the game since the rabbit is already located. The Rabbit player has a {\it strategy $\cal R$ starting from $W \subseteq V$ against $k\geq 1$ hunters} if, for every hunter strategy $\cal S$ using $k$ hunters, there exists a rabbit trajectory ${\cal R}({\cal S})$ that is winning against $\cal S$. Note that, if such a strategy $\cal R$ exists, then $h_W(G)>k$. 
 
The following lemmas will be used throughout this paper.   In~\cite{AbramovskayaFGP16}, it is shown that the hunter number is closed under taking subgraphs. We first show that this result trivially extends to the case when the starting positions of the rabbit are restricted.

 
 \begin{lemma}\label{lem:subgraph}
Let $G=(V,E)$ be any graph and let $H$ be a subgraph of $G$, and let $W \subseteq V$ with $W \cap V(H) \neq \emptyset$. Then, $h_{W \cap V(H)}(H) \leq h_W(G) \leq h(G)$.
\end{lemma}
\begin{proof}
By definition, $h_W(G) \leq h(G)$. Let us show the other inequality.

Let $\mathcal{S}=(S_1,\dots, S_\ell)$ be a winning hunter strategy in $G$ with respect to $W$. Let $\mathcal{S}'=(S'_1, S_2',\dots, S_\ell')$ be such that, for every $1 \leq i \leq \ell$, $S'_i=S_i \cap V(H)$ if $S_i \cap V(H) \neq \emptyset$ and $S'_i$ consists of any vertex of $V(H)$ otherwise. Then, ${\cal S}'$  is a winning hunter strategy in $H$ with respect to $W \cap V(H)$. Indeed, any rabbit trajectory $(r_0 \in W \cap V(H),r_1,\dots,r_{\ell})$ in $H$ is also a trajectory starting from $W$ in $G$. Since $\cal S$ is winning w.r.t. $W$, there exists $i< \ell$ such that $r_i \in S_{i+1}\cap V(H) \subseteq S'_{i+1}$, and so ${\cal S}'$ is winning w.r.t. $W \cap V(H)$. Moreover,  $h({\cal S}') \leq h({\cal S})$.
\end{proof}

For any hunter strategy ${\cal S}=(S_1,\dots,S_{\ell})$, it will be convenient to identify the potential positions of a rabbit (starting in $W \subseteq V$) after each round. Precisely, let ${\cal Z}^W({\cal S})=(Z^W_0({\cal S}),\dots,Z^W_{\ell}({\cal S}))$ be defined as follows. Let $Z^W_0({\cal S})=W$ and, for every $0 < i \leq \ell$, let $Z^W_i({\cal S})$ be the set of vertices $v$ such that there exists a rabbit trajectory $(r_0,r_1,\dots,r_i=v)$ such that $r_0 \in W$ and, for every $0 \leq j < i$, $r_j \notin S_{j+1}$. Formally, for any $1\leq i\leq \ell$, let $Z^W_{i}({\cal S})=\{x\in V(G)\mid\exists y\in (Z^W_{i-1}({\cal S})\setminus S_i) \wedge (xy\in E(G))\}$. Intuitively, $Z^W_i({\cal S})$ is the set of vertices that the rabbit (starting from some vertex in $W$) can have reached at the end of the $i^{th}$ round without having been shot. We will refer to the vertices in $Z^W_i({\cal S})$ as the {\it contaminated} vertices after round $i$. Note that, if $\cal S$ is winning, then $Z^W_{\ell}({\cal S})=\emptyset$. In what follows, we write $Z_i$ (resp., $Z_i({\cal S})$) instead of $Z^W_i({\cal S})$ when $\cal S$ and $W$ (resp., when $W$) are clear from the context.

We now show that we can only consider hunter strategies that consist only of ``useful shots''. 
A hunter strategy ${\cal S}=(S_1,\dots,S_{\ell})$ is said to be {\it parsimonious} if, for every $1 \leq i \leq \ell$, $S_i \subseteq Z_{i-1}({\cal S})$. Note that, if ${\cal S}$ is parsimonious, then $Z_i \neq \emptyset$ for every $i<\ell$. Note that if ${\cal S}$ is parsimonious, then it can be retrieved only from the sequence ${\cal Z}({\cal S})=(Z_0,\dots,Z_{\ell})$ of the contaminated sets. Indeed, for any $1\leq i\leq \ell$, $S_i=\{w\in Z_{i-1}\mid\exists x\in N(w)\backslash Z_{i}\}$. 

In the following lemma, we establish that we can hunt the rabbit in a parsimonious manner without increasing the number of required hunters.
 
 \begin{lemma}\label{lem:parsimonious}
 For any graph $G=(V,E)$ and any non empty subset $W \subseteq V$, there is a parsimonious winning hunter strategy in $G$ with respect to $W$ and that uses $h_W(G)$ hunters.
 \end{lemma}
 \begin{proof}
 Let $\mathcal{S}=(S_1,\dots, S_\ell)$ be a winning hunter strategy with respect to $W\subseteq V$ using at most $k\geq 1$ hunters. Let $\mathcal{Z}({\cal S})=(Z_0({\cal S}),\dots,Z_{\ell}({\cal S}))$ be the set of contaminated vertices for each round of $\mathcal{S}$. If there exists an integer $\ell' < \ell$ such that $Z_{\ell'}({\cal S}) = \emptyset$, then $\mathcal{S}=(S_1,\dots, S_{\ell'})$ is also a winning hunter strategy with respect to $W\subseteq V$ using at most $k$ hunters. Hence, we may assume that $Z_{i}({\cal S}) \neq \emptyset$ for every $0 \leq i < \ell$. 

Moreover, if there exists an integer $1 \leq i \leq \ell$ such that $S_i \cap  Z_{i-1}({\cal S}) = \emptyset$, let $h$ be the smallest such integer and let $v \in Z_{h-1}({\cal S})$. Then, $\mathcal{S}'=(S_1,\dots,S_{h-1},\{v\},S_{h+1},\dots,S_{\ell})$ is also a winning strategy with respect to $W\subseteq V$ using at most $k\geq 1$ hunters (since $S_h \cap Z_{h-1}({\cal S}) = \emptyset$). By repeating this process, we may assume that, for every $1 \leq i \leq \ell$,  $S_i \cap  Z_{i-1}({\cal S}) \neq \emptyset$.

Let $\mathcal{S}'=(S'_1, S_2',\dots, S_{\ell'}')$ be such that, for every $1 \leq i \leq \ell'$, $S'_i=S_i \cap Z_{i-1}({\cal S})$. 
It is easy to see that, for every $i \leq \ell'$, $Z_i({\cal S})=Z_i({\cal S}')$, and then ${\cal S}'$ is parsimonious. Furthermore, ${\cal S}'$  is a winning hunter strategy with respect to $W$. Indeed, since $\cal S$ is winning w.r.t. $W$,  for any rabbit trajectory $(r_0,r_1,\dots,r_{\ell})$, there exists an integer $j < \ell$ such that $r_j \in S_{j+1}$. Let $i$ be the smallest such integer. By definition, $r_i \in Z_i \cap S_{i+1}=S'_{i+1}$ and so ${\cal S}'$ is winning w.r.t. $W$.
Moreover,  $h({\cal S}') \leq h({\cal S})$.
 \end{proof}

 It must be noticed that there exist graphs $G=(V,E)$ and hunter strategies $(S_1,\dots,S_{\ell})$ that are winning in $G$ without shooting to all vertices, i.e., such that $V \setminus \bigcup_{1 \leq i \leq \ell} S_i \neq \emptyset$. For instance, in the graph $G$ that consists of a single edge $uv$, the strategy $(\{u\},\{u\})$ is a winning hunter strategy using one hunter and without shooting at $v$. Note that, in that example, there exists no winning parsimonious hunter strategy using one hunter and that shots to both $u$ and $v$. The next lemma, that characterises the set of such unshot vertices,  will be used throughout the paper.



\begin{lemma}\label{lem:independent}
Let $H$ be any non-empty connected subgraph of any graph $G=(V,E)$. Let $W \subseteq V$ such that $W\cap V(H)\neq \emptyset$. Let ${\cal S}=(S_1,\dots, S_\ell)$ be any winning hunter strategy in $G$ with respect to $W$. If $S_i\cap V(H)=\emptyset$ for all $1\leq i\leq \ell$, then $|V(H)|=1$.
\end{lemma}
\begin{proof}
Let $x \in V(H)\cap W$. Towards a contradiction, assume that $|V(H)|\geq 2$. Let $y\in N_H(x)$ (it exists since $H$ is connected). Note that since $S_i\cap V(H)=\emptyset$ for all $1\leq i\leq \ell$, $\{x,y\} \cap \bigcup_{1\leq i\leq \ell} S_i = \emptyset$. Thus, the rabbit can oscillate between $x$ and $y$ during the whole game without being shot. That is, ${\cal R}=(r_0=x,r_1=y,r_2=x,\dots,r_\ell)$ is a winning rabbit trajectory against ${\cal S}$ starting from $W\cap V(H)$. This contradicts that ${\cal S}$ is a winning hunter strategy in $G$ with respect to $W$.
\end{proof}

In what follows, we will use the following result of~\cite{BOLKEMA2019360}:

\begin{lemma}\cite{BOLKEMA2019360}\label{lem:minDegree}
For any graph $G$, $h(G) \geq \delta(G)$.
\end{lemma}

The \textsc{Hunters and Rabbit} game has been particularly studied in bipartite graphs~\cite{AbramovskayaFGP16,BOLKEMA2019360,gruslys2015catching} and we continue this study in Section~\ref{sec:bip}. In what follows, bipartite graphs are referred to as $G=(V_r\cup V_w,E)$ where $(V_r,V_w)$ is implicitly a bipartition of $V(G)$ such that $V_r$ and $V_w$ are independent sets respectively. We refer to the vertices in $V_r$ (resp., in $V_w$) as the {\it red} (resp., {\it white}) vertices. 
 
In~\cite{AbramovskayaFGP16}, it is shown that, in bipartite graphs, it is sufficient to consider winning hunter strategies with respect to one of the independent sets of the bipartition. For completeness and to further motivate some of our results, we briefly recall their result. Precisely:

\begin{lemma}\cite{AbramovskayaFGP16}\label{lem:bipartition}
For any bipartite graph $G=(V_r \cup V_w,E)$, $h(G)=h_{V_r}(G)=h_{V_w}(G)$. 
\end{lemma}
\begin{proof}
By definition,  $\max \{h_{V_r}(G),h_{V_w}(G)\} \leq h(G)$. To show that $h(G) \leq h_{V_r}(G)$ (resp.,  $h(G) \leq h_{V_w}(G)$), let $\mathcal{S}_r=(S_1,\dots, S_\ell)$ be a winning hunter strategy in $G$ with respect to $V_r$ (resp., w.r.t. $V_w$). If $\ell$ is odd, then $(S_1,\dots, S_\ell,S_1,\dots, S_\ell)$  is a winning hunter strategy, and otherwise, $(S_1,\dots, S_\ell,\{u\},S_1,\dots, S_\ell)$ where $u$ is any arbitrary vertex is a winning hunter strategy.
\end{proof}

Note that, in most of the paper, we will consider hunter strategies with respect to $V$, but in section~\ref{sec:bip}.
More precisely, in Section~\ref{sec:bip}, we will consider the \textsc{Hunters and Rabbit} game in bipartite graphs when the rabbit must start at some vertex of $V_r$. We will refer to this variant as the {\it red variant} of the game. The following remark will be widely used. 

\paragraph{Remark.} Let $G=(V_r \cup V_w,E)$ be a bipartite graph and $\mathcal{S}_r=(S_1,\dots, S_\ell)$ be a parsimonious 
hunter strategy in $G$ with respect to $V_r$. Then, for every $1 \leq i \leq \lceil \ell/2 \rceil$, $S_{2i-1} \subseteq Z_{2i-2} \subseteq V_r$ and (if $2i\leq \ell$) $S_{2i} \subseteq  Z_{2i-1} \subseteq V_w$. Indeed, in a bipartite graph, if the rabbit starts at a vertex in $V_r$ (resp., $V_w$), it must occupy a vertex of $V_r$ at the end of every even (resp., odd) round and a vertex of $V_w$ at the end of every odd  (resp., even)  round.

\section{Monotonicity}\label{sec:monotonicity} 

In classical graph pursuit-evasion games, an important notion is that of {\it monotonicity}. Without going into the details, in these games, a strategy is {\it monotone} if the area reachable by the fugitive never increases. Said differently, in the particular case of graph searching games, a strategy is monotone if, once a searcher is removed from one vertex, it is never necessary to occupy this vertex during a subsequent round (note that, in some specific cases, for instance in directed graphs, these two definitions are not rigorously equivalent~\cite{Adler07}). Monotone strategies have been widely studied~\cite{BienstockS91,YangDA09,DBLP:journals/tcs/MazoitN08} because, on the one hand, it is generally easier to design them and, on the other hand, monotone strategies have length polynomial in the size of the graph and, so, corresponding decision problems (is there a monotone strategy using $k$ searchers?) can be proven to be in \textsf{NP}.

It is clear that such a definition is not suitable to the \textsc{Hunters and Rabbit} game. Indeed, consider the graph that consists of a single edge $uv$: the hunter must shoot at some vertex, say $u$, and, if the rabbit was at $v$, it will move to $u$, i.e., the vertex $u$ is ``recontaminated''. Therefore, we propose to define monotonicity in the \textsc{Hunters and Rabbit} game as follows (see the formal definition below): once a vertex has been 
``cleared'', if the rabbit can access it in a subsequent round, then the vertex must be shot immediately. 

In classical graph searching games, a vertex being cleared at some round means that the searchers' strategy ensures that the fugitive cannot occupy this vertex at this round. Being recontaminated can then be intuitively defined by the fact that a vertex can be reached by the fugitive while having been cleared in a previous round. This intuitive definition does not make any sense in the \textsc{Hunters and Rabbit} game and, in particular, in its red variant in bipartite graphs. Indeed, in such case, every red vertex is cleared at every odd round and so, looking for a strategy without recontamination would be meaningless. To overcome this difficulty, we propose to define the clearing of a vertex at some round by the fact that the actions of the hunters ensure that this vertex cannot be occupied by the rabbit at this round. 

A related difficulty comes from the fact that, contrary to classical graph searching games, a vertex may be ``cleared'' without having been shot during the game.  Recall, for instance, our previous discussion for the graph consisting of a single edge. As a less trivial example, consider a star with three leaves whose edges have been subdivided once each. Then, assuming that the leaves and the centre are red, in the red variant, it is possible for one hunter to win without shooting any of the leaves (while any of the leaves may be occupied by the rabbit initially). Indeed, consider the strategy for one hunter where on every odd round it shoots on the centre and on every even round it shoots on an arbitrary neighbour of the centre that was not previously shot. Figure~\ref{fig:example1} illustrates the above strategy.

\begin{figure}[!t]
\centering
\captionsetup{justification=centering}
\subfloat[$Z_0=\{a,c,e,g\}$]{
\scalebox{0.4}{
\begin{tikzpicture}[inner sep=1.5mm]
    \node[draw, circle, line width=1pt, fill=red](u1) at (2,8)[label=left: {\Huge$a$}] {};
    \node[draw, circle, line width=1pt, fill=white](u2) at (2,6)[label=left: {\Huge$b$}] {};
    \node[draw, circle, line width=1pt, fill=red](u3) at (2,4)[label=left: {\Huge$c$}] {};
    \node[draw, circle, line width=1pt, fill=white](u4) at (0.5,2)[label=right: {\Huge$d$}] {};
    \node[draw, circle, line width=1pt, fill=red](u5) at (-1,0)[label=right: {\Huge$e$}] {};
    \node[draw, circle, line width=1pt, fill=white](u6) at (3.5,2)[label=left: {\Huge$f$}] {};
    \node[draw, circle, line width=1pt, fill=red](u7) at (5,0)[label=left: {\Huge$g$}] {};
    
    \node[] at (6,0) {\includegraphics[scale=0.4]{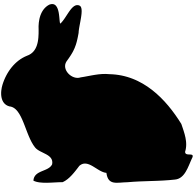}};
    \node[] at (-2,0) {\includegraphics[scale=0.4]{figure/rabbit.png}};
    \node[] at (3,4) {\includegraphics[scale=0.4]{figure/rabbit.png}};
    \node[] at (3,8) {\includegraphics[scale=0.4]{figure/rabbit.png}};
	
	\draw[-, line width=1pt]  (u1) -- (u2);
    \draw[-, line width=1pt]  (u2) -- (u3);
    \draw[-, line width=1pt]  (u3) -- (u4);
    \draw[-, line width=1pt]  (u3) -- (u6);
    \draw[-, line width=1pt]  (u4) -- (u5);
    \draw[-, line width=1pt]  (u6) -- (u7);

\end{tikzpicture}
}
}
\hspace{40pt}
\subfloat[$S_1=\{c\}$, \newline $Z_1=\{b,d,f\}$]{
\scalebox{0.4}{
\begin{tikzpicture}[inner sep=1.5mm]
\captionsetup{justification=centering}

    \node[draw, circle, line width=1pt, fill=red](u1) at (2,8)[label=left: {\Huge$a$}] {};
    \node[draw, circle, line width=1pt, fill=white](u2) at (2,6)[label=left: {\Huge$b$}] {};
    \node[draw, circle, line width=1pt, fill=red](u3) at (2,4)[label=left: {\Huge$c$}] {};
    \node[draw, circle, line width=1pt, fill=white](u4) at (0.5,2)[label=right: {\Huge$d$}] {};
    \node[draw, circle, line width=1pt, fill=red](u5) at (-1,0)[label=right: {\Huge$e$}] {};
    \node[draw, circle, line width=1pt, fill=white](u6) at (3.5,2)[label=left: {\Huge$f$}] {};
    \node[draw, circle, line width=1pt, fill=red](u7) at (5,0)[label=left: {\Huge$g$}] {};
	
	\draw[-, line width=1pt]  (u1) -- (u2);
    \draw[-, line width=1pt]  (u2) -- (u3);
    \draw[-, line width=1pt]  (u3) -- (u4);
    \draw[-, line width=1pt]  (u3) -- (u6);
    \draw[-, line width=1pt]  (u4) -- (u5);
    \draw[-, line width=1pt]  (u6) -- (u7);

    \node[](t) at (u3) {\fontsize{50pt}{58pt}\selectfont \ding{55}};

    \node[] at (3,6) {\includegraphics[scale=0.4]{figure/rabbit.png}};
    \node[] at (-0.5,2) {\includegraphics[scale=0.4]{figure/rabbit.png}};
    \node[] at (4.5,2) {\includegraphics[scale=0.4]{figure/rabbit.png}};

\end{tikzpicture}
}
}
\hspace{40pt}
\subfloat[$S_2=\{d\}$, $Z_2=\{a,c,g\}$]{
\scalebox{0.4}{
\begin{tikzpicture}[inner sep=1.5mm]

    \node[draw, circle, line width=1pt, fill=red](u1) at (2,8)[label=left: {\Huge$a$}] {};
    \node[draw, circle, line width=1pt, fill=white](u2) at (2,6)[label=left: {\Huge$b$}] {};
    \node[draw, circle, line width=1pt, fill=red](u3) at (2,4)[label=left: {\Huge$c$}] {};
    \node[draw, circle, line width=1pt, fill=white](u4) at (0.5,2)[label=right: {\Huge$d$}] {};
    \node[draw, circle, line width=1pt, fill=red](u5) at (-1,0)[label=right: {\Huge$e$}] {};
    \node[draw, circle, line width=1pt, fill=white](u6) at (3.5,2)[label=left: {\Huge$f$}] {};
    \node[draw, circle, line width=1pt, fill=red](u7) at (5,0)[label=left: {\Huge$g$}] {};
	
	\draw[-, line width=1pt]  (u1) -- (u2);
    \draw[-, line width=1pt]  (u2) -- (u3);
    \draw[-, line width=1pt]  (u3) -- (u4);
    \draw[-, line width=1pt]  (u3) -- (u6);
    \draw[-, line width=1pt]  (u4) -- (u5);
    \draw[-, line width=1pt]  (u6) -- (u7);

    \node[](t) at (u4) {\fontsize{50pt}{58pt}\selectfont \ding{55}};

    \node[] at (6,0) {\includegraphics[scale=0.4]{figure/rabbit.png}};
    \node[] at (3,4) {\includegraphics[scale=0.4]{figure/rabbit.png}};
    \node[] at (3,8) {\includegraphics[scale=0.4]{figure/rabbit.png}};

\end{tikzpicture}
}
}

\subfloat[$S_3=\{c\}$, $Z_3=\{b,f\}$]{
\scalebox{0.4}{
\begin{tikzpicture}[inner sep=1.5mm]

    \node[draw, circle, line width=1pt, fill=red](u1) at (2,8)[label=left: {\Huge$a$}] {};
    \node[draw, circle, line width=1pt, fill=white](u2) at (2,6)[label=left: {\Huge$b$}] {};
    \node[draw, circle, line width=1pt, fill=red](u3) at (2,4)[label=left: {\Huge$c$}] {};
    \node[draw, circle, line width=1pt, fill=white](u4) at (0.5,2)[label=right: {\Huge$d$}] {};
    \node[draw, circle, line width=1pt, fill=red](u5) at (-1,0)[label=right: {\Huge$e$}] {};
    \node[draw, circle, line width=1pt, fill=white](u6) at (3.5,2)[label=left: {\Huge$f$}] {};
    \node[draw, circle, line width=1pt, fill=red](u7) at (5,0)[label=left: {\Huge$g$}] {};
	
	\draw[-, line width=1pt]  (u1) -- (u2);
    \draw[-, line width=1pt]  (u2) -- (u3);
    \draw[-, line width=1pt]  (u3) -- (u4);
    \draw[-, line width=1pt]  (u3) -- (u6);
    \draw[-, line width=1pt]  (u4) -- (u5);
    \draw[-, line width=1pt]  (u6) -- (u7);

    \node[](t) at (u3) {\fontsize{50pt}{58pt}\selectfont \ding{55}};

    \node[] at (3,6) {\includegraphics[scale=0.4]{figure/rabbit.png}};
    \node[] at (4.5,2) {\includegraphics[scale=0.4]{figure/rabbit.png}};

\end{tikzpicture}
}
}
\hspace{40pt}
\subfloat[$S_4=\{f\}$, $Z_4=\{a,c\}$]{
\scalebox{0.4}{
\begin{tikzpicture}[inner sep=1.5mm]

    \node[draw, circle, line width=1pt, fill=red](u1) at (2,8)[label=left: {\Huge$a$}] {};
    \node[draw, circle, line width=1pt, fill=white](u2) at (2,6)[label=left: {\Huge$b$}] {};
    \node[draw, circle, line width=1pt, fill=red](u3) at (2,4)[label=left: {\Huge$c$}] {};
    \node[draw, circle, line width=1pt, fill=white](u4) at (0.5,2)[label=right: {\Huge$d$}] {};
    \node[draw, circle, line width=1pt, fill=red](u5) at (-1,0)[label=right: {\Huge$e$}] {};
    \node[draw, circle, line width=1pt, fill=white](u6) at (3.5,2)[label=left: {\Huge$f$}] {};
    \node[draw, circle, line width=1pt, fill=red](u7) at (5,0)[label=left: {\Huge$g$}] {};
	
	\draw[-, line width=1pt]  (u1) -- (u2);
    \draw[-, line width=1pt]  (u2) -- (u3);
    \draw[-, line width=1pt]  (u3) -- (u4);
    \draw[-, line width=1pt]  (u3) -- (u6);
    \draw[-, line width=1pt]  (u4) -- (u5);
    \draw[-, line width=1pt]  (u6) -- (u7);

    \node[](t) at (u6) {\fontsize{50pt}{58pt}\selectfont \ding{55}};

    \node[] at (3,4) {\includegraphics[scale=0.4]{figure/rabbit.png}};
    \node[] at (3,8) {\includegraphics[scale=0.4]{figure/rabbit.png}};

\end{tikzpicture}
}
}
\hspace{40pt}
\subfloat[$S_5=\{c\}$, $Z_5=\{b\}$]{
\scalebox{0.4}{
\begin{tikzpicture}[inner sep=1.5mm]

    \node[draw, circle, line width=1pt, fill=red](u1) at (2,8)[label=left: {\Huge$a$}] {};
    \node[draw, circle, line width=1pt, fill=white](u2) at (2,6)[label=left: {\Huge$b$}] {};
    \node[draw, circle, line width=1pt, fill=red](u3) at (2,4)[label=left: {\Huge$c$}] {};
    \node[draw, circle, line width=1pt, fill=white](u4) at (0.5,2)[label=right: {\Huge$d$}] {};
    \node[draw, circle, line width=1pt, fill=red](u5) at (-1,0)[label=right: {\Huge$e$}] {};
    \node[draw, circle, line width=1pt, fill=white](u6) at (3.5,2)[label=left: {\Huge$f$}] {};
    \node[draw, circle, line width=1pt, fill=red](u7) at (5,0)[label=left: {\Huge$g$}] {};
	
	\draw[-, line width=1pt]  (u1) -- (u2);
    \draw[-, line width=1pt]  (u2) -- (u3);
    \draw[-, line width=1pt]  (u3) -- (u4);
    \draw[-, line width=1pt]  (u3) -- (u6);
    \draw[-, line width=1pt]  (u4) -- (u5);
    \draw[-, line width=1pt]  (u6) -- (u7);

    \node[](t) at (u3) {\fontsize{50pt}{58pt}\selectfont \ding{55}};

    \node[] at (3,6) {\includegraphics[scale=0.4]{figure/rabbit.png}};

\end{tikzpicture}
}
}

\subfloat[$S_6=\{b\}$, $Z_6=\emptyset$]{
\scalebox{0.4}{
\begin{tikzpicture}[inner sep=1.5mm]

    \node[draw, circle, line width=1pt, fill=red](u1) at (2,8)[label=left: {\Huge$a$}] {};
    \node[draw, circle, line width=1pt, fill=white](u2) at (2,6)[label=left: {\Huge$b$}] {};
    \node[draw, circle, line width=1pt, fill=red](u3) at (2,4)[label=left: {\Huge$c$}] {};
    \node[draw, circle, line width=1pt, fill=white](u4) at (0.5,2)[label=right: {\Huge$d$}] {};
    \node[draw, circle, line width=1pt, fill=red](u5) at (-1,0)[label=right: {\Huge$e$}] {};
    \node[draw, circle, line width=1pt, fill=white](u6) at (3.5,2)[label=left: {\Huge$f$}] {};
    \node[draw, circle, line width=1pt, fill=red](u7) at (5,0)[label=left: {\Huge$g$}] {};
	
	\draw[-, line width=1pt]  (u1) -- (u2);
    \draw[-, line width=1pt]  (u2) -- (u3);
    \draw[-, line width=1pt]  (u3) -- (u4);
    \draw[-, line width=1pt]  (u3) -- (u6);
    \draw[-, line width=1pt]  (u4) -- (u5);
    \draw[-, line width=1pt]  (u6) -- (u7);

    \node[](t) at (u2) {\fontsize{50pt}{58pt}\selectfont \ding{55}};

\end{tikzpicture}
}
}
\captionsetup{justification=justified}
\caption{Example of a bipartite graph (where $V_r=\{a,c,e,g\}$ corresponds to the red part of the bipartition, illustrated by the red vertices in the figures) and of a parsimonious winning strategy with respect to $V_r$, such that no vertex in $\{a,e,g\}$ is ever shot. Each subfigure depicts the situation at the end of the corresponding round. In round $0$, the rabbit occupies any vertex in $\{a,c,e,g\}$. Then, in round $1$, the hunter shoots the vertex $c$ (depicted as the cross over the corresponding vertex of subfigure~(b)) and the rabbit moves to one of the vertices in $\{b,d,f\}$. The game continues until the end of round $6$ (subfigure~(g)), at which point the hunter is sure to shoot the rabbit in vertex $b$. Formally, we have ${\cal S} =(\{c\},\{d\},\{c\},\{f\},\{c\},\{b\})$ and ${\cal Z}({\cal S}) = (\{a,c,e,g\}, \{b,d,f\}, \{a,c,g\},\{b,f\},\{a,c\},\{b\}, \emptyset)$.}
\label{fig:example1}
\end{figure}

Therefore, two actions of the hunters may clear a vertex: either a hunter shoots a vertex $v$ at round $i$ and does not shoot the rabbit (i.e. there is no rabbit trajectory, such that $r_{i-1}=v$, that is winning against a strategy shooting at $v$ at round $i$), or the hunters shoot at every contaminated vertex in the neighbourhood of $v$. In this case, either $v$ was occupied and the rabbit has to leave $v$, or it was not and cannot be occupied after the move of the rabbit. In both cases, $v \notin Z_i$. This discussion motivates the following definition for the monotonicity of hunter strategies.

\subsection{Definition of monotone strategies and first properties} \label{ssec:propMonot}

Given a graph $G$, a winning hunter strategy $\mathcal{S}$ in $G$ with respect to $W\subseteq V$, is {\it monotone} if for every vertex $v\in V$, once $v$ has been ``cleared'', then it is shot again every time the rabbit can potentially reach $v$. Formally, we say that a vertex $v$ is {\it cleared} at round $i$ if either $v\in S_i$ or $N(v)\cap Z_{i-1}\neq \emptyset$ and $N(v)\cap Z_{i-1} \subseteq S_i$. Note that, in the second condition, the fact that we require that $N(v)\cap Z_{i-1}\neq \emptyset$ comes from technicalities when $W \neq V$.
A strategy $\mathcal{S}=(S_1,\dots,S_{\ell})$ is {\it monotone} if, for every vertex $v\in V$, if there exists an $i$ such that $v$ is cleared at round $i$, then for every $j>i$ such that $v\in Z_j$, the strategy ensures that $v\in S_{j+1}$. A vertex $v$ is {\it recontaminated} at round $j$ if there exists $i\leq j$ such that $v$ is cleared at round $i$ and $v \in Z_j\setminus S_{j+1}$.

The \textit{monotone hunter number} of a graph $G$ with respect to $W \subseteq V(G)$, denoted by $mh_W(G)$, is the minimum number $k$ such that $k$ hunters have a monotone winning hunter strategy in $G$ with respect to $W$. Let us denote the \textit{monotone hunter number} $mh_{V}(G)$ of $G$ by $mh(G)$. Note that, by definition:
\begin{proposition}\label{prop:relationMonotone&nonM}
For every graph $G=(V,E)$ and $W \subseteq V$, $h_W(G) \leq mh_W(G) \leq mh(G)$.
\end{proposition}

In this subsection, we prove some general properties of (non-)monotone strategies.
Let us start with two technical claims that will be used in several proofs below.

\begin{proposition}\label{prop:reach}
Let $\mathcal{S} = (S_1,\ldots, S_\ell)$ be a hunter strategy in a graph $G=(V,E)$. Let $v\in V$ and $1\leq i \leq \ell$. If there exists a vertex $u \in N(v)$ and a vertex $x\in N(u)$ (possibly $x=v$) such that $u\notin \bigcup_{j\leq i}S_j$ and $x\notin \bigcup_{j< i}S_j$, then $v \in Z_{p}$ for each $p\leq i$.
 \end{proposition}
 \begin{proof}
     This clearly holds if $p=0$ since $Z_0=V$. If $p=1$, there exists a rabbit trajectory $(r_0=u \in N(v) \setminus S_1, r_1=v)$  and so $v \in Z_1$. Hence, we assume that $p>1$. 
    
     The rabbit can follow the following strategy depending on whether $p$ is odd or even:
     \begin{enumerate}
         \item $p$ is odd: The rabbit can follow the following trajectory: $(r_0 = u,r_1=x,r_2=u, \ldots, r_{p-2}=x,r_{p-1} = u, r_{p}=v)$ where, for $q<p$, $r_{q} = u$ if $q$ is even and $r_q = x$ if $q$ is odd. 
         \item $p$ is even: The rabbit can follow the following trajectory: $(r_0 = x,r_1=u,r_2=x, \ldots, r_{p-2}=x, r_{p-1} = u, r_{p}=v)$ where, for $q<p$, $r_{q} = x$ if $q$ is even and $r_q = u$ if $q$ is odd.
     \end{enumerate}
In both cases, for every $0\leq j <p$, $r_j\notin S_{j+1}$ since $p\leq i$, $u \notin \bigcup_{j\leq  i}S_j$ and $x\notin \bigcup_{j< i}S_j$. Therefore, $v\in Z_p$.
 \end{proof}

 The next lemma shows that, as expected, if the hunters follow a monotone strategy, the set of potential positions for the rabbit cannot increase.

 \begin{lemma}\label{lem:reach2}
 Let $G=(V,E)$ be a graph with at least two vertices. Let $\mathcal{S} = (S_1,\ldots, S_\ell)$ be a monotone hunter strategy in $G$. For any $0 \leq p \leq i \leq \ell$, $Z_i \subseteq Z_p$.
 \end{lemma}
 \begin{proof}
     This clearly holds if $p=0$ or $i=1$ since $Z_0=V$. Hence, let us assume that $p\geq 1$ and $i>1$. Let $v \in Z_i$. Since $v\in Z_i$, there exists a rabbit trajectory $R=(r_0,\dots, r_{i-2}=x,r_{i-1}=u, r_i=v)$ such that, for any $0\leq j<i$, $r_j\notin S_{j+1}$. By definition of a rabbit trajectory, $u \in N(v)$ and $x\in N(u)$. Moreover, by monotonicity of ${\cal S}$, since $u\in Z_{i-1}\setminus S_i$ (resp. $x\in Z_{i-2}\setminus S_{i-1}$), $u\notin \bigcup_{q\leq i} S_q$ (resp. $x\notin \bigcup_{q\leq i-1} S_q$). By Proposition~\ref{prop:reach}, $v \in Z_{p}$ for each $p\leq i$.
%
 \end{proof}

The next lemma states that, for any non-monotone strategy, there must exist a vertex that has been shot at some round and that is recontaminated later (recall that it is not trivial since a vertex may be recontaminated without being previously shot).
\begin{lemma}\label{lem:nonMonotone}
    Let $\mathcal{S} = (S_1,\ldots, S_\ell)$ be a non-monotone winning hunter strategy in a graph $G = (V,E)$. Then, there exist a vertex $v\in V$ and $1\leq i \leq \ell$ such that $v\in Z_{i-1}\setminus S_i$ and $v\in \bigcup_{p< i}S_p$.
\end{lemma}
\begin{proof}
    Towards a contradiction, assume that the statement of the lemma is false, i.e., for every vertex $v\in V$ and every $1\leq i \leq \ell$, if $v\in Z_{i-1} \setminus S_i$ then $v\notin \bigcup_{p< i}S_p$. 

    Since $\mathcal{S}$ is non-monotone and winning, there exists a vertex $u$ such that $u$ is cleared at a round $1\leq q\leq \ell -2$, and then recontaminated at a round $j> q$ (i.e., $u\in Z_j\setminus S_{j+1}$). Moreover, by our assumption, $u$ is cleared by shooting each contaminated vertex in $N(u)$ at round $q$, i.e., $Z_{q-1} \cap N(u) \subseteq S_q$.

    Let us show that $N(u) \subseteq \bigcup_{p\leq q} S_p$. Let us assume that there exists a vertex $x \in N(u)$ such that $x\notin \bigcup_{p< q} S_p$.  Since both $u,x\notin \bigcup_{p< q} S_q$ and $u\in N(x)$, by Proposition~\ref{prop:reach}, we get that $x\in Z_{q-1}$. Therefore, $x\in Z_{q-1} \cap N(u) \subseteq S_q$. Hence, $N(u) \subseteq \bigcup_{p\leq q} S_p$.

    Since $u \in Z_j$ then there exists $w \in (N(u)\cap Z_{j-1})\setminus S_j$ and $w\in \bigcup_{p< j}S_p$, i.e., $w$ satisfies the statement of the lemma, a contradiction. 
\end{proof}

Now, let us generalise Lemmas~\ref{lem:subgraph} and~\ref{lem:parsimonious} to monotone strategies.


\begin{lemma}\label{lem:MonotoneSubgraph}
For any non-empty connected subgraph $H$ of a graph $G=(V,E)$, $mh(H)\leq mh(G)$. More precisely, if there exists a monotone winning hunter strategy ${\cal S}=(S_1,\ldots, S_\ell)$ in $G$, then there exists a monotone winning hunter strategy ${\cal S}'$ in $H$ using at most $\max_{1\leq i\leq \ell} |S_i\cap V(H)|$ hunters.
\end{lemma} 
\begin{proof}
Let $\mathcal{S}=(S_1,\dots,S_\ell)$ be a monotone winning hunter strategy for $G$.

If $|V(H)|=1$, the result clearly holds since $mh(H)=0$. Hence, let us assume that $|V(H)|>1$.
Let $m$ be the minimum integer such that $S_m \cap V(H) \neq \emptyset$ and let $u \in S_m \cap V(H)$ (by Lemma~\ref{lem:independent}, such an integer $m$ exists because $|V(H)|>1$). Let $\mathcal{S}'=(S'_1,\dots,S'_\ell)$ be the hunter strategy, such that for every $1\leq i\leq \ell$,
\[   
S'_i = 
     \begin{cases}
       S_i\cap V(H), &\quad\text{if } S_i\cap V(H)\neq \emptyset\\
       \{u\}, &\quad\text{otherwise}\\
     \end{cases}
\]
 First, we have the following claim.

\begin{claim}\label{C:L1}
    For every $0 \leq i\leq  \ell$ and for any vertex $v\in V(H)$, if $v\in Z_i({\cal S}')$, then $v\in Z_i({\cal S})$.
\end{claim}
\begin{proofclaim}
    Let $\mathcal{R}=(r_0,\dots, r_i=v)$ be a rabbit trajectory in $H$ such that for any $0\leq j < i$, $r_j r_{j+1}\in E(H)$, $r_j\notin S'_{j+1}$ and $r_i=v$ (such a trajectory exists since $v\in Z_i({\cal S}')$). By construction of $\mathcal{S}'$, for any $1\leq j\leq \ell$, $S_j\cap V(H)\subseteq S'_j$. Therefore, $\mathcal{R}$ is also a rabbit trajectory in $G$ with $r_j \notin S_{j+1}$, for all $0 \leq j < i$. Thus, $v\in Z_i({\cal S})$. 
\end{proofclaim}

Let us show that $\mathcal{S}'$ is a monotone winning hunter strategy in $H$.
First, we show that $\mathcal{S}'$ is indeed a winning hunter strategy in $H$. Towards a contradiction, assume that $\mathcal{S}'$ is not a winning strategy in $H$. This implies that $Z_\ell ({\cal S}')\neq \emptyset$. Hence, Claim~\ref{C:L1} implies that $Z_\ell ({\cal S})\neq \emptyset$, contradicting the fact that $\mathcal{S}$ is a winning hunter strategy in $G$.

Thus, $\mathcal{S}'$ is a winning strategy in $H$. Next, we establish that $\mathcal{S}'$ is indeed monotone. Towards a contradiction, let us assume that $\mathcal{S}'$ is non-monotone. Hence, by Lemma~\ref{lem:nonMonotone}, there exist $v \in V(H)$ and $1 \leq q <i \leq \ell$ such that $v \in S'_q$ and $v \in Z_i({\cal S'})\setminus S'_{i+1}$.  By Claim~\ref{C:L1} and because $v\in Z_i({\cal S}')$, $v \in Z_i({\cal S})$.
Since $S_p\cap V(H)\subseteq S'_p$ for any $1\leq p\leq \ell$ and because $v\notin S'_{i+1}$, $v\notin S_{i+1}$. 

If $v=u$, $i+1>m$ (since $u \in S'_p$ for all $1 \leq p \leq m$) and so $v \in S_m$ and in $Z_i({\cal S}) \setminus S_{i+1}$, contradicting the monotonicity of $\cal S$.

Otherwise, $v \neq u$. By construction of $\mathcal{S}'$, $S'_p \setminus \{u\} \subseteq S_p$ for all $1 \leq p \leq \ell$. Hence,  $v \in S_q$ and $v \in Z_i({\cal S}) \setminus S_{i+1}$, contradicting the monotonicity of $\cal S$.

Finally, the fact that $h(\mathcal{S}')\leq \max_{1\leq i\leq \ell} |S_i\cap V(H)|\leq h({\cal S})$ completes the proof.
\end{proof}


\begin{lemma}\label{lem:monotoneParsimonious}
For a graph $G=(V,E)$ and any $k \geq mh(G)$, there exists a parsimonious monotone winning hunter strategy in $G$ using at most $k$ hunters.
\end{lemma}
\begin{proof}
Let $\mathcal{S}=(S_1,\dots, S_\ell)$ be a monotone winning hunter strategy in $G$ using at most $k\geq mh(G)$ hunters such that $\ell$ is minimum. If ${\cal S}$ is parsimonious, we are done. Otherwise, among such strategies, let us consider $\cal S$ that maximizes the first round $1 \leq j \leq \ell$ that makes $\mathcal{S}$ not parsimonious. There are several cases to be considered.
\begin{itemize}
\item
Let $\mathcal{Z}({\cal S})=(Z_0({\cal S}),\dots,Z_{\ell}({\cal S}))$ be the set of contaminated vertices for each round of $\mathcal{S}$. If there exists an integer $\ell' < \ell$ such that $Z_{\ell'}({\cal S}) = \emptyset$, then $\mathcal{S}=(S_1,\dots, S_{\ell'})$ is also a monotone winning hunter strategy in $G$ using at most $k$ hunters, contradicting the minimality of $\ell$. 

Hence, we may assume that $Z_{i}({\cal S}) \neq \emptyset$ for every $0 \leq i < \ell$. 
\item
Let $1 \leq j \leq  \ell$ be the smallest integer such that $S_j \setminus Z_{j-1}({\cal S}) \neq \emptyset$ (if no such integer exists, then ${\cal S}$ is parsimonious and we are done). If $S_j \cap Z_{j-1}({\cal S}) \neq \emptyset$, replace $S_j$ by $S_j \cap Z_{j-1}({\cal S})$. This leads to a winning monotone hunter strategy ${\cal S}'$  (indeed, $Z_h({\cal S})=Z_h({\cal S}')$ for all $1 \leq h \leq \ell$) contradicting the maximality of $j$.

Hence, we may assume that $S_j \cap Z_{j-1}({\cal S}) = \emptyset$. Note that this implies that $j<\ell$ (since otherwise, $\cal S$ would not be winning).
\item 
If any, let $0<i$ be the minimum integer such that $S_{j+i} \cap Z_{j+i-1}({\cal S}) \neq \emptyset$. Let $v \in S_{j+i} \cap Z_{j+i-1}({\cal S})$. Since $v \in Z_{j+i-1}({\cal S})$, by Lemma~\ref{lem:reach2}, $v \in Z_{j-1+i'}({\cal S})$ for every $0 \leq i'< i$.
Then, for every $0 \leq i' <i$, replace $S_{j+i'}$ by $\{v\}$. Let us prove that this leads to a monotone hunter strategy contradicting the maximality of $j$. 

Let ${\cal S'}$ be the strategy obtained by the above modifications. First, note that, for any  $0\leq h< j$, $S_h=S'_h$ and so $Z_h({\cal S})= Z_h({\cal S}')$. By definition, $Z_j({\cal S})=\{x\in V\mid\exists y\in Z_{j-1}({\cal S})\setminus S_j \wedge (xy\in E)\}$ and, since $S_j\cap Z_{j-1}=\emptyset$, we get $Z_j({\cal S})=\{x\in V\mid\exists y\in Z_{j-1}({\cal S}) \wedge (xy\in E)\}$. On the other hand, $Z_j({\cal S}')=\{x\in V\mid\exists y\in Z_{j-1}({\cal S'})\setminus S'_j \wedge (xy\in E)\} = \{x\in V\mid\exists y\in Z_{j-1}({\cal S})\setminus \{v\} \wedge (xy\in E)\}$ since $Z_{j-1}({\cal S}')=Z_{j-1}({\cal S})$ and $S'_j=\{v\}$. Hence, $Z_j({\cal S}')\subseteq Z_j({\cal S})$. By induction on $j\leq i' \leq \ell$ and using the same arguments, we get that $Z_{i'}({\cal S}') \subseteq Z_{i'}({\cal S})$ for every $j\leq i'\leq \ell$.
Thus, ${\cal S}'$ is a winning hunter strategy in $G$ using at most $k\geq mh(G) $ hunters (because $Z_{\ell}({\cal S}')\subseteq Z_{\ell}({\cal S})=\emptyset$).
It remains to show that ${\cal S}'$ is monotone. 

For purpose of contradiction, let us assume that ${\cal S}'$ is non-monotone. By Lemma~\ref{lem:nonMonotone}, there exists a vertex $x$  and $1  < m \leq \ell$ such that $x\in Z_{m-1}({\cal S}')\setminus S'_m$ and $x \in \bigcup_{h<m}S'_h$. 

If $x \neq v$, then by definition of ${\cal S}'$ (for every $1 \leq r \leq \ell$, either $S'_r=S_r$ or $S'_r=\{v\}$) and because $Z_r({\cal S}')\subseteq Z_r({\cal S})$ for all $1 \leq r \leq \ell$, we get that $x \in \bigcup_{h<m}S_h$ and $x\in Z_{m-1}({\cal S})\setminus S_m$ which contradicts the monotonicity of ${\cal S}$. Hence, let us assume that $x=v$. Recall that we proved that $v\in Z_{j-1}({\cal S})\setminus S_j$. Therefore, by monotonicity of $\cal S$, $v \notin \bigcup_{r<j}S_r$ which implies that $v \notin \bigcup_{r<j} S'_r$. Since $v \in S'_{j+i'}$ for all $0 \leq i' \leq i$, we get that $m>j+i$. This means that $v \in Z_{m-1}({\cal S})\setminus S_m$ (because $Z_r({\cal S}')\subseteq Z_r({\cal S})$ for all $1 \leq r \leq \ell$ and $S'_r=S_r$ for all $r\geq m>j+i$) and $v \in S_{j+i}$, which contradicts the monotonicity of $\cal S$. 

Hence, we may assume that $S_{j+i'} \cap Z_{j+i'-1}({\cal S}) = \emptyset$ for all $0 \leq i'$ such that $j+i' \leq \ell$.

\item Let us recall that $j< \ell$ and that $Z_{j-1}({\cal S})\neq \emptyset$. Thus, let $v\in Z_{j-1}({\cal S})$. 
So, there exists $w\in N(v)\cap Z_{j-2}({\cal S})\setminus S_{j-1}$. Moreover, since $w\in Z_{j-2}({\cal S})\setminus S_{j-1}$ (resp. $v\in Z_{j-1}({\cal S})\setminus S_j$), $w\notin S_r$ (resp. $v\notin S_r$) for any $r\leq j$ (otherwise, it contradicts the monotonicity of ${\cal S}$). 
Let us recall that for any $0\leq i'$ such that $j+i'\leq \ell$, $S_{j+i'}\cap Z_{j+i'-1}({\cal S})= \emptyset$. Therefore, $w\notin S_r$ and $v\notin S_r$ for any $r\leq \ell$. 
By Proposition~\ref{prop:reach}, $v \in Z_{\ell}({\cal S})$, which contradicts that ${\cal S}$ is winning.

\end{itemize}
This completes the proof.
\end{proof}

To conclude this subsection, we give an alternative point of view of Lemma~\ref{lem:reach2}. In particular, we show that, if the hunters follow a monotone parsimonious strategy, after having shot at one vertex, the hunters must always shoot at this vertex until it  cannot be reached by the rabbit anymore.

\begin{lemma}\label{lem:MonotoneProperties}
Let $G=(V,E)$ be a graph and $\mathcal{S}=(S_1,\dots,S_{\ell})$ be a parsimonious monotone winning hunter strategy in $G$.
\begin{itemize}
\item If there exist $1 \leq i < j \leq \ell$ such that $v \in S_i \cap S_j$, then $v \in S_{i+1}$.
\item If there exists an integer $1 \leq i < \ell$ such that $v \notin Z_{i-1}$, then $v \notin S_j$ for every $j \geq i$.
\end{itemize}
\end{lemma}
\begin{proof}
First assume that $v \in S_i \cap S_j$. Since $\cal S$ is parsimonious, it implies that $v \in Z_{i-1} \cap Z_{j-1}$. By Lemma~\ref{lem:reach2} and since $v \in Z_j$, $v \in Z_{i'}$ for all $i' <j$. Since $v \in S_i$ and $v \in Z_{p-1}$ for all $i \leq p \leq j$, by monotonicity of $\cal S$, $v \in S_p$ for all $i \leq p \leq j$.

For the second statement, the fact that $v \notin Z_{i-1}$ and Lemma~\ref{lem:reach2} imply that $v \notin Z_j$ for all $j\geq i$. Since $\cal S$ is parsimonious, $v \notin S_j$ for every $j \geq i$.
\end{proof}

Surprisingly, the above lemma is not a characterisation of monotone strategies. Indeed, consider the path $(a,b,c,d)$ on four vertices. It can be checked that the hunter strategy $(\{a\},\{b,c\},\{b,c\})$ is parsimonious and winning (with respect to $V$) and satisfies the condition of the previous lemma, but this strategy is non-monotone (since $a \in S_1$ and $a \in Z_1 \setminus S_2$). 

\subsection{Monotone Hunter Number and Pathwidth}\label{sec:Monotone&Pathwidth}

In this subsection, we relate the monotone hunter number of a graph to its pathwidth. Our result might be surprising since the pathwidth of a graph $G$ is equivalent to the number of searchers required to (monotonously) capture an arbitrary fast invisible fugitive~\cite{BienstockS91} while, in our case, the invisible rabbit seems much weaker than the fugitive: the rabbit is ``slow'' (it moves only to neighbours) and constrained to move at every round. In this view, we might guess that the monotone hunter number of a graph could be arbitrary smaller than its pathwidth. On the contrary, we show that both parameters differ by at most one.

A {\it path-decomposition} of a graph $G=(V,E)$ is a sequence $P=(X_1,\dots,X_p)$ of subsets of vertices, called {\it bags}, such that (1) $\bigcup_{i\leq p} X_i = V$; (2) for every $uv \in E$, there exists $i\leq p$ with $\{u,v\} \subseteq X_i$; and (3): for every $1 \leq i \leq j \leq q \leq p$, $X_i \cap X_q \subseteq X_j$. The {\it width} $w(P)$ of $P$ is the size of a largest bag of $P$ minus one, i.e., $w(P)=\max_{i \leq p} |X_i|-1$. The {\it pathwidth} $pw(G)$ of $G$ is the minimum width of its path-decompositions. A path-decomposition of $G$ of width $pw(G)$ is said to be {\it optimal}. A path-decomposition is {\it reduced} if no bag is contained in another one. It is well known that any graph admits an optimal reduced path-decomposition.

\begin{theorem}\label{theo:pw}
For any graph $G=(V,E)$, $pw(G) \leq mh(G) \leq pw(G)+1$. 
\end{theorem}
\begin{proof}
First, let $P=(X_1,\dots,X_{\ell})$ be a reduced path-decomposition of $G$ with width $k$. Then, $P$ is a monotone hunter strategy in $G$ using $k+1$ hunters. This directly comes from the well known fact that, for every $1 \leq i < \ell$, $X_i \cap X_{i+1}$ separates $\bigcup_{1 \leq j \leq i} X_j\setminus X_{i+1}$ from $\bigcup_{i < j \leq \ell} X_j\setminus X_{i}$, and so $Z_i\subseteq \bigcup_{i < j \leq \ell} X_j$ for every $1 \leq i \leq \ell$. In particular, $mh(G) \leq pw(G)+1$.

To show the other inequality, let $\mathcal{S} = (S_1,\dots, S_\ell)$ be a parsimonious winning monotone hunter strategy in $G$ using at most $k\geq mh(G)$ hunters (it exists by Lemma~\ref{lem:monotoneParsimonious}).


\begin{claim}\label{claim:neighbourhood}
    For every $v \in V\setminus \bigcup_{1 \leq i \leq \ell} S_i$, there exists $1\leq j\leq \ell$ such that $N(v) \subseteq S_j$. 
\end{claim}
\begin{proofclaim}
    For the purpose of contradiction, let $v\notin \bigcup_{1 \leq i \leq \ell} S_i$ such that, for every $1 \leq i \leq \ell$, there exists $u_i \in N(v) \setminus S_i$. Then, $\mathcal{R} = (r_0 = v, u_1,v, u_3,\dots,u_{2i-1},v,u_{2i+1},v \dots)$ is a winning rabbit trajectory against $\cal S$, contradicting the fact that $\cal S$ is a winning hunter strategy. 
\end{proofclaim}

 Let us build a path-decomposition $\mathcal{P}$ of $G$ as follows. 
Start with $\mathcal{P}_0=\emptyset$ and let $Y_0 = V\setminus \bigcup_{1 \leq i \leq \ell} S_i$ ($Y_0$ is the set of vertices that are never shot by $\cal S$). Assume, by induction, that the sequence $\mathcal{P}_i$ and the set $Y_i$ have been built for some $0\leq i<\ell$. Let us define $\mathcal{P}_{i+1}$ and $Y_{i+1}$ as follows. Let $H_{i+1}=\{u^{i+1}_1,\dots,u^{i+1}_{r_{i+1}}\}=\{v \in Y_i \mid N(v) \subseteq S_{i+1}\}$. Let $\odot$ denote the concatenation of two sequences. Let $\mathcal{P}_{i+1}= \mathcal{P}_{i} \odot (S_{i+1} \cup \{u^{i+1}_1\},\dots,S_{i+1} \cup \{u^{i+1}_{r_{i+1}}\})$ and let $Y_{i+1}=Y_i \setminus H_{i+1}$. Finally, let ${\cal P}=(X_1,\dots,X_r)={\cal P}_{\ell}$.

Note that, by construction, for every $1 \leq i \leq r$, $|X_i| \leq k+1$, and so $w({\cal P}) \leq k$. Let us show that $\cal P$ satisfies the three properties of a path-decomposition.

By construction, $\bigcup_{1 \leq i \leq \ell} S_i \subseteq \bigcup_{1 \leq i \leq r} X_i$. Moreover, by Claim~\ref{claim:neighbourhood}, $Y_0 \subseteq \bigcup_{1 \leq i \leq r} X_i$. Hence, $\bigcup_{1 \leq i \leq r} X_i=V$ and the Property (1) of path-decomposition is satisfied. 

By construction, for every $v \in Y_0$, there exists a unique $1 \leq i \leq r$ such that $v \in X_i$. Now, for any $v \in V \setminus Y_0$, let $1 \leq i \leq j \leq r$ such that $v \in X_i \cap X_j$. Let $1 \leq i' \leq \ell$ (resp., $i'\leq j' \leq \ell$) such that $X_i$ has been built from $S_{i'}$ (resp., $X_j$ has been built from $S_{j'}$). By Lemma~\ref{lem:MonotoneProperties}, $v \in S_{p'}$ for all $i' \leq p' \leq j'$. Hence, by construction, $v \in X_p$ for all $i \leq p \leq j$. Therefore, Property (3) of the path-decomposition is satisfied for every $v \in V$. 

Let $uv \in E$. First, let us assume that $u \in Y_0$. By Claim~\ref{claim:neighbourhood}, $v \notin Y_0$. Let $1 \leq j \leq r$ such that $u \in X_j$, then by construction, $N(u) \subset X_j$ and so $u,v \in X_j$. Second, assume that $u,v \in V \setminus Y_0$. For purpose of contradiction, let us assume that, for every $1 \leq i \leq r$, $|\{u,v\} \cap X_i|\leq 1$. W.l.o.g., $M=\max \{1 \leq j \leq r \mid u \in X_j\} < m = \min \{1 \leq j \leq r \mid v \in X_j\}$ (both integers $m$ and $M$ are well defined since Properties (1) and (3) are satisfied). Let $1 \leq M' \leq \ell$ (resp., $M'\leq m' \leq \ell$) such that $X_M$ has been built from $S_{M'}$ (resp., $X_m$ has been built from $S_{m'}$). By definition of $\cal P$, $u \notin \bigcup_{M'<i\leq \ell}S_i$ and $v \in S_{m'} \setminus \bigcup_{1<i< m'}S_i$. Because $\cal S$ is parsimonious, $v \in Z_{m'-1}({\cal S})$ and so, there exists a $w \in N(v) \cap Z_{m'-2}({\cal S})\setminus S_{m'-1}$. By monotonicity of $\cal S$, $w \notin \bigcup_{1 \leq i <m'} S_i$. By Proposition~\ref{prop:reach}, $u \in Z_{m'-1}({\cal S})$. Since $u \in (Z_{m'-1}({\cal S}) \cap S_{M'}) \setminus S_{m'}$, the vertex $u$ is recontaminated, contradicting the monotonicity of $\cal S$. Therefore, in all cases, for every $uv \in E$, there exists $1 \leq j \leq r$ such that $u,v \in X_j$. Hence, Property (2) of path-decompositions is satisfied. 

Hence, $\cal P$ is a path-decomposition of width at most $k$. In particular, $pw(G) \leq mh(G)$.
\end{proof}

Theorem~\ref{theo:pw} has important consequences. 


\begin{corollary}\label{cor:noPolyAlgo}
Given an $n$-node graph $G$ and $k\in \mathbb{N}$, it is \textsf{NP}-hard to decide whether $mh(G) \leq k$. Moreover, it is \textsf{NP}-hard to approximate $mh(G)$ up to an additive error of $n^{\varepsilon}$, for $0<\varepsilon <1$.
\end{corollary}
\begin{proof}
    This comes from Theorem~\ref{theo:pw} and the fact that it is \textsf{NP}-hard to approximate the pathwidth of a graph up to an additive error of $n^\varepsilon$, for $0<\varepsilon<1$~\cite{BGHK95}.
\end{proof}

%

Moreover, Theorem~\ref{theo:pw} implies that recontamination may help in the \textsc{Hunters and Rabbit} game.

\begin{corollary}\label{cor:costMonotone}
There exists  $\varepsilon>0$ such that, for any $k\in \mathbb{N}$, there exists a tree $T$ with $h(T) \geq k$ and $mh(T)\geq (1+\varepsilon) h(T)$.
\end{corollary}
\begin{proof}
For any $n\in \mathbb{N}$, let $T_n$ be the rooted tree defined as follows: $T_0$ is a single node, and, for any $n>0$, $T_n$ is obtained from three copies of $T_{n-1}$ and a new node $r$ (the root of $T_n$) such that $r$ is made adjacent to each of the three roots of the copies of $T_{n-1}$. We have that $|V(T_n)|=\frac{3^{n+1}-1}{2}$ and, by the Parsons' Lemma~\cite{parsons1}, $pw(T_n)=n=\Theta(\log_3(V(T)))$. On the other hand, it is shown in~\cite{gruslys2015catching} that, for any tree $T$, $h(T)\leq \lceil \frac{\log_2 |V(T)|}{2} \rceil$. The result follows for $(1+\varepsilon)=2\frac{\ln(2)}{\ln(3)}$.
\end{proof}

\section{(Monotone) hunter number of some graph classes}
\label{sec:classes}

In this section, we characterise the monotone hunter number of several graph classes such as split graphs, interval graphs, cographs and trees. In particular, in all these cases, our results lead to a polynomial time algorithm to compute the monotone hunter number.

\subsection{Split and interval graphs}\label{sec:split}

A graph $G=(V,E)$ is a \textit{split graph} if $V=C\cup I$ can be partitioned into a set $C$ inducing an inclusion-maximal clique and a set $I$ inducing an independent set. Note that given a split graph $G$, a partition $(C, I)$ of $V(G)$ can be computed in linear time~\cite{hammer}. In what follows, we denote a split graph by $G=(C \cup I,E)$ where $C$ induces an inclusion-maximal clique and $I$ induces an independent set. 
Let us recall the following result on the pathwidth of split graphs:
\begin{lemma}~\cite{DBLP:journals/dam/Gustedt93}\label{lem:splitPw}
Let $G=(C \cup I,E)$ be a split graph. Then, $|C|-1\leq pw(G)\leq |C|$.
\end{lemma}

First, we have the following easy observation. 

\begin{proposition}\label{prop:split}
Let $G=(C \cup I,E)$ be a split graph. Then, $|C|-1 \leq h(G)\leq mh(G)\leq |C|$.
\end{proposition}
\begin{proof}
By Lemma~\ref{lem:subgraph}, $h(G) \geq h(G[C])$, and by Lemma~\ref{lem:minDegree}, $h(G[C]) \geq \delta(G[C])=|C|-1$.  Therefore, $h(G) \geq |C|-1$. Moreover, the hunter strategy that consists in shooting to all the vertices of $C$ twice is clearly a monotone winning hunter strategy in $G$. Hence, $h(G) \leq mh(G) \leq |C|$.
\end{proof}

The following theorem fully characterises the hunter number of split graphs. 

\begin{theorem}\label{theo:split}
Let $G=(C \cup I,E)$ be a split graph. Then, $h(G) = |C|$ if and only if for every two distinct vertices $x,y \in C$, there exists a vertex $z \in N_{I}(x)\cap N_{I}(y)$. Otherwise, $h(G)=|C|-1$.
\end{theorem}
\begin{proof}
First we show that if, for every two distinct vertices $x,y \in C$, there exists a vertex  $z \in I$ such that $xz \in E$ and $yz \in E$, then $h(G) = |C|$. We prove this by showing that there exists a winning rabbit strategy against $|C|-1$ hunters. That is, for any (fixed) hunter strategy $\mathcal{S} = (S_1, \dots,  S_{\ell} )$ such that $|S_i| \leq |C|-1$ for every $i \geq 1$, we design a rabbit trajectory $\mathcal{R} = (r_0, r_1, \dots, r_{\ell-1})$ such that for every $i\geq 0$, $r_i \notin S_{i+1}$. Since $|S_1| \leq |C|-1$, there is at least one vertex, say, $v \in C$, such that $v \notin S_1$. Let $r_0 = v$. Hence the rabbit is safe for the first round (since $r_0 \notin S_{1}$). Now, for $i\geq 0$, let us assume that we have built $(r_0,\dots,r_i)$ such that $r_j \notin S_{j+1}$ for every $0 \leq j \leq i$ and $r_i \in C$. If there is at least one vertex $u \neq r_i$ in $C$ such that $u$ is not shot in round $i+2$ (i.e., $u \notin S_{i+2}$), then let $r_{i+1} = u$. Otherwise, $S_{i+2}=C \setminus\{r_i\}$. Moreover, observe that there is at least one vertex $w \in C$ such that $w \notin S_{i+3}$ (since $|S_{i+3}| < |C|$). We note here that $w$ may be the same vertex as $r_i$. Due to our assumptions, there exists a vertex $z \in I$  such that $wz,r_{i}z \in E$. Let us set $r_{i+1} = z$ and $r_{i+2} = w$. Observe that $r_{i+1} \notin S_{i+2}$ (since $S_{i+2} = C\setminus \{r_i\}$ and $z \in I$), $r_{i+2} \notin S_{i+3}$, and $r_{i+2} \in C$. Therefore, using the above strategy, we can design $\mathcal{R}$ such that it is a winning trajectory against $\mathcal{S}$. Therefore, $h(G) \geq |C|$. Since $h(G) \leq |C|$ (due to Proposition~\ref{prop:split}), we have that $h(G) = |C|$. 

To prove the reverse direction, we show that if there exist two distinct vertices $x,y \in C$ such that $N_{I}(x)\cap N_{I}(y)=\emptyset$ (i.e., there is no $z \in I$ such that $xz \in E$ and $yz \in E$), then $h(G) \leq |C|-1$ (and so $h(G)=|C|-1$ by Proposition~\ref{prop:split}). We prove this by giving a (simple) winning hunter strategy $\mathcal{S}$ using $|C|-1$ hunters. Let $\mathcal{S} = (S_1, S_2,S_3,S_4,S_5)$ where $ S_1=S_2=S_5= C \setminus \{y\}$ and  $S_3 = S_4 = C \setminus \{x\}$. Let $\mathcal{R}=(r_0,\dots,r_4)$ be any rabbit trajectory. If the rabbit is not shot at the first round, i.e., $r_0 \notin S_1$, then either $r_0 = y$ or $r_0 \in I$. Accordingly, we consider both these cases below to show that $\mathcal{S}$ is a winning hunter strategy. 

\smallskip
\noindent\textbf{Case 1.} $\mathbf{r_0 = y:}$ In this case, assume $r_1 \notin S_2$ (otherwise, the rabbit will be shot). Note that $r_1\in N_I(y)$. Since $N_I(x)\cap N_I(y) = \emptyset$, 
$r_2 \in C \setminus \{x\}$.  As $S_3 = C \setminus \{x\}$, $r_2 \in S_3$, and therefore, $\mathcal{S}$ is a winning hunter strategy.

\smallskip
\noindent\textbf{Case 2.} $\mathbf{r_0 \in I:}$ In this case, if $r_1 \notin S_2$, then $r_1=y$. Now, assuming that $r_2 \notin S_3$, the rabbit can either move to $x$ (i.e, $r_2 = x$) or the rabbit can move to $N_I(y)$ (i.e., $r_2 \in N_I(y)$). We have the following two cases accordingly:

\begin{enumerate}
    \item [2.a] $\mathbf{r_2 \in N_I(y)}:$ This case is similar to Case 1. Since $N_I(x)\cap N_I(y) = \emptyset$, $r_3 \in C \setminus \{x\}$.  As $S_4 = C \setminus \{x\}$, $r_3 \in S_4$, and therefore, $\mathcal{S}$ is a winning hunter strategy.
    
    \item [2.b] $\mathbf{r_2 = x:}$ In this case, if $r_3 \notin S_4$, then $r_3 \in N_I(x)$. Therefore, similarly to previous arguments, $r_4 \in C\setminus \{y\}$. Since $S_5 = C\setminus \{y\}$, $\mathcal{S}$ is a winning hunter strategy.
\end{enumerate}
This completes the proof.
\end{proof}

The above characterisation allows us to show that the hunter number and the pathwidth of split graphs coincide.

\begin{corollary}
For any split graph $G=(C\cup I, E)$, $h(G)=pw(G)$.
\end{corollary}
\begin{proof}
If $|C|=1$ and $I=\emptyset$, then $pw(G)=h(G)=0$. If 
$|C|=1$ and $I\neq \emptyset$, then $pw(G)=h(G)=1$. 
Let us now assume that $|C|>1$. 

Since $pw(G),h(G)\in \{|C|-1,|C|\}$ by Lemma~\ref{lem:splitPw} and Proposition~\ref{prop:split}, let us assume first that $h(G)=|C|$. 
By Theorem~\ref{theo:split}, for any two distinct vertices $x,y\in C$, there exist $z\in N_I(x)\cap N_I(y)$. For purpose of contradiction, let us assume that there exists a reduced optimal path-decomposition $P=(X_1,\dots, X_\ell)$ of width $|C|-1$. 
It is well known that there exists $1 \leq i \leq \ell$ with $C \subseteq X_i$. Moreover, since $P$ has width $|C|-1$, $X_i=C$.
Let us prove now that $1<i<\ell$. Let us suppose by contradiction that $i=1$, i.e., $X_1=C$ (the case $i=\ell$ is symmetric). Since $P$ has width $|C|-1$ and is reduced, let $v \in X_1\setminus X_2$ and let $z \in N_I(v)$ ($z$ exists since $|C|>1$ and any two distinct vertices of $C$ have a common neighbour in $I$). Since $v \in X_1=C$ and $v \notin \bigcup_{1<j\leq \ell}X_j$, no bag contains both $v$ an $z$, contradicting the definition of a path-decomposition. Hence, $1<i<\ell$.
Now, let $x \in C\setminus X_{i-1}$ and $y\in C \setminus X_{i+1}$ ($x$ and $y$ exists since $P$ is reduced and $X_i=C$). Let $z \in N_I(x) \cap N_I(y)$.
If $x=y$, no bag contains both $x$ and $z$ (since $x$ appears only in $X_i=C$). If $x\neq y$, there must exist $1 \leq j \leq \ell$ such that $\{y,z\}\subseteq X_j$. Since $y \notin \bigcup_{i<h\leq \ell}X_h$, $j<i$ and since $z \in X_j\setminus X_i$, $z\notin \bigcup_{i<h\leq \ell}X_h$. Finally, since $x \notin \bigcup_{1 \leq h<i}X_h$, there is no bag containing both $x$ and $z$, contradicting the definition of a path-decomposition.

Second, let us assume that $h(G)=|C|-1$. By Theorem~\ref{theo:split}, there exist distinct vertices $x,y\in C$ such that $N_I(x)\cap N_I(y)=\emptyset$. Let us prove that, in that case, $pw(G)=|C|-1$. Let $N_I(x)=\{x_1,\dots, x_m\}$ and $I\setminus N_I(x)=\{y_1,\dots, y_t\}$. Then, $(x_1 \cup (C\setminus \{y\}),\dots,x_m \cup (C\setminus \{y\}),C,y_1 \cup (C\setminus \{x\}),\dots,y_t \cup (C\setminus \{x\}))$ is a path decomposition of $G$ with width $|C|-1$ and, since $pw(G)\geq |C|-1$ by Lemma~\ref{lem:splitPw}, we get that $pw(G)=|C|-1$.
\end{proof}

Next, let us characterise the monotone hunter number of split graphs. We start with the following general lemma.


\begin{lemma}\label{lem:simplicialClique}
Let $G$ be a graph that contains a complete subgraph $C$ such that $N(v)\setminus C \neq \emptyset$ for every $v \in C$. Then, $mh(G)\geq |C|$. 
\end{lemma}
\begin{proof}
By Lemmas~\ref{lem:subgraph},~\ref{lem:minDegree} and Proposition~\ref{prop:relationMonotone&nonM}, $mh(G) \geq h(G) \geq h(C) \geq |C|-1$.
Let $H=G[N[C]]$. We will show that $mh(H)\geq |C|$ and so, the result will follow from Lemma~\ref{lem:MonotoneSubgraph}.

Let us assume by contradiction that $mh(H)=|C|-1$. By Lemma~\ref{lem:monotoneParsimonious}, there exists a parsimonious monotone winning hunter strategy ${\cal S}=(S_1,\dots,S_{\ell})$ in $H$ using $|C|-1$ hunters.

There must be an index $1 \leq i \leq \ell$ such that $|S_i \cap C|=|C|-1$. Otherwise, $(r_0,\dots,r_{\ell})$ where $r_0 \in C \setminus S_1$ and $r_j \in C \setminus (S_{j+1} \cup \{r_{j-1}\})$ for every $1 \leq j \leq \ell$ is a winning rabbit trajectory, contradicting the fact that $\cal S$ is winning. 
Hence, let $i$ be the smallest integer such that $|S_i \cap C|=|C|-1$, let $\{v\}=C \setminus S_i$ and let $w \in N(v) \setminus C$ (which exists by hypothesis). Let us define a rabbit trajectory ${\cal R}=(r_0,\dots,r_{i-1}=v)$  such that $r_j \in C \setminus (S_{j+1} \cup \{r_{j-1}\})$ for every $1 \leq j <i-1$, which is possible since $|S_j \cap C|<|C|-1$ for all $j<i$. Thus $v\in Z_{i-1}\setminus S_i$. Therefore, $v \notin \bigcup_{1 \leq j \leq i}S_j$ since otherwise, $v$ would have been recontaminated.

Let us show now that $w \notin \bigcup_{1 \leq j <i}S_j$. Towards contradiction, let us assume that $w \notin \bigcup_{1 \leq j <i}S_j$. If $w \notin S_{i+1}$, then let $r_{i}=w$: this contradicts the monotonicity of $\cal S$ (since $w \in Z_{i} \cap \bigcup_{1 \leq j <i}S_j) \setminus S_{i+1}$). Hence, $w \in S_{i+1}$ and so there exists $z \in C \setminus (S_{i+1} \cup \{v\})$. In this latter case, let $r_i=z$, contradicting the monotonicity of $\cal S$ (since $z \in Z_{i} \cap S_i) \setminus S_{i+1}$).

Since $w \notin \bigcup_{1 \leq j <i}S_j$, let $j>i$ be the smallest integer such that $w \in S_j$, or,  if $w$ is never shot, let $j>i$ be the smallest integer such that $v \in S_j$ (it must exists otherwise the rabbit may oscillate between $v$ and $w$ without being never shot). In both cases, let $z \in C \setminus (S_j \cup \{v\})$. Thus, by Proposition~\ref{prop:reach}, 
$z\in Z_{j-1}\setminus S_j$, contradicting the monotonicity of $\cal S$.
 %
%
\end{proof}

Recall that a vertex in a graph $G$ is {\it simplicial} if its neighbourhood induces a clique. In particular, in a split graph $G=(C \cup I, E)$, a vertex $v \in C$ is simplicial if and only if $N(v) \setminus C =\emptyset$ (recall that $C$ is supposed to be an inclusion-maximal clique).

\begin{theorem}\label{theo:monotoneSplit}
Let $G=(C \cup I, E)$ be a split graph. Then, $mh(G)=|C|-1$ if and only if there exists a simplicial vertex in $C$. Otherwise, $mh(G)=|C|$.
\end{theorem}
\begin{proof}
Note first that if there is no simplicial vertex in $C$, then by Lemma~\ref{lem:simplicialClique}, $mh(G)\geq |C|$ and so, by Proposition~\ref{prop:split}, $mh(G)=|C|$.
Otherwise, if there exists a simplicial vertex $v\in C$, then $S=(C\setminus v,C\setminus v)$ is a monotone winning hunter strategy in $G$.
\end{proof}

Note that above results imply that there exist split graphs $G$ for which $mh(G)\neq h(G)$, i.e., recontamination helps in split graphs. For instance, let $G$ be the split graph obtained from a clique $C$ and an independent set $I$ by adding a perfect matching between $C$ and $I$. By Theorems~\ref{theo:monotoneSplit} and~\ref{theo:split}, we get that $mh(G)=n$ and $h(G)=n-1$.

To conclude this section, let us show another application of Lemma~\ref{lem:simplicialClique}. Recall that an {\it interval graph} is the intersection graph of a set of intervals in the real line. It is well known that, for any interval graph $G$, $pw(G)=\omega(G)-1$ where $\omega(G)$ is the maximum size of a clique in $G$, and that $G$ admits an optimal path-decomposition where each bag induces a complete graph.

\begin{theorem}
Let $G$ be an interval graph. Then, $h(G)=mh(G)=\omega(G)-1$ if every maximum clique has a simplicial vertex. Otherwise, $mh(G)=\omega(G)$.
\end{theorem}
\begin{proof}
By Theorem~\ref{theo:pw}, $pw(G) \leq mh(G)\leq pw(G)+1=\omega(G)$. Moreover, by Lemma~\ref{lem:minDegree}, $h(G)\geq \omega(G)-1$. If there exists a clique of maximum size that does not contain any simplicial vertex, then by Lemma~\ref{lem:simplicialClique}, $mh(G)=\omega(G)$. Otherwise, let $(X_1,\dots,X_{\ell})$ be an optimal path-decomposition of $G$ such that all bags induce a complete graph. For every $1 \leq i \leq \ell$, if $X_i$ contains a simplicial vertex $v_i$, let $Y_i=\{v_i\}$ and let $Y_i=\emptyset$ otherwise. Then, $(X_1\setminus Y_1,X_1\setminus Y_1,X_2\setminus Y_2,X_2\setminus Y_2,X_3\setminus Y_3,\dots,X_{\ell}\setminus Y_{\ell},X_{\ell}\setminus Y_{\ell})$ is a monotone winning hunter strategy using $\omega(G)-1$ hunters. 
\end{proof}

It follows that $\omega(G)-1 \leq h(G)\leq \omega(G)$. But, the question of deciding $h(G)$ in interval graph when some maximum clique has no simplicial vertex seems more challenging.

\subsection{Cographs} \label{ssec:cograph}

The class of {\it cographs} can be defined recursively as follows~\cite{CPS85}. One vertex is a cograph. Given two cographs $A$ and $B$, their disjoint union $A \cup B$ is a cograph, and their join $A \Join B$ (where all edges between $A$ and $B$ are added) is a cograph. Note that a decomposition of a cograph ({\it i.e.}, a building sequence of unions and joins performed from single vertices) can be computed in linear time~\cite{CPS85}.

\begin{theorem}\label{theo:cographs}
$mh(G)$ can be computed in linear time in the class of cographs.
\end{theorem}
\begin{proof}
Let $A$ and $B$ be two cographs. We prove that: 
\begin{itemize}
\item $mh(A \cup B)=max(mh(A),mh(B))$, and
\item  $mh(A \Join B)=min(mh(A)+|V(B)|, |V(A)|+mh(B))$.
\end{itemize}
The result then follows from the linear time algorithm to compute the recursive decomposition of cographs~\cite{CPS85}.

The first statement is obvious, so let us prove the second one. Let $G=A \Join B$ and let ${\cal S}^A=(S^A_1,\dots,S^A_{\ell})$ and ${\cal S}^B$ be two monotone winning hunter strategies for $A$ and $B$ and using respectively $mh(A)$ and $mh(B)$ hunters. 
Note that ${\cal S}^A\cup V(B)=(S^A_1 \cup V(B),\dots,S^A_{\ell}\cup V(B))$ and ${\cal S}^B\cup V(A)$ are both monotone winning hunter strategies in $G$. Therefore,  $mh(G)\leq min(mh(A)+|V(B)|, |V(A)|+mh(B))$. 

Let ${\cal S}=(S_1,\dots, S_{\ell})$ be a parsimonious monotone winning hunter strategy in $G$ using at most $k \geq mh(G)$ hunters (it exists by Lemma~\ref{lem:monotoneParsimonious}) and such that $\ell$ is minimized among all such strategies. If $\ell =1$, then $k= |S_1|= |V(G)| \geq min(mh(A)+|V(B)|, |V(A)|+mh(B))$. Hence, let us assume that $\ell>1$.
Let ${\cal Z}=(Z_0,\dots Z_\ell)$ be the set of contaminated vertices with respect to ${\cal S}$. Note first that since $\ell$ is minimum, $(S_2,\dots, S_\ell)$ is not a winning hunter strategy, and so $Z_1\neq V$. Let $v\in V\setminus Z_1$.

Let us assume that $v\in V(A)$ (the case $v \in V(B)$ is symmetric). Since $v\notin Z_1$ and $Z_0=V$, $N(v)\subseteq S_1$ . Since $B\subseteq N(v)$, we have that $V(B)\subseteq S_1$. Moreover, $V(B)\subseteq Z_1$. Otherwise, there exists $w\in B$ such that $N(w)\subseteq S_1$ and since $V(A)\subseteq N(w)$, we would have $S_1=V(A)\cup V(B)=V(G)$, a contradiction to the fact that $\ell>1$.

Let us prove by induction on $i$ that, for every $1\leq i < \ell$, $V(B)\subseteq Z_i$ and $V(B)\subseteq S_i$. This statement holds for $i=1$ by the previous paragraph.
By induction, let us assume that $V(B)\subseteq Z_i$ and $V(B)\subseteq S_i$ for some $1 \leq i <\ell-1$. 
Since ${\cal S}$ is monotone, $V(B)\subseteq S_{i+1}$. 
Let us assume that there exists $b\in V(B)$ such that $b\notin Z_{i+1}$. It implies that $V(A)\cap Z_i\subseteq S_{i+1}$. Therefore, $Z_i=(V(A) \cap Z_{i}) \cup (V(B) \cap Z_{i})\subseteq S_{i+1}$, which implies that $Z_{i+1}=\emptyset$,  contradicting the minimality of $\ell$. 
Thus $V(B)\subseteq Z_{i+1}$ and the induction hypothesis holds for $i+1$. In particular, $V(B)\subseteq Z_{\ell-1}$ and so, by monotonicity, $V(B)\subseteq S_{\ell}$. Therefore, $V(B) \subseteq S_i$ for all $1 \leq i \leq \ell$.
Since $V(B) \subseteq S_i$ for all $1 \leq i \leq \ell$, the strategy ${\cal S}\cap V(A)=(S_1 \cap V(A),\dots,S_{\ell} \cap V(A))$ is a monotone winning hunter strategy in $G[A]$ using $k-|V(B)|$ hunters. Hence, $k-|V(B)|\geq mh(A)$ which concludes the proof.
%
\end{proof}

Once again, the case of the hunter number seems more challenging. In particular, the following lemma shows that recontamination may help in cographs.

\begin{lemma}\label{lem:cographs-gap}
For every $k \geq 2$, there exists a cograph $G$ such that $h(G)\geq k$ and $mh(G)\geq \frac{3}{2}h(G)-1$.
\end{lemma}
\begin{proof}
Let $a \geq 1$. Let $A$ and $B$ be two (isomorphic) cographs that consist of the disjoint union of a complete graph with $a$ vertices (denoted by $K_A$ and $K_B$ respectively) and $a$ independent vertices (so $|V(A)|=|V(B)|=2a$). Let $G = A \Join B$. 
Clearly, $h(A)=mh(A)=h(B)=mh(B)=a-1$ and, by the proof of Theorem~\ref{theo:cographs}, $mh(G)=3a-1$. Note also that $h(G)\geq 2a$ by Lemma~\ref{lem:minDegree}. Now, $(A,K_A \cup K_B,K_B,A)$ is a (non-monotone) winning hunter strategy in $G$ using $2a$ hunters and so $h(G) = 2a$.
\end{proof}

\subsection{Trees}\label{ssec:trees}


This section is devoted to showing that the monotone hunter number of trees can be computed in polynomial time. Roughly, we show that a Parsons' like lemma~\cite{parsons1} holds for the monotone hunter number in trees and then the algorithm follows the one for computing the pathwidth of trees in~\cite{appVertexSeparation}.

Let us start with the easy case of paths.

\begin{proposition}\label{prop:monotPath}
Let $P$ be any path with at least $4$ vertices. Then, $1=h(P)<mh(P)=2$.
\end{proposition}
\begin{proof}
The fact that $h(P)=1$ has been proven in~\cite{AbramovskayaFGP16}, and the fact that $mh(P)\leq 2$ is easy.

%
Towards a contradiction, let us assume that there exists a winning monotone hunter strategy  in $P$ using one hunter and let ${\cal S}=(S_1,\dots S_\ell)$ be a shortest such strategy (i.e., minimizing $\ell$). Let $Z=(Z_1,\dots, Z_\ell)$ be the set of contaminated vertices with respect to ${\cal S}$. Let $w \in V(P)$ such that $S_1=\{w\}$. Note that $w \in Z_1$ and so, $\ell >1$. Since $P$ has length at least $4$, there exist $x,y\in V(P)$ such that $x\in N(w)$ and $y\in N(x)\setminus \{w\}$. 
We will prove by induction on $i$ that $S_i=\{w\}$ for all $1 \leq i \leq \ell$. The base case ($i=1$) is already proven. Assume now that for some $1\leq q< \ell$, it holds that $S_j=\{w\}$ for all $1 \leq j \leq q$. Thus, $x,y\notin \bigcup_{1\leq j\leq q} S_j$ and so, by Proposition~\ref{prop:reach}, $w\in Z_q$. 
Hence, by the monotonicity of $\mathcal{S}$, we have that $w\in S_{q+1}$, showing the step of the induction. 
Therefore, $x,y\notin \bigcup_{1\leq j\leq \ell} S_j$ and so, by Proposition~\ref{prop:reach}, $w,x,y\in Z_\ell$, contradicting the fact that $\mathcal{S}$ is a winning strategy in $P$. 
Therefore, $mh(P)\geq 2$.
\end{proof}

We then need the following two technical results.

\begin{proposition}\label{prop:distance-1}
Let $G=(V,E)$ be any connected graph and $H$ be a connected subgraph of $G$. Let $\mathcal{S}=(S_1,\dots, S_\ell)$ be any parsimonious monotone winning hunter strategy in $G$. Moreover, let $1\leq i\leq \ell$ and $x,y \in V(H)$ such that $x\in \bigcup_{j< i}S_j$, $y\in Z_{i-1}$ and minimising the distance between such $x$ and $y$ in $H$. 
If $x,y\notin S_i$, then $xy\in E(H)$.
\end{proposition}
\begin{proof}
Note first that $x\neq y$. Indeed, assuming otherwise would imply that $\mathcal{S}$ is not monotone since $y=x \in (\bigcup_{j< i}S_j \cap Z_{i-1}) \setminus S_{i}$. Hence, we may assume that $x \neq y$. 
Let $P$ be a shortest path from $x$ to $y$ in $H$. Let $a$ be the neighbour of $x$ in $P$. If $a=y$, then $\{x,y\}\in E(G)$, and the claim holds. 
Hence, we may assume that $a\neq y$. By minimality of the distance between $x$ and $y$, $a\notin Z_{i-1}$ and $a\notin \bigcup_{j<i} S_j$.  Let $b \neq x$ be the other neighbour of $a$ in $P$.
We show that $b \notin \bigcup_{j< i}S_j$. If $b\neq y$, then, by minimality of the distance between $x$ and $y$, $b\notin \bigcup_{j<i} S_j$. If $b=y$, since $y\in Z_{i-1}\setminus S_i$, then $y\notin \bigcup_{j< i}S_j$, because otherwise this would contradict the monotonicity of $\mathcal{S}$. Therefore, by Proposition~\ref{prop:reach}, $a\in Z_q$ for every $q< i$. In particular, $a\in Z_{i-1}$, a contradiction.
\end{proof}

\begin{lemma}\label{lem:noRoundEmpty}
Let $G=(V,E)$ be any graph and ${\cal S}=(S_1,\dots, S_\ell)$ be a parsimonious monotone winning hunter strategy in $G$ that uses at most $k$ hunters. Let $H$ be a connected subgraph of $G$ with $|V(H)|>1$. If  $S_i\cap V(H)\neq \emptyset$ and $S_j\cap V(H)\neq \emptyset$ for some $1\leq i<j\leq \ell$, then $S_z\cap V(H)\neq  \emptyset$ for every $i\leq  z \leq j$.
\end{lemma}
\begin{proof}
Let $2 \leq i+1<j \leq \ell$ be such that $V(H)\cap S_i\neq \emptyset$ and $V(H)\cap S_j\neq \emptyset$. Towards a contradiction, let us assume that there exists $i<z<j$ such that $S_z\cap V(H)=\emptyset$.
Let $X =V(H) \cap \bigcup_{q< z} S_q$. Since $V(H)\cap S_i\neq \emptyset$ and $i< z$, we get that $X\neq \emptyset$. Let $Y =V(H) \cap Z_{z-1}$. Since $S_{j}\cap V(H)\neq \emptyset$ and ${\cal S}$ is parsimonious, we have that $Z_{j-1}\cap V(H)\neq \emptyset$. Let $u\in Z_{j-1}\cap S_j\cap V(H)$. By Lemma~\ref{lem:reach2}, $u\in Z_q$ for every $q< j$. In particular, $u\in Z_{z-1}$ and so $Y\neq \emptyset$.
Let $x\in X$ and $y\in Y$ such that the distance bewteen $x$ and $y$ in $H$ is minimum. By Proposition~\ref{prop:distance-1}, $xy\in E$. Thus, since $y\in Z_{z-1}\setminus S_z$, we get that $x\in Z_z$. Therefore, since $\cal S$ is monotone and $x \in \bigcup_{q< z} S_q$, we must have $x \in S_{z+1}$. By Lemma~\ref{lem:MonotoneProperties}, since $x \in \bigcup_{q< z} S_q$ and $x \in S_{z+1}$, then $x \in S_z$. Hence, $V(H)\cap S_z\neq \emptyset$, a contradiction.
\end{proof}


Let $T$ be a tree and $v \in V(T)$. A {\it branch} at $v$ is any connected component of $T-v$. A {\it star} is any tree with at least two vertices and at most one vertex with degree at least two.
Roughly speaking, Parsons' Lemma~\cite{parsons1} states that, for any tree $T$ and $k \in \mathbb{N}$, $pw(T)\geq k+1$ if and only if there exists a vertex $v$ such that at least three branches at $v$ have pathwidth at least $k$. Here, we adapt this lemma in the case of the monotone hunter number of trees.


\begin{lemma}[Parsons' like lemma]\label{lem:parson}
 Let $T=(V,E)$ be any tree.
\begin{itemize}
\item $mh(T)= 0$ if and only if $|V|=1$;
\item $mh(T)= 1$ if and only if $T$ is a star;
\item $mh(T)=2$ if and only if $T$ is not a star and contains a path $P$ such that $T \setminus P$ is a forest of stars and isolated vertices;

\item For every $k\geq 3$,  
$mh(T)\geq k$ if and only if there exists a vertex $v \in V$ such that at least three branches at $v$ have monotone hunter number at least $k-1$.

\end{itemize}
\end{lemma}
\begin{proof}
The first item is trivial. Then, if $T$ is a star, then shooting twice at the vertex of maximum degree is a monotone winning hunter strategy using one hunter, and so, (since $|V(T)|\geq 2$), $mh(T)=1$. If $T$ is not a star (and $|V(T)|>1$), then it contains a path with at least $4$ vertices as a subgraph. By Proposition~\ref{prop:monotPath} and Lemma~\ref{lem:MonotoneSubgraph}, it follows that $mh(T)\geq 2$, which concludes the proof of the second item. If $T$ is not reduced to a star and contains a path $P$ such that $T \setminus P$ is a forest of stars and isolated vertices, it is easy to show that $mh(T)\leq 2$. Otherwise, $T$ contains a vertex $v$ such that at least three components of $T -v$ contain a path with $4$ vertices. The ``if" statement of the fourth item then shows that $mh(T)> 2$ and concludes the proof of the third item. 

Let us prove the fourth item. Let $k\geq 3$.

%
\noindent {\bf Proof of $\Leftarrow$:} Let us first assume that there exists some vertex $v$ and three branches $B_1$, $B_2$ and $B_3$ at $v$, such that $mh(B_1),mh(B_2),mh(B_3)\geq k-1$. We will show that $mh(T)\geq k$. Towards a contradiction, let us assume that $mh(T)<k$. By Lemma~\ref{lem:monotoneParsimonious}, there exists a parsimonious monotone strategy ${\cal S}=(S_1,\dots, S_\ell)$ in $T$ that uses at most $k-1$ hunters.
Let $Z=(Z_1,\dots, Z_\ell)$ be the set of contaminated vertices with respect to ${\cal S}$.

For $j \in \{1,2,3\}$, let $1 \leq i_j \leq \ell$ be the minimum integer such that $V(B_j)\cap Z_{i_j}=\emptyset$.
Note that, by Lemma~\ref{lem:reach2}, $V(B_j)\cap Z_q=\emptyset$ for all $i_j \leq q \leq \ell$, and since $\cal S$ is parsimonious, $V(B_j)\cap S_q=\emptyset$ for all $i_j < q \leq \ell$. Note also that, since $mh(B_j)\geq k-1 \geq 2$, $B_j$ has at least two vertices. Since $Z_{i_j-1} \cap V(B_j) \neq \emptyset$ and $Z_{i_j} \cap V(B_j) = \emptyset$, it implies that $S_{i_j} \cap V(B_j) \neq \emptyset$ (otherwise, let $w \in (Z_{i_j-1} \cap V(B_j))$ and $u \in N(w) \cap V(B_j)$, then $u \in Z_{i_j} \cap V(B_j)$, a contradiction with the definition of $i_j$). W.l.o.g., let us assume that $i_1 \leq i_2 \leq i_3$. 

We will show that there exists a round $j_2$ during which all the $k-1$ hunters will have to shoot on vertices of $B_2$, and that $v\in Z_{j_{2}-1}$, which will lead to a contradiction.

For any $1\leq i\leq 3$, let $j_i$ be an index such that $|S_{j_i}\cap V(B_i)|=k-1$. These indices exist as, otherwise, by Lemma~\ref{lem:MonotoneSubgraph}, $mh(B_i)< k-1$. 
We will first show that $j_2<i_3$, which will be used to prove that $v\in Z_{j_{2}-1}$. Observe that $j_2\leq i_2\leq i_3$. 
Moreover, $j_2\neq i_3$ since $S_{i_3}\cap V(B_3)\neq \emptyset$, $|S_{j_2}\cap V(B_2)|=k-1$ and ${\cal S}$ uses at most $k-1$ hunters.
Hence, $j_2< i_3$. Therefore, by Proposition~\ref{prop:reach} and Lemma~\ref{lem:noRoundEmpty}, $V(B_3)\cup \{v\}\subseteq Z_q$ for all $q\leq j_2$. In particular, $v\in Z_{j_2-1}$. Moreover, since $S_{j_2}\subseteq V(B_2)$, $v \notin S_{j_2}$. Hence, $x\in Z_{j_2}$ where $x$ is the neighbour of $v$ in $B_1$, i.e., $Z_{j_2} \cap V(B_1) \neq \emptyset$. 

Since $Z_{i_1} \cap V(B_1) = \emptyset$, if $i_1< j_2$, then there is a contradiction to Lemma~\ref{lem:reach2}. Otherwise, $j_2<i_1 \leq i_2$ (because $j_2\neq i_1$) and so either $j_1< j_2< i_1$ or $j_2< j_1< i_2$ (because $j_1\neq j_2$, $j_1\leq i_1\leq i_2$ and $j_1\neq i_2$), both contradicting Lemma~\ref{lem:noRoundEmpty}.

\medskip

\noindent {\bf Proof of $\Rightarrow$:} 
Now let us assume that, for every $v \in V$, at most two branches at $v$ have monotone hunter number at least $k-1$. 
%
%
Let us prove that there exists a parsimonious monotone winning hunter strategy ${\cal S}$ in $T$ using at most $k-1$ hunters.

First, let us assume that there exists a path $P=(v_1,\dots, v_{p})$ such that for any connected component $C$ of $T\setminus P$, $mh(C)<k-1$. The following hunter strategy is parsimonious monotone winning and uses at most $k-1$ hunters. The strategy consists of $p$ phases executed sequentially from $i=1$ to $p$.
Phase $i$ consists in shooting $v_i$ at each round, and in using the $k-2$ remaining shots to clear sequentially each connected component of $T\setminus P$ that is adjacent to $v_i$ (this is possible since each of these branches at $v$ has monotone hunter number at most $k-2$). Finally, the last round of Phase $i$ (except for $i=p$) consists in shooting to both $v_{i}$ and $v_{i+1}$ (recall that $k-1 \geq 2$). 

Let us now show that a path $P$, defined as in the previous paragraph, exists.
Let $X$ be the set of vertices $v$ such that exactly two branches at $v$ have monotone hunter number at least $k-1$. First, let us assume that $X \neq \emptyset$ and let us show that it induces a path. Let $x,y \in X$ and let $z$ be any internal vertex of the path between $x$ and $y$. Let $B$ (resp., $B'$) be the branch at $z$ that contains $x$ (resp., that contains $y$). One branch $B_x$ at $x$ with $mh(B_x)\geq k-1$ is a subgraph of $B$ and so, by Lemma~\ref{lem:MonotoneSubgraph}, $mh(B) \geq k-1$. Similarly, $mh(B') \geq k-1$ and so, $z \in X$ and therefore $X$ induces a subtree of $T$. If there exists a node $w$ of degree at least three in $T[X]$, by similar arguments, there are at least three branches at $w$ with monotone hunter number at least $k-1$, a contradiction with the initial hypothesis. Hence, $X$ induces a path $(v_2,\dots,v_{p-1})$. Let $B_1$ be the branch at $v_2$ not containing $v_3$ such that $mh(B_1)=k-1$ (if $v_3$ does not exist, $B_1$ is any branch at $v_2$ with $mh(B_1)\geq k-1$) and let $v_1$ be the neighbour of $v_2$ in $B_1$. Symmetrically, let $B_p$ be the branch at $v_{p-1}$ not containing $v_{p-2}$ such that $mh(B_p)=k-1$ (if $p-1=2$, let  $B_p$ be the branch at $v_2$,
distinct from $B_1$, and with $mh(B_p)\geq k-1$) and let $v_p$ be the neighbour of $v_{p-1}$ in $B_p$. Then, $P=(v_1,v_2,\dots,v_{p-1},v_p)$ satisfies the desired conditions.

Finally, if $X= \emptyset$, let $v_1$ be any vertex of $T$. We build the path $P$ starting from $v_1$ as follows. Let us assume that a path $(v_1,\dots,v_i)$ has already been built for some $i\geq 1$. If there exists a branch $B$ at $v_i$, not containing $v_{i-1}$ (if $i>1$), and with monotone hunter number at least $k-1$ (if any, such a branch must be unique since $X = \emptyset$), then, let $v_{i+1}$ be the neighbour of $v_i$ in $B$. The process ends when no such branch $B$ exists and the obtained path  satisfies the desired conditions.

This completes the proof.
\end{proof}

We design a dynamic programming algorithm to compute the monotone hunter number of a tree $T$. Let us first root $T$ at any vertex $r\in V(T)$. Let $T[u]$ denote the subtree of $T$ induced by $u\in V(T)$ and all the descendent of $u$. For any $x_1,\cdots,x_p \in V(T[u])$, let $T[u,x_1,\ldots, x_p]$ denote the subtree obtained from $T[u]$ after the vertices of $\bigcup_{i \leq p} V(T[x_i])$ have been removed. Finally, for $k\geq 2$, a vertex $x\in V(T)$ is \textit{$k$-critical} if and only if $mh(T[x])=k$ and there exist two children $v_1$ and $v_2$ of $x$ in $T$ such that $mh(T[v_1])=mh(T[v_2])=k$. In the case $k=1$, we will say that a vertex $x$ is \textit{$1$-critical} if and only if $mh(T[x])=1$ and there exists a unique child $v$ of $x$ such that $mh(T[v])=1$. 

\paragraph{Remark:} Let us recall that $mh(T[w])\geq 1$ if $T[w]$ contains at least one edge, i.e., $w$ has at least one child.  Therefore, if a vertex $w$ is $1$-critical in $T[w]$, then $T[w]$ is a star centred in the child $w'$ of $w$ such that $mh(T[w'])=1$, i.e., $w'$ is the only vertex of $T[w]$ that has degree at least $2$. Moreover, if a vertex $w$ is not $1$-critical and $mh(T[w])=1$, then $T[w]$ is star rooted in $w$, i.e., $w$ is the only vertex that may have degree greater than $1$ in $T[w]$ ($w$ may also have degree $1$ when $|V(T[w])|=2$). 

\medskip

The next lemma is used through out the proof of Corollary~\ref{cor:parson} and Lemma~\ref{lem:algTree}.

\begin{lemma}\label{lem:atMost1Branch}
Let $T$ be any tree rooted in a vertex $v\in V(T)$. Let $k=\max_{1\leq i\leq d}\{mh(T[v_i]\}$, where $v_1,\ldots, v_d$ are the children of $v$ in $T$. For any vertex $w\in V(T)$, there is at most $1$ branch at $w$ in $T$ that has monotone hunter number $k+1$. Moreover, $mh(T)\leq k+1$.
\end{lemma}
\begin{proof} 
Note first that, by definition of $k$, there is no branch at $v$ that has monotone hunter number $k+1$.
Hence, let $w$ be any vertex $V(T)\setminus \{v\}$ and let us denote by $x$ the child of $v$ such that $w\in V(T[x])$. By definition of $k$, $mh(T[x])\leq k$. 
For any child $z$ of $w$,  $T[z]$ is a subtree of $T[x]$. Thus, by Lemma~\ref{lem:MonotoneSubgraph}, $mh(T[z])\leq k$ and so there is at most $1$ branch at $w$ that has monotone hunter number at least $k+1$.
To conclude, note that the hunter strategy consisting in applying sequentially all the monotone winning hunter strategies for the branches at $v$, while shooting continuously on $v$, is a monotone winning hunter strategy in $T$ using at most $k+1$ hunters, i.e., $mh(T)\leq k+1$. 
\end{proof}

The upcoming corollary, obtained from Lemmas~\ref{lem:parson} and~\ref{lem:atMost1Branch}, describes how to compute $mh(T)$ of a rooted tree $T$, bottom-up, from its leaves to the root, such that, for any $u\in V(T)$, $mh(T[u])$ is computed from the values $(mh(T[u_i]))_{u_i~child~of~u}$ and from the critical vertices the subtrees $T[u_i]$ contain.

\begin{corollary}\label{cor:parson}
Let $T$ be a rooted tree, $u \in V(T)$ and let $u_1,\dots,u_d$ be the $d$ children of $u$ in $T$. Let us order the children of $u$ such that $k=mh(T[u_1]) \geq mh(T[u_2])\geq \dots mh(T[u_d])$. 
\begin{enumerate}
\item If $d=0$, then $mh(T[u])=0$;

\item If $k=0$ and $d>0$, then $mh(T[u])=1$; 

\item If $k=1$, $d=1$ and the only child of $u$ is not $1$-critical, then $mh(T[u])=1$;

\item If $k=1$ and $d=1$ and the only child of $u$ is $1$-critical, then $mh(T[u])=2$; 

\item If $k=1$ and $d\geq 2$, then $mh(T[u])=2$;

\item If $k>1$ and $mh(T[u_3])=k$, then $mh(T[u])=k+1$;

\item If $k>1$, $mh(T[u_2])=k$ and ($mh(T[u_3])<k$ or $d=2$), and $T[u_1]$ or $T[u_2]$ contains a $k$-critical vertex, then $mh(T[u])=k+1$.

\item If $k>1$, $mh(T[u_2])=k$ and ($mh(T[u_3])<k$ or $d=2$), and neither $T[u_1]$ nor $T[u_2]$ contains a $k$-critical vertex, then $mh(T[u])=k$.

\item If $k>1$, $mh(T[u_1])=k$ and ($mh(T[u_2])<k$ or $d=1$), $T[u_1]$ contains a $k$-critical vertex $x$ and $mh(T[u,x])=k$, then $mh(T[u])=k+1$;

\item If $k>1$, $mh(T[u_1])=k$ and ($mh(T[u_2])<k$ or $d=1$) and $T[u_1]$ contains a $k$-critical vertex $x$ and $mh(T[u,x])<k$, then $mh(T[u])=k$;

\item If $k>1$, $mh(T[u_1])=k$ and ($mh(T[u_2])<k$ or $d=1$) and $T[u_1]$ does not contain any $k$-critical vertex, then $mh(T[u])=k$.
\end{enumerate}
\end{corollary}
\begin{proof}
Each statement can be proved using Lemma~\ref{lem:parson}. 

\begin{enumerate}
    \item If $d=0$, then $T[u]$ is a single vertex and by Lemma~\ref{lem:parson}, $mh(T[u])=0$.

    \item If $k=0$ and $d>0$, then there exists at least one child $w$ of $u$ such that $mh(T[w])=0$ and there is no child $w'$ of $u$ such that $mh(T[w'])>0$. Thus, $T[u]$ contains at least one edge and is a star graph centred in $u$. By Lemma~\ref{lem:parson}, we get that $mh(T[u])=1$.

    \item If $k=1$, $d=1$ and the only child of $u$ is not $1$-critical, i.e., $u_1$ is the only child of $u$ and $mh(T[u_1])=1$. Thus, $T[u_1]$ is a star centred at $u_1$. Therefore, $T[u]$ is also a star centred at $u_1$, and so, by Lemma~\ref{lem:parson}, $mh(T[u])=1$.

    \item If $k=1$, $d=1$ and the only child of $u$ is $1$-critical, i.e., $u_1$ is the only child of $u$, $mh(T[u_1])=1$ and there exists a child $w$ of $u_1$ such that $mh(T[w])=1$. By Lemma~\ref{lem:parson}, $T[w]$ contains at least $1$ edge. Let $v$ be the child of $w$ in $T[w]$. Thus, $T[u]$ contains the path $(v,w,u_1,u)$, and so, is not a star. Moreover, $T[u]\setminus u$ is a forest of stars, and so, by Lemma~\ref{lem:parson}, $mh(T[u])=2$.

    \item If $k=1$ and $d\geq2$, then $u$ has at least two children. Thus,  
    by Lemma~\ref{lem:parson}, $T[u_1]$ contains at least one edge. Let $w$ be a child of $u_1$ in $T[u_1]$. Note that $T[u]$ contains the path $P=(u_2,u,u_1,w)$ as subgraph and that $T[u]\setminus u$ is a forest of stars. Therefore, by Lemma~\ref{lem:parson}, $mh(T[u])=2$.

    \item If $k>1$ and $mh(T[v_3])=k$, then $mh(T[u])\geq k+1$ by Lemma~\ref{lem:parson}. Moreover, by Lemma~\ref{lem:atMost1Branch}, we get that $mh(T[u])\leq k+1$, and so $mh(T[u])=k+1$.

    \item If $k>1$, $mh(T[u_1])=mh(T[u_2])=k$ and ($mh(T[v_3])<k$ or $d=2$) and $T[u_1]$ or $T[u_2]$ contains a $k$-critical vertex, let $x$ be a $k$-critical vertex in $T[u]\setminus \{u\}$. W.l.o.g., let us assume that $x\in V(T[u_1])$. Let us denote by $y$ the parent of $x$ in $T[u]$. Note that $y$ may be equal to $u$. Let  $B_1$ and $B_2$ denote the two branches at $x$ in $T[x]$ such that $mh(B_1)=mh(B_2)=k$. Let $B_y$ denote the branch at $x$ in $T[u]$ that contains $y$. Note that $B_y$ contains $T[u_2]$ as subgraph. Thus, by Lemma~\ref{lem:MonotoneSubgraph}, $mh(B_y)\geq mh(T[u_2])= k$. Therefore, since there are $3$ branches at $x$ with monotone hunter number at least $k$, by Lemma~\ref{lem:parson}, we get that $mh(T[u])\geq k+1$. Finally, by Lemma~\ref{lem:atMost1Branch}, $mh(T[u])\leq k+1$, and so $mh(T[u])=k+1$.

    \item If $k>1$, $mh(T[u_1])=mh(T[u_2])=k$ and ($mh(T[v_3])<k$ or $d=2$) and neither $T[u_1]$ nor $T[u_2]$ contains a $k$-critical vertex, then there do not exist three branches at $u$ having monotone hunter number at least $k$. 
    Note that for any $w\in V(T[u_1])$ (resp., $w\in V(T[u_2])$), $w$ is not $k$-critical, and so, there is at most $1$ branch at $w$ in $T[w]$
    that has monotone hunter number at least $k$. Therefore, since every branch at $w$ in $T[u]$ is also a branch at $w$ in $T[u_1]$, except the one containing the parent of $w$, we get that there are at most $2$ branches at $w$ that has monotone hunter number at least $k$.
    Finally, if $d>2$, for any $2<j\leq d$ and any vertex $w\in V(T[u_j])$, 
    let us denote by $y$ the parent of $w$ in $T[u_j]\cup\{u\}$. Note that for any vertex $z\in N(w)\setminus \{y\}$, $T[z]$ is a subgraph of $T[u_j]$ and so, by Lemma~\ref{lem:MonotoneSubgraph}, $mh(T[z])\leq mh(T[u_j])< k$. 
    Therefore, there is no vertex $w\in V(T[u])$ such that $w$ has $3$ branches that have monotone hunter number at least $k$. Hence, by Lemma~\ref{lem:parson}, $mh(T[u])\leq k$. By Lemma~\ref{lem:MonotoneSubgraph} and because $mh(T[u_1])=k$, $mh(T[u])= k$.

    \item If $k>1$, $mh(T[u_1])=k$ and ($mh(T[u_2])<k$ or $d=1$) and $T[u_1]$ contains a $k$-critical vertex $x$ and $mh(T[u,x])=k$, then, $x$ has $3$ branches with monotone hunter number at least $k$ in $T[u]$. Therefore, by Lemma~\ref{lem:parson}, $mh(T[u])\geq k+1$. Finally, by Lemma~\ref{lem:atMost1Branch}, $mh(T[u])\leq k+1$ and so $mh(T[u])=k+1$.

    \item If $k>1$, $mh(T[u_1])=k$ and ($mh(T[u_2])<k$ or $d=1$) and $T[u_1]$ contains a $k$-critical vertex $x$ and $mh(T[u,x])<k$, then, there do not exist three branches at $u$ with monotone hunter number at least $k$. Let $w$ be any vertex of $T[u]\setminus \{u\}$. Let us assume first that $w$ is not $k$-critical. Thus, there is at most one branch at $w$ in $T[w]$ with monotone hunter number $k$.
    Since there is only one other branch left for $w$ in $T[u]$ (the one containing its parent), there do not exist three branches at $w$ in $T[u]$ with monotone hunter number at least $k$. 
    
    Let us assume now that $w$ is $k$-critical. Thus, $w\in V(T[u_1])$ since $mh(T[u_j])<k$ for any $2\leq j\leq d$.
    Note that $w=x$. Indeed, towards contradiction, let us assume that $w\neq x$. Then, the branch at $w$ containing its parent in $T[u_1]$ contains $T[x]$ as a subgraph and/or the branch at $x$ containing its parent in $T[u_1]$ contains $T[w]$. Thus, in every case, due to Lemma~\ref{lem:MonotoneSubgraph}, there exists a vertex in $T[u_1]$ that has $3$ branches having monotone hunter number $k$, each. Therefore, by Lemma~\ref{lem:parson}, $mh(T[u_1])=k+1$, which is a contradiction to the definition of $k$.
    To sum up, for any vertex $w\in V(T[u])\setminus \{x\}$, there exist at most two branches at $w$ that require monotone hunter number at least $k$.
    Let us recall that, by hypothesis, $mh(T[u,x])<k$. Hence, there are only two branches at $x$ that have monotone hunter number at least $k$ (otherwise, by Lemma~\ref{lem:parson} and Lemma~\ref{lem:MonotoneSubgraph}, $k< mh(T[x])\leq mh(T[u_1])$, a contradiction). Therefore, there is no vertex $w$ in $T[u]$ such that there exist at least $3$ branches at $w$, each having monotone hunter number at least $k$. Thus, by Lemma~\ref{lem:parson}, $mh(T[u])\leq k$ and by Lemma~\ref{lem:MonotoneSubgraph}, $mh(T[u])\geq mh(T[u_1])\geq k$.

    \item If $k>1$, $mh(T[u_1])=k$ and ($mh(T[v_2])<k$ or $d=1$), and $T[u_1]$ does not contain any $k$-critical vertex, then there is no vertex in $T[u]$ that has $3$ branches that require monotone hunter number at least $k$ (as in previous case). Therefore, by Lemma~\ref{lem:parson}, $mh(T[u])\leq k$ and by Lemma~\ref{lem:MonotoneSubgraph}, $mh(T[u])\geq mh(T[u_1])\geq k$, and so, $mh(T[u])=k$.
    
\end{enumerate}
\end{proof}

We need the following technical definition to finally describe our algorithm, which will recursively build a ``label'' for each vertex from the labels of its children. Recall that, for any $u,x_1,\cdots,x_p \in V(T[u])$, the tree $T[u,x_1,\ldots, x_p]$ denotes the subtree obtained from $T[u]$ after the vertices of $\bigcup_{i \leq p} V(T[x_i])$ have been removed.

\begin{definition}\label{definition:labels}
    For any tree $T[u]$, the {\rm label} $\lambda(u,T[u])$ of $u$ is a list of integers $(a_1,\ldots, a_p)$, where $a_1>a_2>\dots>a_p\geq 0$, possibly $a_p$ may be {\it marked} with a star, and there exists a set of vertices $\{\ell_1,\dots,\ell_p\}$, such that:
    \begin{itemize}
    \item $mh(T[u])=a_1$.
    \item For $1\leq i< p$, $mh(T[u,\ell_1,\ldots,\ell_i])=a_{i+1}$ and $\ell_i$ is an $a_i$-critical vertex in $T[u,\ell_1,\ldots, \ell_{i-1}]$. We will say that $\ell_i$ {\it is associated} to $a_i$. 
     
    \item $\ell_p=u$. If $a_p$ is not marked with a star ($*$), then there is no $a_p$-critical vertex in $T[u,\ell_1,\ldots, \ell_{p-1}]$. If $a_p$ is marked (with a star), then $\ell_p$ is an $a_p$-critical vertex. In both cases, $T[u,\ell_p]=T[u,u]$ is the empty tree.
    \end{itemize}
\end{definition}




\paragraph{Examples.} Let us exemplify the above definition. In the following examples, we start from two trees $T_1$ and $T_2$ whose roots have ``almost" the same labels and show that adding one vertex adjacent to their root may lead to two new trees whose roots have ``very different'' labels. In particular, this illustrates the importance of the presence of a star on the last integer of a label.

First, let $T_1$ be a tree rooted in a vertex $u_1$ with the label $\lambda(u_1,T_1[u_1])=(a_1,a_2,a_3)=(3,2,1)$. In Figure~\ref{figure:label-example} we provide an illustration of one such tree. Since $a_1=3$, then $mh(T_1[u_1])=3$ and there exists a vertex $\ell^1_1\in T_1[u_1]$ that is $3$-critical. 
Moreover, $mh(T_1[u_1,\ell^1_1])=a_2=2$ and there exists a vertex $\ell^1_2\in T_1[u_1,\ell^1_1]$ that is $2$-critical. Finally, $mh(T_1[u_1,\ell^1_1,\ell^1_2])=a_3=1$ and since $a_3$ is not marked with a star, there is no $1$-critical vertex in $T_1[u_1,\ell^1_1,\ell^1_2]$. 
By previous remarks, $T_1[u_1,\ell^1_1,\ell^1_2]$ is a star centred in $u_1$ and $\ell^1_3=u_1$ (moreover, the star contains at least $2$ vertices since $mh(T_1[u_1,\ell^1_1,\ell^1_2])>0$).

Let $T'_1$ be the tree obtained from $T_1$ by adding a vertex $u$ (the root of $T'_1$) adjacent to $u_1$. Note that, $\ell^1_1$ (resp., $\ell^1_2$) is still $3$-critical (resp., $2$-critical) in $T'_1[u]$ (resp., in $T'_1[u,\ell^1_1]$). Note also that $T'_1[u,\ell^1_1,\ell^1_2]$ is actually equal to the tree obtained from $T_1[u_1,\ell^1_1,\ell^1_2]$ by making $u$ adjacent to $u_1$. Hence, $T'_1[u,\ell^1_1,\ell^1_2]$ is a star (containing at least $3$ vertices) centred in $u_1$, and so $u$ is $1$-critical in $T'_1[u,\ell^1_1,\ell^1_2]$. Therefore, the label of $u$ in $T'_1$ is $(3,2,1^*)$.

Second, let $T_2$ be a tree rooted at $u_2$ with the label $\lambda(u_2,T_2[u_2])=(a_1,a_2,a_3)=(3,2,1^*)$. Similarly to the previous example, there exist $\ell^2_1$ and $\ell^2_2$ that are the $3$-critical vertex of $T_2[u_1]$ and the $2$-critical vertex of $T_2[u_1,\ell^2_1]$, respectively. Moreover, since $a_3$ is marked with a star, there exists a $1$-critical vertex $\ell^2_3$ in $T_2[u_2,\ell^2_1,\ell^2_2]$. Let us recall that by definition of a label, $\ell^2_3=u_2$, and so $T_2[u_2,\ell^2_1,\ell^2_2]$ is a star with at least $3$ vertices centred in the only child of $u_2$.

Let $T'_2$ be the tree obtained from $T_2$ by adding a vertex $u'$ (the root of $T'_2$) adjacent to $u_2$. Note that $T'_2[u',\ell^2_1,\ell^2_2]$ is actually equal to the tree obtained from $T_2[u_2,\ell^2_1,\ell^2_2]$ by making $u'$ adjacent to $u_2$. Therefore, $T'_2[u',\ell^2_1,\ell^2_2]$ contains a path with $4$ vertices and so $mh(T'_2[u',\ell^2_1,\ell^2_2])>1$. It follows that there exist three branches at $\ell^2_2$ in $T'_2[u',\ell^2_1]$ that have monotone hunter number at least $2$, and so, $mh(T'_2[u',\ell^2_1])\geq 3$. This implies that there exist three branches at $\ell^2_1$ in $T'_2[u']$ that have monotone hunter number at least $3$. Hence, we get that $mh(T'_2[u'])\geq 4$ and $\lambda(u',T'_2[u'])=(4)$.

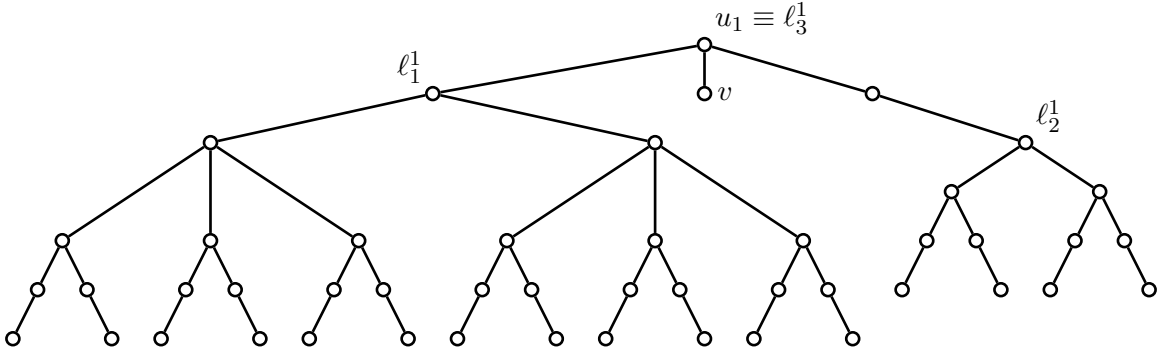
\begin{figure}[!t]
\centering

\begin{tikzpicture}[scale=0.65, inner sep=0.6mm]
    \node[draw, circle, line width=1pt, fill=white](v1) at (0,0)[] {};
    \node[draw, circle, line width=1pt, fill=white](v2) at (2,0)[] {};
    \node[draw, circle, line width=1pt, fill=white](v3) at (3,0)[] {};
    \node[draw, circle, line width=1pt, fill=white](v4) at (5,0)[] {};
    \node[draw, circle, line width=1pt, fill=white](v5) at (6,0)[] {};
    \node[draw, circle, line width=1pt, fill=white](v6) at (8,0)[] {};

    \node[draw, circle, line width=1pt, fill=white](u1) at (-9,0)[] {};
    \node[draw, circle, line width=1pt, fill=white](u2) at (-7,0)[] {};
    \node[draw, circle, line width=1pt, fill=white](u3) at (-6,0)[] {};
    \node[draw, circle, line width=1pt, fill=white](u4) at (-4,0)[] {};
    \node[draw, circle, line width=1pt, fill=white](u5) at (-3,0)[] {};
    \node[draw, circle, line width=1pt, fill=white](u6) at (-1,0)[] {};

    \node[draw, circle, line width=1pt, fill=white](v7) at (0.5,1)[] {};
    \node[draw, circle, line width=1pt, fill=white](v8) at (1.5,1)[] {};
    \node[draw, circle, line width=1pt, fill=white](v9) at (3.5,1)[] {};
    \node[draw, circle, line width=1pt, fill=white](v10) at (4.5,1)[] {};
    \node[draw, circle, line width=1pt, fill=white](v11) at (6.5,1)[] {};
    \node[draw, circle, line width=1pt, fill=white](v12) at (7.5,1)[] {};

    \node[draw, circle, line width=1pt, fill=white](u7) at (-8.5,1)[] {};
    \node[draw, circle, line width=1pt, fill=white](u8) at (-7.5,1)[] {};
    \node[draw, circle, line width=1pt, fill=white](u9) at (-5.5,1)[] {};
    \node[draw, circle, line width=1pt, fill=white](u10) at (-4.5,1)[] {};
    \node[draw, circle, line width=1pt, fill=white](u11) at (-2.5,1)[] {};
    \node[draw, circle, line width=1pt, fill=white](u12) at (-1.5,1)[] {};

    \node[draw, circle, line width=1pt, fill=white](v14) at (9,1)[] {};
    \node[draw, circle, line width=1pt, fill=white](v15) at (11,1)[] {};
    \node[draw, circle, line width=1pt, fill=white](v16) at (12,1)[] {};
    \node[draw, circle, line width=1pt, fill=white](v17) at (14,1)[] {};
    \node[draw, circle, line width=1pt, fill=white](v18) at (1,2)[] {};
    \node[draw, circle, line width=1pt, fill=white](v19) at (4,2)[] {};
    \node[draw, circle, line width=1pt, fill=white](v20) at (7,2)[] {};

    \node[draw, circle, line width=1pt, fill=white](u18) at (-8,2)[] {};
    \node[draw, circle, line width=1pt, fill=white](u19) at (-5,2)[] {};
    \node[draw, circle, line width=1pt, fill=white](u20) at (-2,2)[] {};

    \node[draw, circle, line width=1pt, fill=white](v21) at (9.5,2)[] {};
    \node[draw, circle, line width=1pt, fill=white](v22) at (10.5,2)[] {};
    \node[draw, circle, line width=1pt, fill=white](v23) at (12.5,2)[] {};
    \node[draw, circle, line width=1pt, fill=white](v24) at (13.5,2)[] {};
    \node[draw, circle, line width=1pt, fill=white](v25) at (10,3)[] {};
    \node[draw, circle, line width=1pt, fill=white](v26) at (13,3)[] {};
    \node[draw, circle, line width=1pt, fill=white](v27) at (4,4)[] {};

    \node[draw, circle, line width=1pt, fill=white](u27) at (-5,4)[] {};

    \node[draw, circle, line width=1pt, fill=white](v28) at (11.5,4)[label=above right: $\ell_2^1$] {};
    \node[draw, circle, line width=1pt, fill=white](v29) at (-0.5,5)[label=above left: $\ell_1^1$] {};
    \node[draw, circle, line width=1pt, fill=white](v30) at (5,5)[label=right: $v$ ] {};
    \node[draw, circle, line width=1pt, fill=white](v31) at (8.4,5)[] {};
    \node[draw, circle, line width=1pt, fill=white](v32) at (5,6)[label=above right: $u_1\equiv\ell_3^1$ ] {};

    \draw[-, line width=1pt]  (v1) -- (v7);
    \draw[-, line width=1pt]  (v2) -- (v8);
    \draw[-, line width=1pt]  (v3) -- (v9);
    \draw[-, line width=1pt]  (v4) -- (v10);
    \draw[-, line width=1pt]  (v5) -- (v11);
    \draw[-, line width=1pt]  (v6) -- (v12);

    \draw[-, line width=1pt]  (u1) -- (u7);
    \draw[-, line width=1pt]  (u2) -- (u8);
    \draw[-, line width=1pt]  (u3) -- (u9);
    \draw[-, line width=1pt]  (u4) -- (u10);
    \draw[-, line width=1pt]  (u5) -- (u11);
    \draw[-, line width=1pt]  (u6) -- (u12);

    \draw[-, line width=1pt]  (v7) -- (v18);
    \draw[-, line width=1pt]  (v8) -- (v18);
    \draw[-, line width=1pt]  (v9) -- (v19);
    \draw[-, line width=1pt]  (v10) -- (v19);
    \draw[-, line width=1pt]  (v11) -- (v20);
    \draw[-, line width=1pt]  (v12) -- (v20);

    \draw[-, line width=1pt]  (u7) -- (u18);
    \draw[-, line width=1pt]  (u8) -- (u18);
    \draw[-, line width=1pt]  (u9) -- (u19);
    \draw[-, line width=1pt]  (u10) -- (u19);
    \draw[-, line width=1pt]  (u11) -- (u20);
    \draw[-, line width=1pt]  (u12) -- (u20);

    \draw[-, line width=1pt]  (v14) -- (v21);
    \draw[-, line width=1pt]  (v15) -- (v22);
    \draw[-, line width=1pt]  (v16) -- (v23);
    \draw[-, line width=1pt]  (v17) -- (v24);
    \draw[-, line width=1pt]  (v18) -- (v27);
    \draw[-, line width=1pt]  (v19) -- (v27);
    \draw[-, line width=1pt]  (v20) -- (v27);

    \draw[-, line width=1pt]  (u18) -- (u27);
    \draw[-, line width=1pt]  (u19) -- (u27);
    \draw[-, line width=1pt]  (u20) -- (u27);

    \draw[-, line width=1pt]  (v21) -- (v25);
    \draw[-, line width=1pt]  (v22) -- (v25);
    \draw[-, line width=1pt]  (v23) -- (v26);
    \draw[-, line width=1pt]  (v24) -- (v26);
    \draw[-, line width=1pt]  (v25) -- (v28);
    \draw[-, line width=1pt]  (v26) -- (v28);
    \draw[-, line width=1pt]  (v27) -- (v29);

    \draw[-, line width=1pt]  (u27) -- (v29);

    \draw[-, line width=1pt]  (v28) -- (v31);
    \draw[-, line width=1pt]  (v29) -- (v32);
    \draw[-, line width=1pt]  (v30) -- (v32);
    \draw[-, line width=1pt]  (v31) -- (v32);

\end{tikzpicture}

\caption{An example of the tree $T_1$, described in the examples of Definition~\ref{definition:labels}, such that $\lambda(u_1,T_1[u_1])=(3,2,1)$. Observe that $mh(T_1[\ell_1^1])=3$ and that $\ell_1^1$ is $3$-critical since there are two branches attached to $\ell_1^1$, both of monotone hunter number equal to $3$. Similarly, $mh(T_1[\ell_2^1])=2$ and $\ell_3^1$ is $2$-critical. Finally, the graph $T_1[u_1,\ell_1^1,\ell_1^2]$ is the star on $3$ vertices, centred in $u_1$. Adding any number of leaves to attached to $u_1$ would result in a tree $T'_1$ such that $\lambda(v,T'_1[v])=(3,2,1)$. Adding one leaf $u_2$ attached to $v$ would result in a tree $T_2$ such that $\lambda(u_2,T_2[u_2])=(3,2,1^*)$. Adding a leaf $u'$ attached to $u_2$ would result in a tree $T'_2$ with $\lambda(u',T'_2[u'])=(4)$.}\label{figure:label-example}
\end{figure}

\begin{claim}\label{claim:no(1*,0)}
There exists no tree $T$ rooted in $v\in V(T)$ such that $v$ has label $(1,0)$ in $T[v]$.
\end{claim}
\begin{proof}
Towards contradiction, let us assume that there exists a tree $T$ rooted in $v\in V(T)$ such that $v$ has label $\lambda=(a_1,a_2)=(1,0)$ in $T[v]$. Then, by the definition of the labelling, there exists a vertex $\ell_1$ in $T[v]$ such that $\ell_1$ is $1$-critical (because $a_1=1$) and moreover $\ell_1 \neq v$ because $\lambda \neq (a_1)$. By the definition of a vertex being $1$-critical, we get that $\ell_1$ has a child $y$ such that $mh(T[y])=1$. Since $mh(T[y])=1$, we get that $T[y]$ contains at least one edge and so there exists $x\in N(y)$. Since $T[v]$ is connected, there exists a $(\ell_1,v)$-path $P=(p_1=\ell_1,\ldots p_q=v)$. Thus, $P'=(x,y,p_1=\ell_1,\dots,p_q=v)$ is a path with at least $4$ vertices. Hence, $mh(T[v])>1$ which contradicts that $\lambda(v)=(1,0)$ since $a_1=mh(T[v])$. 
\end{proof}

The following notation is used in Algorithm~\ref{alg:Tree} and in its proof (Lemma~\ref{lem:algTree}). For any sequence $\lambda=(a_1,\ldots, a_p)$ and any integer $a>a_1$, let $a \odot \lambda=(a'_1=a ,a'_2=a_1,\ldots, a'_{p'-1}=a_{p-1}, a'_{p'}=a_p)$, i.e., $\odot$ denotes the concatenation. Moreover, let $pref(\lambda)=(a_1,\cdots,a_{p-1})$ (the prefix of $\lambda$) and $t(\lambda)=a_p$ (the tail of $\lambda$).




\begin{algorithm}\caption{Let $T$ be a tree rooted in a vertex $v$. Let $v_1,\ldots, v_d$ be the $d$ children of $v$ in $T[v]$ and let $\lambda_1,\ldots, \lambda_d$ be their corresponding labels.}\label{alg:Tree}
\begin{algorithmic}[1]

\IF{$d=0$}
 \RETURN $(0)$; \label{Alg:l:cas1}
\ENDIF
\STATE Let $\lambda= ()$ and let  $k=\max_{i\in \{1,\dots,d\}} \max_{q\in \{1,\dots,|\lambda_i|\}} \lambda_i$;
\FOR{$m$ from $0$ to $k$}
\STATE Let $n(m)$ be the number of children with $m$ or $m^*$ in their label;
\IF{$m=0$}
  \IF{$n(m)\geq 1$}
  \STATE $\lambda\gets (1)$; \label{Alg:l:cas2}
  \ELSE 
  \STATE $\lambda\gets (0)$; \label{Alg:l:cas1bis}
  \ENDIF
 \ELSIF{$m=1$}
 \STATE {\it \hspace{-5ex}//Case 3: one child of $v$ has $1$ in its label, its subtree does not contain any $1$-critical vertex and no child of $v$ has $0$ in its label}
\IF{$n(m)=1$, $\forall 1\leq i\leq d$, $m\notin pref(\lambda_i)$ and $t(\lambda_i) \neq m^*$ and $\lambda=(0)$} 
\STATE $\lambda \gets (1^*)$; \label{Alg:l:cas3}  
\STATE {\it \hspace{-7.5ex}//Case 4: $v$ has a unique child, and moreover, this child is a $1$-critical vertex, or}
\STATE {\it \hspace{-7.5ex}//Case 5: $v$ has at least $2$ children with at least one having $1$ in its label}
 \ELSIF{$n(m)\geq 1$}
 \STATE$\lambda \gets (2)$;  \label{Alg:l:cas4-5}
\ENDIF
\ELSIF{$m>1$}
\STATE {\it \hspace{-5ex}// Invariant/intuition: at this step, $\lambda = \lambda(v,T_{m-1})$ where $T_{m-1}$ is the subtree obtained from $T$ by removing the substrees $T[w]$ such that $w \neq v$ and $mh(T[w]) \geq m$.}
\STATE {\it \hspace{-5ex}//Case 6: at least three children of $v$ have $m$ or $m^*$ in their label}
\IF{$n(m)\geq 3$}
\STATE $\lambda \gets (m+1)$; \label{Alg:l:cas6} 
\STATE {\it \hspace{-7.5ex}//Case 7: two children of $v$ have $m$ or $m^*$ in their label and at least one of their subtrees contains a $m$-critical vertex}
\ELSIF{$n(m)=2$ and $\exists 1\leq i\leq d$, $m\in pref(\lambda_i)$ or $t(\lambda_i) = m^*$} 
\STATE $\lambda \gets (m+1)$;  \label{Alg:l:cas7}
\STATE {\it \hspace{-7.5ex}//Case 8: two children of $v$ have $m$ or $m^*$ in their label. Moreover, their subtrees do not contain any $m$-critical vertex}
\ELSIF{$n(m)=2$ and, $\forall 1\leq i\leq d$, $m\notin pref(\lambda_i)$ and $t(\lambda_i) \neq m^*$} 
\STATE $\lambda \gets (m^*)$;  \label{Alg:l:cas8}
\STATE {\it \hspace{-7.5ex}//Case 9: one child of $v$ has $m$ or $m^*$ in its label. Moreover, its subtree contains an $m$-critical vertex and $mh(T_{m-1})= m$}
\ELSIF{$n(m)=1$, $\exists 1\leq i\leq d$, $m\in pref(\lambda_i)$ or $t(\lambda_i) = m^*$, and $\lambda$ contains $m$ or $m^*$}  
\STATE $\lambda \gets (m+1)$; \label{Alg:l:cas9} 
\STATE {\it \hspace{-7.5ex}//Case 10: one child of $v$ has $m$ or $m^*$ in its label. Moreover, its subtree contains an $m$-critical vertex and $mh(T_{m-1})< m$}
\ELSIF{$n(m)=1$, $\exists 1\leq i\leq d$, $m\in pref(\lambda_i)$ or $t(\lambda_i) = m^*$, and $\lambda$ contains neither $m$ nor $m^*$}  
\STATE $\lambda \gets (m)\odot\lambda$;  \label{Alg:l:cas10}
\STATE {\it \hspace{-7.5ex}//Case 11: one child of $v$ has $m$ or $m^*$ in its label. Moreover, its subtree does not contain an $m$-critical vertex}
\ELSIF{$n(m)=1$ and, $\forall 1\leq i\leq d$, $m\notin pref(\lambda_i)$ and $t(\lambda_i) \neq m^*$} 
\STATE $\lambda \gets (m)$; \label{Alg:l:cas11} 
\ENDIF
\ENDIF
\ENDFOR

\RETURN $\lambda$;
\end{algorithmic}
\end{algorithm}

\begin{lemma}\label{lem:algTree}
Algorithm~\ref{alg:Tree} takes a tree $T$ rooted at some vertex $v\in V(T)$ and the label of each child of $v$ as inputs and it returns the label of $v$ in $T[v]$.
\end{lemma}
\begin{proof}
Our goal is to prove that the algorithm returns the label of $v$ in $T[v]$, denoted as $\lambda(v,T[v])=(a^v_1,\ldots, a^v_{p^v})$. Let $d_{T[v]}(v)=d$. Let $\lambda^{alg}$ be the value computed by Algorithm~\ref{alg:Tree}.

Observe first that if $d=0$, by Corollary~\ref{cor:parson}, $mh(T[v])=0$, and so, by definition of a label, $\lambda(v,T[v])=(0)$. Note that Algorithm~\ref{alg:Tree} returns $(0)$ in line~\ref{Alg:l:cas1}. Hence, $\lambda^{alg}=\lambda(v,T[v])=(0)$.

Thus let us assume that $d>0$. Let $v_1,\ldots v_d$ be the $d$ children of $v$ in $T[v]$. For any $1\leq i\leq d$, let $\lambda_i=\lambda(v_i,T[v_i])=(a^i_1,\ldots, a^i_{p_i})$ be the label of $v_i$ in $T[v_i]$. W.l.o.g., let us assume that $a^1_1\geq a^2_1\geq \dots\geq a^d_1$ and let $k=a^1_1=\max_{i\in \{1,\dots,d\}} \max_{q\in \{1,\dots,|\lambda_i|\}} a^i_q$. 


\medskip 

For each $1\leq i\leq d$, let $T^i_k=T[v_i]$. Also, for $0\leq m\leq  k$, if $m=a^i_j$ for any $1\leq j\leq p_i$, let $T^i_{m-1}$ be obtained from $T^i_m$ by removing $T[\ell^i_j]$ where $\ell^i_j$ is the vertex associated to $a^i_j$. Otherwise, let $T^i_{m-1}= T^i_m$, i.e., $T^i_m=T[v_i]\setminus \bigcup_{a^i_j\in \lambda_i, a^i_j>m} T[a^i_j]$. Finally, for any $0\leq m\leq k$, let $T_m$ be the subtree of $T[v]$ induced by $\bigcup_{1\leq i\leq d} V(T^i_m)$. Intuitively, $T_m$ is the subtree obtained from $T[v]$ by removing the substrees $T[w]$ for every vertex $w \neq v$ 
such that $mh(T[w])\geq m+1$. Note that $T_k=T[v]=T$ and that $mh(T_m) \leq m+1$ for every $m\leq k$. Note also that $T_0$ consists of $v$ and, possibly, some of its children.

\medskip

Let us prove by induction on $0\leq m\leq k$ that after the $(m+1)$-th iteration of the loop of the algorithm, the current value of $\lambda$, the variable of Algorithm~\ref{alg:Tree}, denoted by $\lambda^m$, is equal to the label $\lambda(v,T_m)$, i.e., the label of $v$ in $T_m$. If this induction holds, then, when $m=k$, the algorithm returns $\lambda^{alg}=\lambda^k= \lambda(v,T_k)=\lambda(v,T[v])$, which concludes the proof.

Let $n(m)$ be the number of children $w$ of $v$, such that $m$ or $m^*$ is in the label of $w$. 

The base case is for $m=0$.
\begin{itemize}
    \item Let us assume first that $n(0)\geq 1$.  Recall that, by the definition of $T_0$, for any child $w$ of $v$ in $T_0$, $mh(T_0[w])\leq 0$, i.e., $T_0[w]$ only contains $w$. Therefore, $v$ is not $1$-critical in $T_0$. Thus, by Case~$2$ of Corollary~\ref{cor:parson} and by the definition of the labelling,  $mh(T_0)=1$ and $\lambda(v,T_0)=(1)$, which corresponds to $\lambda^0$ (line~\ref{Alg:l:cas2} of Algorithm~\ref{alg:Tree}).

    \item Let us assume now that $n(0)=0$, i.e. $T_0$ only contains  $v$. By Case~$1$ of Corollary~\ref{cor:parson}, $mh(T_0)=0$. Thus, $\lambda(v,T_0)=(0)$. Moreover, since, $n(m)=0$, $\lambda^0=(0)$ (line~\ref{Alg:l:cas1bis} of Algorithm~\ref{alg:Tree}).
\end{itemize} 

Let us assume now that $m=1$. It follows from the case $m=0$ that $\lambda^0=\lambda(v,T_0)$. Before analysing the several subcases for $m=1$, let us recall that $\lambda(v,T_m)$ cannot have label $(1,0)$ by Claim~\ref{claim:no(1*,0)}. 

\begin{itemize}
    \item Assume first that $n(1)=0$. This implies that $T_1=T_0$, and so,  we get that $\lambda^1=\lambda^0=\lambda(v,T_0)=\lambda(v,T_1)$ (in this case, Algorithm~\ref{alg:Tree} does nothing during the $2$-nd iteration of the loop).
    
    \item Assume next that $n(1)=1$, and let $v_i$ be the child of $v$ such that $(1)$ or $(1^*)$ is in the label of $v_i$. 
    
    \begin{itemize}
        \item Let us first assume that we are in Case $3$ (line~\ref{Alg:l:cas3} of Algorithm~\ref{alg:Tree}), i.e. $1$ (not $1^*$) is the last element of the label of $v_i$, and  $\lambda^0=(0)$.
        Recall that since $\lambda^0=(0)$, and by the case $m=0$, $T_0$ is a single vertex. Thus, $v$ has only one child $w$ in $T_1$, and  $mh(T_1[w])= 1$ (because $n(1)=1$), i.e., $T_1$ is a star centred in $w$ (but rooted in $v$). So, by Case~$3$ of Corollary~\ref{cor:parson}, $mh(T_1)=1$. Moreover, $v$ is $1$-critical in $T_1$. Thus, $\lambda(v,T_1)=(1^*)$, which correspond to $\lambda^1$ (line~\ref{Alg:l:cas3} of Algorithm~\ref{alg:Tree}).

        \item Let us then assume that we are in Case $5$ (line~\ref{Alg:l:cas4-5} of Algorithm~\ref{alg:Tree}), i.e. $1$ (not $1^*$) is the last element of the label of $v_i$, and  $\lambda^0=(1)$.
        Observe that if $\lambda^0=(1)$, and by the case $m=0$, then $T_0$ is a star centred (and rooted) in $v$. Thus, $v$ has at least $2$ children in $T_1$ (since $v_i \in V(T_1)\setminus V(T_0)$) but there is one child of $v$, $v_i$, that is not a leaf (i.e., $v_i$ is the root of a subtree with monotone hunter number $1$). So, by Case~$5$ of Corollary~\ref{cor:parson}, $mh(T_1)=2$. By Lemma~\ref{lem:atMost1Branch}, for any vertex $w\in T_1$, there is at most $1$ branch at $w$ that has monotone hunter number at least $2$ (so there are no $2$-critical vertices). Therefore, $\lambda(v,T_1)=(2)$, which corresponds to $\lambda^1$ (line~\ref{Alg:l:cas4-5} of Algorithm~\ref{alg:Tree}).
        

        \item Let us consider the Case $4$ (line~\ref{Alg:l:cas4-5} of Algorithm~\ref{alg:Tree}), i.e. the last element of the label of $v_i$ is $1^*$ and moreover, $v_i$ is the unique child of $v$. By Case~$4$ of Corollary~\ref{cor:parson}, 
         $mh(T_1)=2$. By Lemma~\ref{lem:atMost1Branch}, for any vertex $w\in T_1$, there is at most $1$ branch at $w$ that has monotone hunter number at least $2$ (so there are no $2$-critical vertices). Thus, $\lambda(v,T_1)=(2)$, which corresponds to $\lambda^1$ (line~\ref{Alg:l:cas4-5} of Algorithm~\ref{alg:Tree}).
    \end{itemize}
    
    \item Finally, assume that $n(m)=n(1)\geq 2$. Hence, in $T_1$, $v$ has at least two children which are the roots of subtrees with monotone hunter number $1$. By Case~$5$ of Corollary~\ref{cor:parson}, $mh(T_1)=2$ (recall that by definition of $T_1$, $mh(T_1) \leq 2$). Moreover, by Lemma~\ref{lem:atMost1Branch}, for any vertex $w\in T_1$, there is at most $1$ branch at $w$ that has monotone hunter number at least $2$ (so there are no $2$-critical vertices). Therefore, $\lambda(v,T_1)=(2)$, which corresponds to $\lambda^1$ (line~\ref{Alg:l:cas4-5} of Algorithm~\ref{alg:Tree}).
\end{itemize}

We are now ready to prove the induction step. Let $m\geq 2$ and let us assume that $\lambda(v,T_{m-1})=\lambda^{m-1}$. We will prove that $\lambda(v,T_{m})=\lambda^{m}$. 

\begin{itemize}
    \item {\bf Case $\boldsymbol{6}$.} We are in the case where  $n(m)\geq 3$ (line~\ref{Alg:l:cas6} of Algorithm~\ref{alg:Tree}).
    Since $n(m)\geq 3$, in $T_m$, $v$ has at least $3$ children $v_j$, $v_{j'}$ and $v_{j''}$ such that $mh(T^j_m)= mh(T^{j'}_m)=mh(T^{j''}_m)=m$ (and $mh(T^i_m)\leq m$ for every $1 \leq i \leq d$). Thus, we are in the Case~$6$ of Corollary~\ref{cor:parson}, and so, $mh(T_m)=m+1$. Note also that, by Lemma~\ref{lem:atMost1Branch} and for any vertex $w\in T_m$, there exists at most $1$ branch $B$ at $w$ with $mh(B)\geq m+1$. Therefore, $T_m$ has no $m+1$-critical vertex, and so $\lambda(v,T_m)=(m+1)$. To conclude, line~\ref{Alg:l:cas6} of Algorithm~\ref{alg:Tree} precisely returns $\lambda^m=(m+1)$. Hence, $\lambda^m=\lambda(v,T_m)$.

    \item {\bf Case $\boldsymbol{7}$.} We are in the case where $n(m)=2$ and there exists $1\leq i\leq d$ such that $m \in pref(\lambda_i)$ or $t(\lambda_i)=m^*$, i.e., $T^i_m$ contains an $m$-critical vertex (line~\ref{Alg:l:cas7} of Algorithm~\ref{alg:Tree}).
    Let us denote by $y_i$ be the $m$-critical vertex in $T^i_m$. Since $n(m)=2$, in $T_m$, $v$ has exactly $2$ children $v_j$ and $v_{j'}$ (and $i\in \{j,j'\}$) such that $mh(T^j_m)= mh(T^{j'}_m)=m$ and $mh(T^q_m)<m$ for every other child $v_q$ of $v$ in $T_m$. Since $y_i$ is critical, we are in the Case~$7$ of Corollary~\ref{cor:parson}. Thus, $mh(T_m)=m+1$. Note also that, by Lemma~\ref{lem:atMost1Branch}, for any vertex $w\in T_m$, there exists at most $1$ branch $B$ at $w$ with $mh(B)\geq m+1$. Therefore, $T_m$ has no $m+1$-critical vertex, and so $\lambda(v,T_m)=(m+1)$. To conclude, line~\ref{Alg:l:cas7} of Algorithm~\ref{alg:Tree} precisely returns $\lambda^m=(m+1)$. Hence, $\lambda^m=\lambda(v,T_m)$.

    \item  {\bf Case $\boldsymbol{8}$.} We are in the case where $n(m)=2$ and $m \notin pref(\lambda_i)$ and $t(\lambda_i) \neq m^*$ for all $1\leq i\leq d$, i.e., $T^i_m$ does not contain any $m$-critical vertex (line~\ref{Alg:l:cas8} of Algorithm~\ref{alg:Tree}).
    Since $n(m)=2$, in $T_m$, $v$ has exactly $2$ children $v_j$ and $v_{j'}$ such that $mh(T^j_m)= mh(T^{j'}_m)=m$ and $mh(T^q_m)<m$ for every other child $v_q$ of $v$ in $T_m$. Since $v_j$ and $v_{j'}$ are not $m$-critical, we are in the Case~$8$ of Corollary~\ref{cor:parson}. Thus, $mh(T_m)=m$. Since $n(m)=2$, $v$ is clearly an $m$-critical vertex, and so $\lambda(v,T_m)=(m^*)$. To conclude, line~\ref{Alg:l:cas8} of Algorithm~\ref{alg:Tree} precisely returns $\lambda^m=(m^*)$. Hence, $\lambda^m=\lambda(v,T_m)$.

    \item {\bf Case $\boldsymbol{9}$.} We are in the case where $n(m)=1$, there exists $1\leq i\leq d$ such that $m \in pref(\lambda_i)$, or $m^*=t(\lambda_i)$ (i.e., there is an $m$-critical vertex in $T^i_m$) and $\lambda_{m-1}$ contains $m$ or $m^*$, i.e., $mh(T_{m-1})=m$ (line~\ref{Alg:l:cas9} of Algorithm~\ref{alg:Tree}).   
    Let $y_i$ denote the $m$-critical vertex in $T^i_m$. Since $n(m)=1$, in $T_m$, $v$ has exactly $1$ child $v_j$ (and so, $i=j$) such that $mh(T^j_m)=m$ and $mh(T^q_m)<m$ for every other child $v_q$ of $v$ in $T_m$, and therefore $T_{m-1}= T_{m}[v,y_i]$. Thus, by definition of the labelling and since $m\in \lambda^{m-1}$, $mh(T_{m-1})=mh(T_{m}[v,y_i])=m$. Therefore, we are in the Case~$9$ of Corollary~\ref{cor:parson}. Thus, $mh(T_m)= m+1$. Note also that, by Lemma~\ref{lem:atMost1Branch}, for any vertex $w\in T_m$, there exists at most $1$ branch $B$ at $w$ with $mh(B)\geq m+1$. Therefore, $T_m$ has no $(m+1)$-critical vertex, and so $\lambda(v,T_m)=(m+1)$. To conclude, line~\ref{Alg:l:cas9} of Algorithm~\ref{alg:Tree} precisely returns $\lambda^m=(m+1)$. Hence, $\lambda^m=\lambda(v,T_m)$.

    \item {\bf Case $\boldsymbol{10}$.} We are in the case where $n(m)=1$, there exists $1\leq i\leq d$ such that $m \in pref(\lambda_i)$, or $m^*=t(\lambda_i)$ (i.e., there is an $m$-critical vertex in $T^i_m$), and $m\notin \lambda^{m-1}$ (line~\ref{Alg:l:cas10} of Algorithm~\ref{alg:Tree}). Let us denote by $y_i$ be the $m$-critical vertex in $T^i_m$. Since $n(m)=1$, in $T_m$, $v_i$ is the single child of $v$ such that $mh(T^i_m)=m$ and $mh(T^q_m)<m$ for every other child $v_q$ of $v$ in $T_m$. Therefore, $T_{m-1}[v]= T_{m}[v,y_i]$. Thus, by definition of the labelling and since $m\notin \lambda^{m-1}$, $mh(T_{m-1}[v])=mh(T_{m}[v,y_i])<m$. Therefore, we are in the Case~$10$ of Corollary~\ref{cor:parson}. Thus, $mh(T_m)= m$. Let $\lambda^{m-1}=\lambda(v,T_{m-1})=(a_1,\ldots, a_p)$. Recall that, there exist vertices $\ell_1,\ell_2,\ldots ,\ell_{p-1}$ such that $\ell_h$ is $a_h$-critical in $T_{m-1}[v,\ell_1,\cdots,\ell_{h-1}]$ for every $1 \leq h <p$. Since $T_{m-1}[v]= T_{m}[v,y_i]$, $\ell_h$ is $a_h$-critical in $T_{m}[v,y_i,\ell_1,\cdots,\ell_{h-1}]$ for every $1 \leq h <p$. Moreover, for every $1 \leq h \leq p$, $T_{m-1}[v,\ell_1,\dots, \ell_{h-1}]=T_{m}[v,y_i,\ell_1,\dots, \ell_{h-1}]$, and so $mh(T_{m}[v,y_i,\ell_1,\dots, \ell_{h-1}])= mh(T_{m-1}[v,\ell_1,\dots, \ell_{h-1}])=a_h$. Therefore, $\lambda(v,T_m)=(m,a_1,\dots,a_p)$. To conclude, line~\ref{Alg:l:cas10} of Algorithm~\ref{alg:Tree} precisely returns $\lambda^m=(m)\odot\lambda^{m-1}$. Hence, $\lambda^m=\lambda(v,T_m)$.    

    \item {\bf Case $\boldsymbol{11}$.} Finally, we are in the case where $n(m)=1$ and $m \notin pref(\lambda_i)$ and $t(\lambda_i) \neq m^*$ for all $1\leq i\leq d$, i.e., $T^i_m$ does not contain any $m$-critical vertex (line~\ref{Alg:l:cas11} of Algorithm~\ref{alg:Tree}).
    Since $n(m)=1$, in $T_m$, $v$ has exactly $1$ child $v_i$  such that $mh(T^i_m)=m$ and $mh(T^q_m)<m$ for every other child $v_q$ of $v$ in $T_m$. Since, there is no $m$-critical vertex in $T_m$, we are in the Case~$11$ of Corollary~\ref{cor:parson}. Thus, $mh(T_m)= m$. Note also that $v$ is not $m$-critical since $n(m)=1$. Therefore, $T_m$ has no $m$-critical vertex, and so $\lambda(v,T_m)=(m)$. To conclude, line~\ref{Alg:l:cas11} of Algorithm~\ref{alg:Tree} precisely returns $\lambda^m=(m)$. Hence, $\lambda^m=\lambda(v,T_m)$.
\end{itemize}
\end{proof}

The main result of this section follows:

\begin{theorem}
The monotone hunter number of any tree can be computed in polynomial time.
\end{theorem}

 \section{Monotone hunter number in the red variant in trees}\label{sec:bip}

So far, we have investigated the \textsc{Hunters and Rabbit} game with the additional monotonicity property since monotone strategies are often easier to deal with. Previous works on the \textsc{Hunters and Rabbit} game in bipartite graphs $G=(V_r\cup V_w,E)$ have shown that studying the {\it red variant} of the \textsc{Hunters and Rabbit} game, i.e., when the rabbit is constrained to start in a vertex in $V_r$, could be very fruitful. For instance, recall Lemma~\ref{lem:bipartition} which states that $h(G)=h_{V_r}(G)$ for every bipartite graph $G=(V_r\cup V_w,E)$ and which helped to get many results on the \textsc{Hunters and Rabbit} game~\cite{AbramovskayaFGP16,BOLKEMA2019360,gruslys2015catching}. Therefore, it is interesting to consider the monotonicity constraint when restricted to the red variant of the \textsc{Hunters and Rabbit} game. This section is dedicated to this study in the case of trees. Recall that $mh_{V_r}(G)$ denotes the minimum number of hunters required to win against a rabbit starting at $V_r$ in a bipartite graph $G=(V_r\cup V_w,E)$ and in a monotone way. It can be checked that, in~\cite{gruslys2015catching}, it is actually shown that, for any tree $T$, $mh_{V_r}(T)\leq \lceil \frac{\log_2 |V(T)|}{2} \rceil$ (a monotone strategy with respect to $V_r$ is described in~\cite{gruslys2015catching}). Therefore, the proof of Corollary~\ref{cor:costMonotone} actually shows that:

\begin{corollary}\label{cor:costMonotoneBip}
There exists  $\varepsilon>0$ such that, for any $k\in \mathbb{N}$, there exists a tree $T$ with $mh_{V_r}(T) \geq k$ and $mh(T)\geq (1+\varepsilon) mh_{V_r}(T)$.
\end{corollary}

Therefore, Proposition~\ref{prop:monotPath} and Corollary~\ref{cor:costMonotoneBip} already show that 
there exist graphs $G$ for which $mh_{V_r}(G)<mh(G)$.
The main result of this section is that there exists an infinite family of trees $T$ such that the difference between $mh_{V_r}(T)$ and $h_{V_r}(T)$ is arbitrarily large. 
In particular, this improves the result of Corollary~\ref{cor:costMonotone} since $mh_{V_r}(G) \leq mh(G)$ and $h_{V_r}(G)=h(G)$ for any graph $G$.

More precisely, this section is devoted to proving:

\begin{theorem}\label{theo:gapMonotone&nonM}
For every $i \geq 3$, there exists a tree $T$ such that $mh_{V_r}(T)\geq i$ and $h_{V_r}(T)=h(T)=2$.
\end{theorem}
\begin{proof}
In Section~\ref{ssec:h(T3)=2}, we define a family $(T_{i,2i})_{i \geq 3}$ of trees such that $h_{V_r}(T_{i,2i})=2$ for every $i\geq 3$ (Lemma~\ref{lem:nonMonotTi}). Then, in Section~\ref{ssec:mh(T)>=3}, Lemma~\ref{lem:TiNotMonot} proves that $mh_{V_r}(T_{i,2i}) \geq i$ for every $i\geq 3$.
\end{proof}

In order to prove the Lemmas~\ref{lem:nonMonotTi} and~\ref{lem:TiNotMonot} below, we first need to adapt several technical lemmas and propositions above in the case of the red variant in bipartite graphs. Since the proofs of Proposition~\ref{prop:reachBip} and Lemma~\ref{lem:reach2Bip} (in the red variant) share many similarities with the previous, already proven, versions, we decided to postpone their proofs in the appendix. 
 


Next proposition is an adaptation of Proposition~\ref{prop:reach} in the red variant of the game.

\begin{restatable}{proposition}{propreachBip}\label{prop:reachBip}
Let $\mathcal{S} = (S_1,\ldots, S_\ell)$ be a hunter strategy in a bipartite graph $G=(V_r\cup V_w,E)$ with respect to $V_r$. Let $v\in V_r$ (resp. $v\in V_w)$ and $1\leq i \leq \ell$. If there exists a vertex $u \in N(v)$ and a vertex $x\in N(u)$ (possibly $x=v$) such that $u \notin \bigcup_{j\leq i}S_j$ and $x\notin \bigcup_{j< i}S_j$, then $v \in Z_{2p}$ for every $2p\leq i$  (resp. $v \in Z_{2p+1}$ for every $2p+1\leq i$).
\end{restatable}


Next lemma adapts Lemma~\ref{lem:reach2} in the red variant of the game.

\begin{restatable}{lemma}{lemreachBip}\label{lem:reach2Bip}
Let $G=(V_r\cup V_w,E)$ be a bipartite graph with at least two vertices. Let $\mathcal{S} = (S_1,\ldots, S_\ell)$ be a monotone hunter strategy in $G$ with respect to $V_r$. For any $0 \leq p < i \leq \lceil \ell/2 \rceil$, $Z_{2i} \subseteq Z_{2p}$ and $Z_{2i+1} \subseteq Z_{2p+1}$.
\end{restatable}


Lemma~\ref{lem:nonMonotoneBip} is a direct adaptation, in the red variant of the game, of Lemma~\ref{lem:nonMonotone}. The only difference in their proofs is that, in the case of Lemma~\ref{lem:nonMonotoneBip}, Proposition~\ref{prop:reachBip} must be used instead of Proposition~\ref{prop:reach}. Therefore, we present the proof of Lemma~\ref{lem:nonMonotoneBip} only in the appendix.

\begin{restatable}{lemma}{lemnonMonotoneBip}\label{lem:nonMonotoneBip}
Let $\mathcal{S} = (S_1,\ldots, S_\ell)$ be a non-monotone winning hunter strategy in a bipartite graph $G = (V_r\cup V_w,E)$ with respect to $V_r$. Then, there exist a vertex $v\in V$ and $1\leq i \leq \ell$ such that $v\in Z_{i-1}\setminus S_i$ and $v\in \bigcup_{p< i}S_p$.
\end{restatable}

Lemma~\ref{lem:MonotoneSubgraphBip} is a direct adaptation, in the red variant of the game, of Lemma~\ref{lem:MonotoneSubgraph}. The only difference in their proofs is that, in the case of Lemma~\ref{lem:MonotoneSubgraphBip}, Proposition~\ref{prop:reachBip} and Lemma~\ref{lem:nonMonotoneBip} must be used instead of Proposition~\ref{prop:reach} and of Lemma~\ref{lem:nonMonotone}. Therefore, we present the proof of Lemma~\ref{lem:MonotoneSubgraphBip} only in the appendix.


\begin{restatable}{lemma}{lemMonotoneSubgraphBip}\label{lem:MonotoneSubgraphBip}
For any non-empty connected subgraph $H$ of a bipartite graph $G=(V_r\cup V_w,E)$, $mh_{V_r\cap V(H)}(H)\leq mh_{V_r}(G)$. Moreover, if $|V(H)|> 1$, we get that, if there exists a monotone winning hunter strategy ${\cal S}=(S_1,\ldots, S_\ell)$ in $G$ with respect to $V_r$, then there exists a monotone winning hunter strategy ${\cal S}'$ in $H$ with respect to $V_r\cap V(H)$ using at most $\max_{1\leq i\leq \ell} |S_i\cap V(H)|$ hunters.
\end{restatable} 

Since the proofs of upcoming Lemmas~\ref{lem:monotoneParsimoniousBip} and ~\ref{lem:PseudoMonotoneProperties} (in the red variant) share many similarities with the proofs of Lemmas~\ref{lem:monotoneParsimonious} and ~\ref{lem:MonotoneProperties}, respectively, we decided to postpone their proofs in the appendix.

\begin{restatable}{lemma}{lemmonotoneParsimoniousBip}\label{lem:monotoneParsimoniousBip}
For any bipartite graph $G=(V_r\cup V_w,E)$ and any $k\geq mh_{V_r}(G)$, there exists a parsimonious monotone winning hunter strategy in $G$ with respect to $V_r$ and that uses $k$ hunters.
\end{restatable}


\begin{restatable}{lemma}{lemPseudoMonotoneProperties}\label{lem:PseudoMonotoneProperties}
Let $G=(V_r \cup V_w,E)$ be a bipartite graph and $\mathcal{S}=(S_1,\dots,S_{\ell})$ be a parsimonious monotone winning hunter strategy with respect to $V_r$.
\begin{itemize}
\item If there exist $1 \leq i < j \leq \ell$ such that $v \in S_i \cap S_j$, then $v \in S_{i+2}$.
\item If $v \in V_r$ (resp., $v \in V_w$) and there exists an odd (resp., even) integer $1 \leq i < \ell$ such that $v \notin Z_{i-1}$, then $v \notin S_j$ for every $j \geq i$. 
 %
\end{itemize}
\end{restatable}

\subsection{The family of trees $\boldsymbol{(T_{i,q})_{i\geq 3, q\geq 6}}$: definition and hunter number}\label{ssec:h(T3)=2}

In this section, we will prove that the gap between the hunter number and the monotone hunter number in the red variant of the game 
may be arbitrary large. More precisely, we will design an infinite family $(T_i)_{i\geq 3}$ of trees which exhibits this behaviour.


Let $S_{k,q}$ be the rooted tree obtained from $q\geq 6$ paths of length $k \geq 3$ (with $k$ edges) by identifying an endpoint of each path into a common vertex called the {\it root} of $S_{k,q}$ and denoted by $c$. Equivalently, $S_{k,q}$ can be obtained from a star with root $c$ of degree $q$ by subdividing each edge $k-1$ times. From now on, let $(V_r,V_w)$ be the bipartition of $V(S_{k,q})$ and let us assume that $c\in V_r$.  

\begin{lemma}\label{lem:spider}
For any $k,q\in \mathbb{N}$ such that $k\geq 3$ and $q \geq 6$, it holds that $h(S_{k,q})=mh_{V_r}(S_{k,q})=2$. 
\end{lemma}
\begin{proof}
The fact that $h(S_{k,q})>1$ comes from the characterisation of trees with hunter number one in~\cite{FPrincess}. 
W.l.o.g., let us suppose that the centre $c$ of $S_{k,q}$ is in $V_r$. We now prove that $mh_{V_r}(S_{k,q})\leq 2$ and the result then follows from Lemma~\ref{lem:bipartition} and Proposition~\ref{prop:relationMonotone&nonM}.

The strategy $\cal S$ with respect to $V_r$ and using two hunters proceeds as follows. At every odd round, the first hunter shoots at $c$. The second hunter considers sequentially each path $P=(v_1,\dots,v_k)$ of $S_{k,q} \setminus c$ by iteratively shooting at $v_1,v_2,\dots,v_k$ (starting by shooting $v_1$ at an even round). 

Formally, let $P^1,\dots, P^q$ be the $q$ branches of $c$ in $S_{k,q}$, and let $P^i=(v^i_1,\dots,v^i_k)$ for every $1 \leq i \leq q$ where $v^i_1$ is the neighbour of $c$ in $P^i$. Then, the strategy $\cal S$ equals $(\{c\},\{v^1_1\},\{c,v^1_2\},\{v^1_3\},$ $\{c,v^1_4\},\dots,\{v^1_{k-1}\},\{c,v^1_k\},\{v^2_1\},\{c,v^2_2\},\dots,\{v^i_{j-1}\},\{c,v^i_j\},\dots,\{c,v^q_k\})$ if $k$ is even, and $\cal S$ equals $(\{c\},\{v^1_1\},\{c,v^1_2\},\{v^1_3\},\{c,v^1_4\},\dots,\{c,v^1_{k-1}\},\{v^1_k\},\{c\},\{v^2_1\},\{c,v^2_2\},\dots, \{v^{i-1}_k\},\{c\},\{v^i_1\}$ \\ $,\{c,v^i_2\},\dots,\{v^i_{j-1}\},\{c,v^i_j\},\dots,\{v^q_k\})$ if $k$ is odd.

Clearly, this is a monotone winning hunter strategy in $S_{k,q}$ with respect to $V_r$.
%
%
\end{proof}

Let us denote the strategy described in the previous proof by ${\cal S}_1$ and let $\ell_1$ be the smallest even integer greater or equal to the length of ${\cal S}_1$ (this length equals to $1+qk$ if $k$ is even and to $q(k+1)$ otherwise). 



\paragraph{The construction of the tree $\boldsymbol{T_{i,q}}$.}
For every $i\geq 2$ and $q \geq 6$, let $T_{i,q}$ be the tree recursively built as follows. First, $T_{1,q}=S_{3,q}$. Then, for $i>1$, let us assume that $T_{i-1,q}$ has been defined recursively and that there exists a winning hunter strategy, of length $\ell_{i-1}$, using $2$ hunters in the red variant in $T_{i-1,q}$ (this holds for $i-1=1$ from the previous lemma and it will be proven to hold for every $i \geq 2$ in the next lemma). 
Let $T_{i,q}$ be obtained from $q$ vertex disjoint copies $T_i^1,\dots, T_i^{q}$ of  $T_{i-1,q}$ and from a vertex $c_i$, the root of $T_{i,q}$. Then, for every $1\leq j\leq q$, add a path $P_i^j$ of length $p_i^j$ (defined below) between the root $c_i^j$ of $T_i^j$ and $c_i$ (that is, $c_i$ and $c_i^j$ are at distance $p^i_j$ in $T_{i,q}$). 

The lengths $p_i^j$ are defined recursively as follows. Let $p_i^1=2$ and, for every $1< j\leq q$, let $p_i^j$ be the minimum even integer greater or equal to  
$\ell_{i-1} + \sum_{1\leq k< j} p_i^k$
(it will be shown in the next lemma that $\ell_{i}$ equals the smallest even integer greater or equal to  $q \ell_{i-1} +  \sum_{1\leq j\leq q}  j p_i^j$).

Finally, let us assume that $c_i \in V_r$ and note that, since $p_i^j$ is even, this implies that $c_i,c_i^1,\dots, c_i^{q}$ all belong to $V_r$. 

\begin{lemma}\label{lem:nonMonotTi}
For any $i\in \mathbb{N}^*$ and $q \geq 6$, $h_{V_r}(T_{i,q})=2$.
\end{lemma}
\begin{proof}
The fact that $h_{V_r}(T_{i,q})\geq 2$ follows from Lemmas~\ref{lem:spider} and~\ref{lem:subgraph} and since $T_{i,q}$ contains $S_{3,q}$ as a subgraph.

We prove that $h_{V_r}(T_{i,q})\leq 2$ by induction on $i$. More precisely, we prove that there exists a winning hunter strategy $\mathcal{S}_i=(S_1,\dots,S_{\ell_i})$ for $T_{i,q}$, with respect to $V_r$, using $2$ hunters and such that, for any $j\geq 1$, if the root $c_i$ of $T_{i,q}$ is in $Z_j$, then $c_i\in S_{j+1}$. 
This holds for $i=1$ by Lemma~\ref{lem:spider}. 
Let $i>1$ and let us assume by induction that such a strategy $\mathcal{S}_{i-1}$ has already been defined for $T_{i-1,q}$.

Recall that, for all $1 \leq j \leq q$, $c^j_i$ denotes the root of the copy $T^j_i$ of $T_{i-1,q}$ linked to the root $c_i$ of $T_i$ by a path $P^j_i$ of length at least $\ell_{i-1} + \sum_{1\leq k< j} p_i^k$. Moreover, recall that $c_i \in V_r$. 

Let us define the strategy ${\cal S}_i$ as follows. It proceeds in $q$ phases and ensures that, at every round $h$, if $c_i \in Z_h$, then $c_i \in S_{h+1}$ and that, for every round $h$ arising after the $j^{th}$ phase, $(Z_h \setminus c_i) \cap \bigcup_{k\leq j} (V(T^k_i) \cup V(P^k_i)) = \emptyset$.
This implies that, at the round $\ell_i$, if the rabbit is still alive, it has to be on $c_i$. But during the last phase, we ensured that the rabbit cannot reach $c_i$. Thus, ${\cal S}_i$ is a winning strategy.

Let $1 \leq j \leq q$, let $i^j_0$ be the last round of phase $j$ ($i^0_0=0$) and assume by induction on $j$ that $i^j_0$ is even. Let us moreover assume by induction on $j$ that $(Z_{i^{j-1}_0} \setminus c_i) \cap \bigcup_{k\leq j-1} (V(T^k_i) \cup V(P^k_i)) = \emptyset$. The $j^{th}$ phase proceeds into two sub-phases as follows.

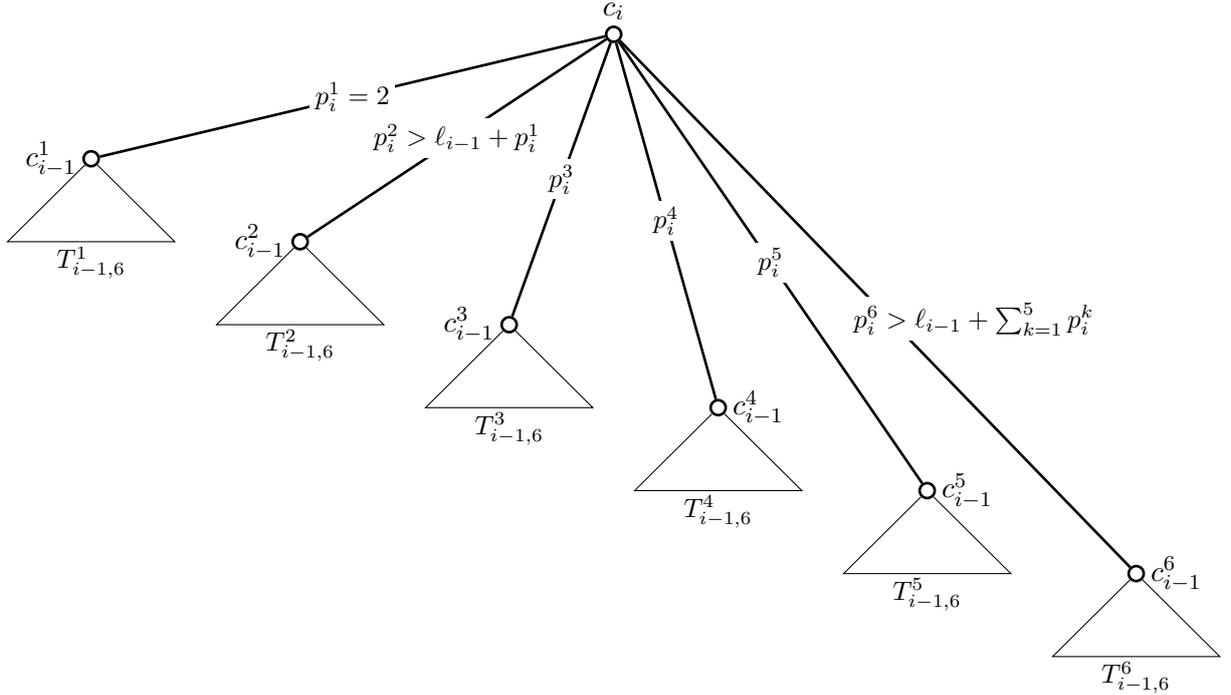
\begin{figure}[!t]
\centering

\begin{tikzpicture}[scale=0.55, inner sep=0.7mm]
    \node[draw, circle, line width=1pt, fill=white](v1) at (0,0)[label=left: $c_{i-1}^1$] {};
    \node[draw, circle, line width=1pt, fill=white](v2) at (5,-2)[label=left: $c_{i-1}^2$] {};
    \node[draw, circle, line width=1pt, fill=white](v3) at (10,-4)[label=left: $c_{i-1}^3$] {};
    \node[draw, circle, line width=1pt, fill=white](v4) at (15,-6)[label=right: $c_{i-1}^4$] {};
    \node[draw, circle, line width=1pt, fill=white](v5) at (20,-8)[label=right: $c_{i-1}^5$] {};
    \node[draw, circle, line width=1pt, fill=white](v6) at (25,-10)[label=right: $c_{i-1}^6$] {};
    \node[draw, circle, line width=1pt, fill=white](v7) at (12.5,3)[label=above: $c_i$] {};

    \draw (v1) --+(-2,-2) --+(2,-2) -- (v1);
    \draw (v2) --+(-2,-2) --+(2,-2) -- (v2);
    \draw (v3) --+(-2,-2) --+(2,-2) -- (v3);
    \draw (v4) --+(-2,-2) --+(2,-2) -- (v4);
    \draw (v5) --+(-2,-2) --+(2,-2) -- (v5);
    \draw (v6) --+(-2,-2) --+(2,-2) -- (v6);
    
    \draw[-, line width=1pt]  (v7) -- node [fill=white] {\small $p_i^1=2$}(v1);
    \draw[-, line width=1pt]  (v7) -- node [fill=white] {\small $p_i^2>\ell_{i-1}+p_i^1$} (v2);
    \draw[-, line width=1pt]  (v7) -- node [fill=white] {\small $p_i^3$}(v3);
    \draw[-, line width=1pt]  (v7) -- node [fill=white] {\small $p_i^4$}(v4);
    \draw[-, line width=1pt]  (v7) -- node [fill=white] {\small $p_i^5$}(v5);
    \draw[-, line width=1pt]  (v7) -- node [xshift=1.3cm, yshift=-0.2cm, fill=white] {\small $p_i^6>\ell_{i-1} + \sum_{k=1}^5 p_i^k$}(v6);


    \node at (0,-2.5) {\small $T_{i-1,6}^1$};
    \node at (5,-4.5) {\small $T_{i-1,6}^2$};
    \node at (10,-6.5) {\small $T_{i-1,6}^3$};
    \node at (15,-8.5) {\small $T_{i-1,6}^4$};
    \node at (20,-10.5) {\small $T_{i-1,6}^5$};
    \node at (25,-12.5) {\small $T_{i-1,6}^6$};

\end{tikzpicture}

\caption{
The graph $T_{i,6}$. The labels on the edges are used to represent their respective lengths. In particular, for every $2\leq j\leq 6$, we have that $p_i^j>\sum_{1\leq k\leq j-1} p_i^k+\ell_{i-1}$, where $p_i^1=2$ and $\ell_{i-1}$ is equal to the number of rounds needed to clear any copy of the $T_{i-1,6}$ graph.}\label{figure:T_i,q}
\end{figure}

Very informally, in the first sub-phase, the hunters ``push'' the rabbit toward the subtrees $T^q_i$, then $T^{q-1}_i$ until the subtree $T^j_i$. Then, in the second sub-phase, the two hunters clear the subtree $T^j_i$ (without the rabbit being able to leave $T^j_i$ if it was there). The lengths of the paths linking the roots of the subtrees to $c_i$ (illustrated in Figure~\ref{figure:T_i,q}) guarantee that the rabbit cannot reach $c_i$ before $T^j_i$ has been cleared. 

Formally, during the first sub-phase, the first hunter shoots at $c_i$ at every odd round. Hence, the rabbit cannot leave the component of $T_{i,q} \setminus c_i$ that it was occupying at the end of the $(j-1)^{th}$ phase.
During the same first sub-phase, the second hunter sequentially shoots the vertices of $P^q_i, P^{q-1}_i,\dots,P^j_i$ in this order and from the neighbours of $c_i$ to the vertices $c^q_i,c^{q-1}_i,\dots,c^j_i$. More precisely, for every $j \leq k \leq q$, let $P^k_i=(v^k_0=c_i,v^k_1,\dots,v^k_{p^k_i}=c^k_i)$. The second hunter starts at round $i^j_0+2$ by shooting $v^q_1$ and then sequentially shoots $v^q_2,v^q_3,\dots,v^q_{p^q_i}=c^q_i,v^{q-1}_1,v^{q-1}_2,\dots,v^{q-1}_{p^{q-1}_i}=c^{q-1}_i,v^{q-2}_1,\dots,c^j_i$. Note that, after the round when the second hunter shoots at $c^q_i$, if the rabbit was occupying a vertex in $T^q_i \cup P^q_i$ at the beginning of the $j^{th}$ phase, then it must occupy a vertex at distance at least $p^q_i$ from $c_i$. 
Similarly, after the round when the second hunter shoots at $c^{q-1}_i$, if the rabbit was occupying a vertex in $T^{q-1}_i \cup P^{q-1}_i$ at the beginning of the $j^{th}$ phase, then it must occupy a vertex at distance at least $p^{q-1}_i$ from $c_i$. Moreover, if the rabbit was occupying a vertex in $T^{q}_i \cup P^{q}_i$ at the beginning of the $j^{th}$ phase, then it must occupy a vertex at distance at least $p^{q}_i-p^{q-1}_i$ from $c_i$ (since there have been $p^{q-1}_i$ rounds between the shots of $c^q_i$ and of $c^{q-1}_i$).
With similar arguments, and by the definition of $p^k_i$ for $j < k \leq q$, after the round when the second hunter shoots at $c^j_i$, the rabbit must be at distance at least $\ell_{i-1}$ from $c_i$ if it was occupying a vertex of $\bigcup_{j<k\leq q} T^{k}_i \cup P^{k}_i$ at the end of the $(j-1)^{th}$ phase. Moreover, the rabbit cannot occupy any vertex in $\bigcup_{k\leq j-1} (V(T^k_i) \cup V(P^k_i))$ since the first hunter is always shooting $c_i$ during the odd rounds. Finally, the rabbit cannot occupy a vertex in $P^j_i$ since the second hunter has just shot sequentially the vertices $v^j_1,v^j_2,\dots,v^j_{p^j_i}=c^j_i$.

Then, the second sub-phase of phase $j$ starts, during which both hunters execute the strategy ${\cal S}_{i-1}$ in the subtree $T^j_i$ (the shot of $c^j_i$ by the second hunter during the first sub-phase, i.e., the last round of the first sub-phase, may be used as the first round of ${\cal S}_{i-1}$). By the induction hypothesis, the strategy ${\cal S}_{i-1}$ ensures that the rabbit cannot reach the root $c^j_i$ of $T^j_i$ without being shot immediately (if $c^j_i$ is in $Z_h$ for some round $h$, then $c^j_i\in S_{h+1}$). Thus, if the rabbit was occupying a vertex of $T^j_i$ at the beginning of the second sub-phase, then the rabbit cannot leave this subtree and it is eventually shot. Otherwise, because ${\cal S}_{i-1}$ has length at most $\ell_{i-1}$, the rabbit cannot reach $c_i$ before the last shots of the hunters in $T^j_i$. Let $i^j_0$ be the last round of the phase. Note that $i^j_0$ is even since $p^h_i$ and $\ell_{h}$ are even for all $1 \leq h \leq q$, and $i^{j-1}_0$ is even by induction. Then, the $j^{th}$ phase ends after the (even) round $i^j_0$ and with the desired property: the rabbit can only occupy a vertex in $c_i\cup \bigcup_{j<k\leq q} (V(T^k_i) \cup V(P^k_i))$, i.e., $(Z_{i^{j}_0} \setminus c_i) \cap \bigcup_{k\leq j} (V(T^k_i) \cup V(P^k_i)) = \emptyset$.


To conclude, note that ${\cal S}_i$ is winning in at most $q \ell_{i-1} +  \sum_{1\leq j\leq q}  j p_i^j$ rounds since each phase $j$ proceeds in $\ell_{i-1}+ \sum_{j \leq k \leq q} p^k_i$ rounds.
\end{proof}

\begin{theorem}\label{theo:subdivision}
For any tree $T$, there exists a subdivision $T'$ of $T$ such that $h(T')\leq 2$.
\end{theorem}
\begin{proof}
Let $q$ be the maximum degree of $T$. Let $r$ be any vertex of $T$, and let $i$ be the eccentricity of $r$ (i.e., the largest distance between $r$ and some vertex of $T$). Then, there exists a subdivision $T'$ of $T$, that is a subgraph of $T_{i,\max\{6,q\}}$ (each vertex of $T$ being ``mapped'' to a vertex of degree at least $3$ of $T_{i,\max\{6,q\}}$ and $r$ being ``mapped'' to the root of $T_{i,\max\{6,q\}}$). By Lemmas~\ref{lem:subgraph},~\ref{lem:bipartition} and~\ref{lem:nonMonotTi}, $h(T')\leq 2$. 
%
\end{proof}

\begin{corollary}
For every $\ell \geq 0$, there exists a tree $T$ and a subdivision $T'$ of $T$ such that $h(T)-h(T') \geq \ell$. 
\end{corollary}

\subsection{Non monotonicity of the red variant in trees}\label{ssec:mh(T)>=3}

 Before proving Lemma~\ref{lem:TiNotMonot}, we need some additional results. Note that the next lemma is the adaptation of Proposition~\ref{prop:distance-1} for the red variant of the game. 

\begin{lemma}\label{lem:PseudoDistance-1}
Let $G=(V_r \cup V_w,E)$ be any bipartite graph and $H$ be a connected subgraph of $G$. Let $\mathcal{S}=(S_1,\dots, S_\ell)$ be any parsimonious monotone winning hunter strategy in $G$ with respect to $V_r$. Let  $1\leq i\leq \ell$ and $x,y \in V(H)$ such that $x\in \bigcup_{j< i}S_j$ and $y\in Z_{i-1}$ and minimising the distance between such $x$ and $y$ in $H$. 
If $x,y\notin S_i$, then $xy\in E(H)$.
\end{lemma}
\begin{proof}
Note first that if $x=y$, then $\mathcal{S}$ is non-monotone since $y=x \in (\bigcup_{j< i}S_j \cap Z_{i-1}) \setminus S_{i}$. Hence, we may assume that $x \neq y$. Let $P$ be a shortest path from $x$ to $y$ in $H$ (it exists since $H$ is connected). Let us assume that $S_i\subseteq V_r$ and so $i$ is odd (the case when $S_i\subseteq V_w$ and $i$ is even is similar). Since $y\in Z_{i-1}$ and $S_i \subseteq Z_{i-1}$ (since ${\cal S}$ is parsimonious), $y\in V_r$.
Let $a$ be the neighbour of $x$ in $P$.

Towards a contradiction, let us assume that $a\neq y$. By the minimality of the distance between $x$ and $y$, $a\notin Z_{i-1}$ and $a\notin \bigcup_{j<i} S_j$. Let $b \neq x$ be the other neighbour of $a$ in $P$.  If $b \neq y$, then by the minimality of the distance between $x$ and $y$, $b\notin Z_{i-1}$ and $b\notin \bigcup_{j<i} S_j$. 
Therefore, by Proposition~\ref{prop:reachBip}, if $a\in V_r$, then $a\in Z_i$ and if $b\in V_r$ then $b\in Z_i$. In both cases, there is a contradiction with the fact that $P$ minimises the distance between $x$ and $y$.


Therefore, we may assume that $b=y$. This implies that $x\in V_r$. Note also that $y\notin S_j$ for all $j\leq i$. Indeed, assuming otherwise would contradict the fact that $\mathcal{S}$ is monotone, since $y\notin S_i$ and $y\in Z_{i-1}$. 
Thus, by Proposition~\ref{prop:reachBip}, and since $a,y\notin \bigcup_{j\leq i-1} S_j$, we have that $x\in Z_{i-1}$. This contradicts the monotonicity of ${\cal S}$ since $x\notin S_i$ and $x\in \bigcup_{j< i}S_j$.
\end{proof}

Before we prove the next lemma, we introduce an extra definition. Let $G=(V,E)$ be any graph and ${\cal S}=(S_1,\dots,S_{\ell})$ be any winning hunter strategy in $G$ with respect to $X\subseteq V$. We say that $W\subseteq V$ is {\it definitively cleaned at the round $i$} if $W\cap Z_{j}({\cal S})=\emptyset$ and $W\cap S_{j+1}=\emptyset$ for every $i \leq j \leq \ell$.


Informally, the following lemma says that if the degree of the root $r$ of a tree $T$ is large enough, compared to the number of hunters, then, when a first branch of $r$ is definitively cleaned according to any monotone hunter strategy, there must be some other branches of $r$ whose vertices have never been shot.

\begin{lemma}\label{lem:other-branches-not-touched}
Let $T=(V_r \cup V_w,E)$ be a tree rooted in some vertex $c \in V_r$ with neighbours $N(c)=\{v_1,\dots,v_d\}$, $d\geq 2k$. For every $1 \leq i \leq d$, let $B_i$ be the branch at $c$ containing $v_i$ and assume that $|V(B_i)| \geq 2$. Let $\mathcal{S}=(S_1\dots,S_\ell)$ be any parsimonious monotone winning hunter strategy in $T$ with respect to $V_r$ using at most $k-1$ hunters. Let $1\leq j\leq \ell$ be the smallest index such that there exists an $1\leq \alpha\leq d$ such that $V(B_{\alpha})$ is definitively cleaned at the round $j$. Then, there exist $1 \leq \beta < \gamma \leq d$, $\alpha \notin \{\beta,\gamma\}$, such that $(\bigcup_{1\leq i\leq j} S_i)\cap V(B_{\beta})=\emptyset$ and $(\bigcup_{1\leq i\leq j} S_i)\cap V(B_{\gamma})=\emptyset$.
\end{lemma}

\begin{proof}
Let $j$ and $\alpha$ be defined as in the statement and, w.l.o.g., let us assume that $\alpha=1$. That is, the branch $B_1$ is definitively cleaned at round $j$, and no other branch has been definitively cleaned before round $j$. 

\begin{claim}\label{claim:definitivelyCleaned}
For any vertex $v\in V(T)$ such that there exist a least $3$ branches at $v$ with at least two vertices, let $q$ be the minimum index such that such a branch $B$ is definitively cleaned at round $q$. Then, $S_q\cap V(B)\neq \emptyset$.
\end{claim}
\begin{proofclaim}. 
By the minimality of $q$, either $Z_{q-1}\cap V(B)\neq \emptyset$ or $S_q\cap V(B)\neq \emptyset$. If $S_q\cap V(B)\neq \emptyset$, then the statement holds. Thus, let us assume that $S_q\cap V(B)= \emptyset$. Then $Z_{q-1}\cap V(B)\neq \emptyset$, as otherwise $B$ would have been definitively cleaned at round prior to $q$. Let $x\in Z_{q-1}\cap V(B)$. Since $x\notin S_q$, $N_{B}(x)\subseteq Z_q$. Since $B$ is connected and contains at least two vertices, we have that $N_{B}(x)\neq \emptyset$. Thus $Z_q\cap V(B)\neq \emptyset$, which contradicts that $B$ is definitively cleaned at round $q$.
\end{proofclaim}

It follows by Claim~\ref{claim:definitivelyCleaned} that $V(B_1)\cap S_j\neq \emptyset$.
Thus, since $|S_j|\leq k-1$, there are at most $k-2$ branches at $c$, other than $B_1$, which can be shot during the round $j$. W.l.o.g., let us assume that $B_{2},\dots, B_{k-1}$ are the branches at $c$ that are also shot during round $j$. That is, $S_j\subseteq  \{c\} \cup \bigcup_{1\leq h< k} V(B_{h})$.

For purpose of contradiction, let us assume that there exists at most one branch, w.l.o.g., say  $B_{k}$, such that $\bigcup_{1\leq i\leq j} S_i\cap V(B_{k})=\emptyset$. Hence, we assume that for every $k<h\leq d$, there exists $j_h< j$ and $x_h \in V(B_h) \cap S_{j_h}$.

For any $k<h\leq d$, let us denote by $j_h^*$ the minimum index such that $B_h$ is definitively cleaned at round $j_h^*$. By Claim~\ref{claim:definitivelyCleaned}, for any $k<h\leq d$, $S_{j_h^*}\cap V(B_h)\neq \emptyset$.
Thus, since $V(B_h)\cap S_j=\emptyset$ for every $k<h\leq d$, we have that $j_h^*\neq j$ for every $k<h\leq d$. In particular, it follows by the minimality of $j$ that  $j_h^*>j$.


Let us prove that, for some $k<h\leq d$, there exists a vertex $y_h\in Z_{j-1}\cap V(B_{h})$. Towards a contradiction, let us assume that for every $k<h\leq d$, we have $Z_{j-1}\cap V(B_h)=\emptyset$. Recall that $S_{j_h^*}\cap V(B_h)\neq \emptyset$, and let $z\in S_{j_h^*}\cap V(B_h)$ (for some $k<h\leq d$). Since ${\cal S}$ is parsimonious, $z\in Z_{j_h^*-1}$.
Thus, there exists a rabbit trajectory $(r_0,\dots, r_{j_h^*-1}=z)$ such that $r_q\notin S_{q+1}$ for every $0\leq q< j_h^*-1$.
Moreover, $r_{j-1}\notin V(B_h)$, since $Z_{j-1}\cap V(B_h)=\emptyset$. Since a rabbit trajectory is a walk, and any walk between a vertex from $V\setminus V(B_h)$ to a vertex of $V(B_h)$ contains $c$, there exists $j-1\leq m<j_h^*-1$ such that $r_m=c$. Since $r_m\notin S_{m+1}$, it follows that $v_1\in Z_{m+1}$, where $v_1$ is the neighbour of $c$ in $B_1$. This contradicts that $B_1$ is definitively cleaned at round $j$. Hence, there exists some vertex $y_h\in Z_{j-1}\cap V(B_{h})$.

Finally, for every $k < h \leq d$, let us choose the vertices $x_h$ and $y_h$, such that $x_h\in V(B_{h})\cap \bigcup_{1 \leq i< j}S_i$ and $y_h\in Z_{j-1}\cap V(B_{h})$ and the distance between $x_h$ and $y_h$ is minimised. Note that $y_h\notin \bigcup_{1 \leq i< j}S_i$, since, otherwise, $\mathcal{S}$ would not be monotone as $y_h\in Z_{j-1}\setminus S_j$. Since $S_j\cap V(B_{h})=\emptyset$ and by Lemma~\ref{lem:PseudoDistance-1}, we obtain that $x_{h} y_h\in E(B_{h})$. Thus, $x_h\in Z_{j}$ for every $k<h\leq d$. Since $\mathcal{S}$ is a monotone strategy, $x_h \in S_{j_h} \cap Z_j$ (with $j_h <j$), that $x_h\in S_{j+1}$ for every $k<h\leq d$. However, $d\geq 2k$ and so $|S_{j+1}|\geq k$, a contradiction. 
\end{proof}




\begin{figure}[!t]
\centering

\begin{tikzpicture}[scale=0.46, inner sep=0.7mm]
    \node[draw, circle, line width=1pt, fill=white](v13) at (-1,2)[] {};
    \node[draw, circle, line width=1pt, fill=white](v14) at (12,2)[] {};
    \node[draw, circle, line width=1pt, fill=white](v15) at (4,2)[] {};
    \node[draw, circle, line width=1pt, fill=white](v16) at (9,2)[] {};
    \node[draw, circle, line width=1pt, fill=white](z1) at (15,2)[] {};
    \node[draw, circle, line width=1pt, fill=white](v17) at (4,4)[label=above right: $\gamma_{2i-s}$] {};

    \node[draw, circle, line width=1pt, fill=white](v22) at (-1,9)[] {};

    \node[draw, circle, line width=1pt, fill=white](u3) at (15,9)[] {};
    \node[draw, circle, line width=1pt, fill=white](z2) at (18,9)[] {};
    
    \node[draw, circle, line width=1pt, fill=white](u2) at (4,9)[] {};
    
    \node[draw, circle, line width=1pt, fill=white](v24) at (12,9)[] {};

    \node[draw, circle, line width=1pt, fill=white](v31) at (-2,14)[] {};
    \node[draw, circle, line width=1pt, fill=white](v32) at (19,14)[] {};
    \node[draw, circle, line width=1pt, fill=white](z3) at (22,14)[] {};
    \node[draw, circle, line width=1pt, fill=white](v33) at (4,11)[label=above right: $\gamma_{2i-1}$] {};

    \node[draw, circle, line width=1pt, fill=white](u1) at (4,14)[] {};
    
    \node[draw, circle, line width=1pt, fill=white](v34) at (16,14)[] {};
    \node[draw, circle, line width=1pt, fill=white](v35) at (8,17)[label=above: $\gamma_{2i}$] {};

    \draw (v13) --+(-2,-2) --+(2,-2) -- (v13);
    \draw[-,line width=0.5pt,black,fill=lightgray] (v14) --+(-1,-1) --+(1,-1) -- (v14);
    \draw[-,line width=0.5pt,black,fill=lightgray] (u3) --+(-1,-1) --+(1,-1) -- (u3);
    \draw[-,line width=0.5pt,black,fill=lightgray] (z1) --+(-1,-1) --+(1,-1) -- (z1);
    \draw[-,line width=0.5pt,black,fill=lightgray] (z2) --+(-1,-1) --+(1,-1) -- (z2);
    \draw[-,line width=0.5pt,black,fill=lightgray] (z3) --+(-1,-1) --+(1,-1) -- (z3);

    \draw (v15) --+(-2,-2) --+(2,-2) -- (v15);
    \draw (v16) --+(-2,-2) --+(2,-2) -- (v16);

    \draw (v22) --+(-2,-2) --+(2,-2) -- (v22);
    \draw (v24) --+(-2,-2) --+(2,-2) -- (v24);

    \draw (v31) --+(-2,-2) --+(2,-2) -- (v31);
    \draw[-,line width=0.5pt,black,fill=lightgray] (v32) --+(-1,-1) --+(1,-1) -- (v32);
    \draw (v34) --+(-2,-2) --+(2,-2) -- (v34);

    \draw (u2) edge[bend right=15] (-4,-2);
    \draw (u2) edge[bend left=17] (19,-2);
    \draw (-4,-2) -- (19,-2);

    \draw (u1) edge[bend right=30] (-6,-4);
    \draw (u1) edge[bend left=33] (28,-4);
    \draw (-6,-4) -- (28,-4);
    
	\draw[-, line width=1pt]  (v17) -- (v13);
    \draw[-, line width=1pt]  (v17) -- (v14);
    \draw[-, line width=1pt]  (v17) -- (v15);
    \draw[-, line width=1pt]  (v17) -- (v16);

    \draw[-, line width=1pt]  (v33) -- (v22);
    \draw[-, line width=1pt]  (v33) -- (u2);
    \draw[-, line width=1pt]  (v33) -- (v24);

    \draw[-, line width=1pt]  (v35) -- (v31);
    \draw[-, line width=1pt]  (v35) -- (v32);
    \draw[-, line width=1pt]  (v35) -- (u1);

    \draw[-, line width=1pt]  (v35) -- (v34);
    \draw[-, line width=1pt]  (v33) -- (u3);
    \draw[-, line width=1pt]  (v17) -- (z1);
    \draw[-, line width=1pt]  (v33) -- (z2);
    \draw[-, line width=1pt]  (v35) -- (z3);


    \path (12,1.5) -- (15,1.5) node [black, midway, sloped] {$\dots$};
    \path (15,8.5) -- (18,8.5) node [black, midway, sloped] {$\dots$};
    \path (19,13.5) -- (22,13.5) node [black, midway, sloped] {$\dots$};
    \draw[decorate, line width=1pt] (u1) -- (v33);
    \draw[linedot] (u2) -- (v17);

    \node at (-1,0.5) {\tiny $B_1^{2i-s}$};
    \node at (4,0.5) {\tiny $B_2^{2i-s}$};
    \node at (9,0.5) {\tiny $B_3^{2i-s}$};
    
    \node at (7,-1.5) {\small $B_2^{2i-1}$};
    \node at (10,-3.5) {\small$B_2^{2i}$};
    \node at (-1,7.5) {\small$B_1^{2i-1}$};
    \node at (12,7.5) {\small$B_3^{2i-1}$};
    
    \node at (-2,12.5) {\small$B_1^{2i}$};
    \node at (16,12.5) {\small$B_3^{2i}$};

    \node[draw, circle, line width=1pt, fill=white]() at (15,2)[] {};
    \node[draw, circle, line width=1pt, fill=white]() at (18,9)[] {};
    \node[draw, circle, line width=1pt, fill=white]() at (22,14)[] {};
    \node[draw, circle, line width=1pt, fill=white]() at (12,2)[] {};
    \node[draw, circle, line width=1pt, fill=white]() at (15,9)[] {};
    \node[draw, circle, line width=1pt, fill=white]() at (19,14)[] {};
    \node[draw, circle, line width=1pt, fill=white]() at (4,9)[] {};
	
\end{tikzpicture}

\caption{A representation of the tree $T_{2i,d}$ illustrating the notation used throughout the proof of Lemma~\ref{lem:TiNotMonot}. Wiggly edges are used to represent paths whose internal vertices have degree $2$. The branch $B_1^{2i}$ is the first branch of $T_{2i,d}$ at $\gamma_{2i}$ that is definitively cleaned, and this happens during the round $j_{2i}$. No vertex of the branches $B_2^{2i}$ and $B_3^{2i}$ has been shot until the round $j_{2i}$. Among the branches $B_2^{2i}$ and $B_3^{2i}$, the first branch that is definitively cleaned is $B_2^{2i}$, and this happens at the round $j'_{2i}>j_{2i}$. The colour grey is used on the small triangles to denote that we do not know the state of the corresponding branches at the same level as $B_1^{2i}$ that are different from $B_2^{2i}$ and $B_3^{2i}$. Observe that $B_2^{2i}$ contains a copy of $T_{2i-1,d}$, rooted at $\gamma_{2i-1}$. Since $B_2^{2i}$ is definitively cleaned at round $j'_{2i}$, we can iterate the same arguments, and define $B_1^{2i-1}$ to be the first branch of $T_{2i-1},d$ at $\gamma_{2i-1}$ that is definitively cleaned, and this happens during the round $j_{2i}< j_{2i-1}<j_{2i}'$, and so on, until we have reached the leaves of $B_2^{2i}$.\label{figure:monotone}
}
\end{figure}
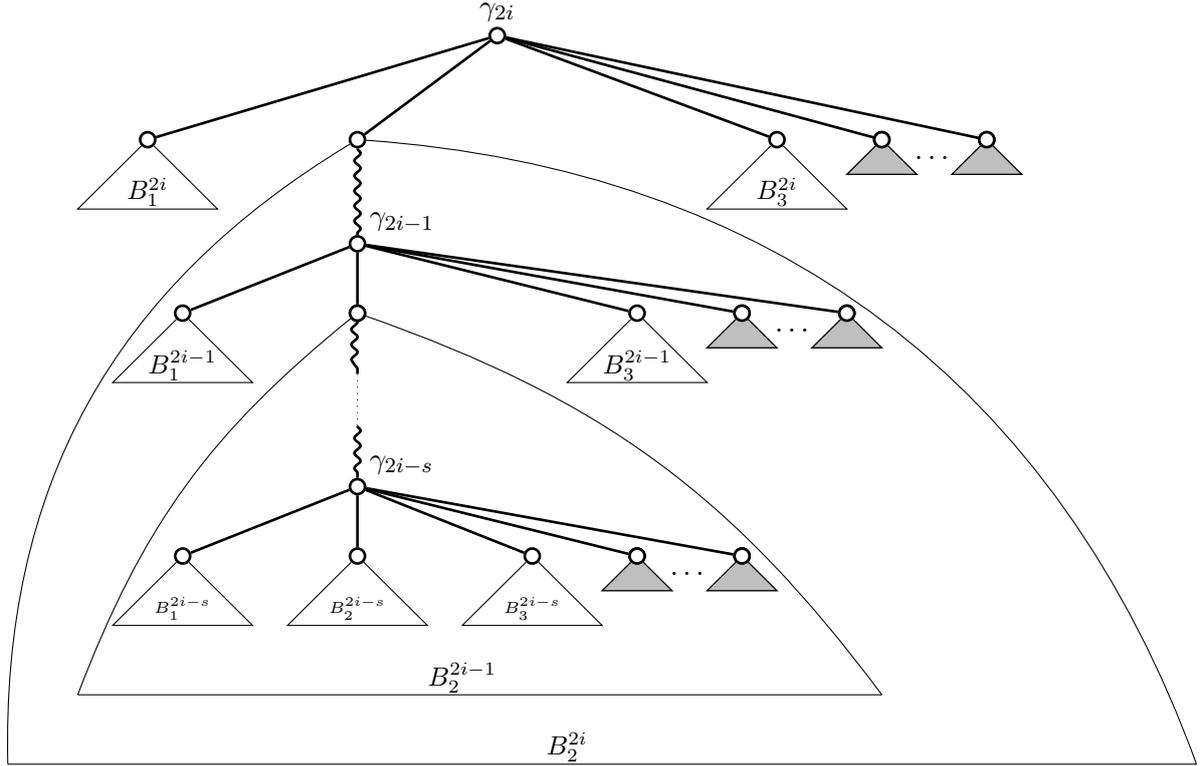

Finally, we will need an extra definition: For any tree $T=(V_r\cup V_w,E)$, and any vertex $v\in V(T)$, let $B$ be any branch at $v$ such that $|V(B)|>1$. For any strategy  ${\cal S}=(S_1,\dots, S_{\ell})$ in $T$ with respect to $V_r$, let $m$ be the minimum integer such that $S_m \cap V(B) \neq \emptyset$ and let $u \in S_m \cap V(B)$ (by Lemma~\ref{lem:independent}, such an integer $m$ exists because $|V(B)|>1$ and $B$ is connected, which implies that $V_r\cap V(B)\neq \emptyset$). Let the \textit{restriction} ${\cal S}_B$ of ${\cal S}$ be the hunter strategy, such that for every $1\leq i\leq \ell$,
\[   
S'_i = 
     \begin{cases}
       S_i\cap V(B), &\quad\text{if } S_i\cap V(B)\neq \emptyset\\
       \{u\}, &\quad\text{otherwise}\\
     \end{cases}
\]


Recall that, $h({\cal S}_B)\leq \max_{1\leq i\leq \ell} S_i\cap V(B)$ by Lemma~\ref{lem:MonotoneSubgraphBip}.

\begin{lemma}\label{lem:TiNotMonot}
For any $i\in \mathbb{N}^*$, $i\geq 3$ and $d \geq 2i$, we have that $mh_{V_r}(T_{2i,d})\geq i$.
\end{lemma}
\begin{proof}
Let $\gamma_{2i}$ denote the root of $T_{2i,d}$.
For purpose of contradiction, let us assume that $mh_{V_r}(T_{2i,d})< i$. By Lemma~\ref{lem:monotoneParsimoniousBip}, there exists a parsimonious monotone winning hunter strategy with respect to $V_r$ using at most $i-1$ hunters; let that strategy be ${\cal S}_{2i}=(S^{2i}_1,\dots, S^{2i}_{\ell})$.

Let $1 \leq j_{2i} \leq \ell$ be the smallest index such that some branch at $\gamma_{2i}$, w.l.o.g., $B^{2i}_1$, is definitively cleaned at round $j_{2i}$. By Lemma~\ref{lem:other-branches-not-touched}, there exist two branches at $\gamma_{2i}$, w.l.o.g., $B^{2i}_2$ and $B^{2i}_3$, such that $(\bigcup_{1\leq q\leq j_{2i}} S_q)\cap V(B^{2i}_2)=\emptyset$ and $(\bigcup_{1\leq q\leq j_{2i}} S_q)\cap V(B^{2i}_3)=\emptyset$.

Let $1 \leq j'_{2i} \leq \ell$ be the smallest index such that at least one branch $B^{2i}_2$ or $B^{2i}_3$ is definitively cleaned. 

Note that $B^{2i}_2$ (resp., $B^{2i}_3$) is connected and has at least two vertices with at least one in $V_r$. Therefore, by Lemma~\ref{lem:independent}, at least one vertex of $B^{2i}_2$ (resp., $B^{2i}_3$) must be shot before the branch is definitively cleaned . Hence, $j_{2i} <j'_{2i}$.
W.l.o.g., assume that $B^{2i}_2$ is definitively cleaned at round $j'_{2i}$ (possibly, $B^{2i}_3$ may also be definitively cleaned at round $j'_{2i}$). 

We now prove by induction on $0\leq h < 2i$ , that there exist $1 \leq j_{2i} <j_{2i-1}<\dots <j_{2i-h} <j'_{2i-h}\leq j'_{2i-h+1}\leq \dots\leq j'_{2i-1}\leq j'_{2i} \leq \ell$ and $\{B^{2i-s}_1,B^{2i-s}_2,B^{2i-s}_3\}_{0\leq s \leq h}$  such that, for every $0 \leq s \leq h$:
\begin{enumerate}
\item $(\bigcup_{1\leq q\leq j_{2i-s}} S_q)\cap V(B^{2i-s}_2)=\emptyset$ and $(\bigcup_{1\leq q\leq j_{2i-s}} S_q)\cap V(B^{2i-s}_3)=\emptyset$;
\item $B^{2i-s}_1,B^{2i-s}_2$ and $B^{2i-s}_3$ are vertex disjoint branches at the root $\gamma_{2i-s}$ of the copy of $T_{2i-s,d}$ contained in $B^{2i-(s-1)}_2$ (with $B^{2i+1}_2=T_{2i,d}$  for $s=0$), each containing a copy of $T_{2i-(s+1),d}$ (with $T_{0,d}=\emptyset$ for $s=h=2i-1$);
\item $B^{2i-s}_1$ is definitively cleaned at round $j_{2i-s}$ (not before, i.e., for every $x<j_{2i-s}$, $B^{2i-s}_1$ is not definitively cleaned at round $x$),  $B^{2i-s}_2$ is definitively cleaned at round $j'_{2i-s}$ (not before), $B^{2i-s}_3$ is definitively cleaned at some round $x \geq j'_{2i-s}$ (not before). 
\end{enumerate}

See Figure~\ref{figure:monotone} for an illustration of the above notation.

We have already proven that the induction hypothesis holds for $h=0$. Let us assume that it holds for some $0 \leq h <2i-1$ and let us show it holds for $h+1$. Let $F$ be the copy of $T_{2i-(h+1),d}$ (rooted at $\gamma_{2i-(h+1)}$) contained in $B^{2i-h}_2$ and let $1 \leq j_{2i-(h+1)} \leq j'_{2i-h}$ be the smallest integer such that some branch $B$ of $F$ at $\gamma_{2i-(h+1)}$ is definitively cleaned. Note that each branch of $F$ at $\gamma_{2i-(h+1)}$ is connected and has at least two vertices with at least one in $V_r$. Therefore, by Lemma~\ref{lem:independent}, at least one vertex of $B$ must be shot before it is definitively cleaned. Hence, $j_{2i-h}  <j_{2i-(h+1)}$. Let $B=B^{2i-(h+1)}_1$. 
Let $\mathcal{S}_{2i-(h+1)}$ denote the restriction of $\mathcal{S}_{2i-h}$ on $F$.
By Lemma~\ref{lem:other-branches-not-touched} 
considering the strategy $\mathcal{S}_{2i-(h+1)}$ on $F$, there exist at least two branches at $\gamma_{2i-(h+1)}$, let us denote them by $B^{2i-(h+1)}_2$ and $B^{2i-(h+1)}_3$, such that $(\bigcup_{1\leq q\leq j_{2i-(h+1)}} S_q)\cap V(B^{2i-(h+1)}_2)=\emptyset$ and $(\bigcup_{1\leq q\leq j_{2i-(h+1)}} S_q)\cap V(B^{2i-(h+1)}_3)=\emptyset$. 
Also, it follows by Lemma~\ref{lem:independent} that $j_{2i-(h+1)}<j'_{2i-(h+1)}$. Finally, $B^{2i-(h+1)}_1$, $B^{2i-(h+1)}_2$ and $B^{2i-(h+1)}_3$ are all contained in $B_{2i-h}$ and, thus, $max(j_{2i-(h+1)},j'_{2i-(h+1)})=j'_{2i-(h+1)}\leq j'_{2i-h}$.
This finishes the proof of the induction step.

For every $1 \leq s \leq 2i$, let $H_s$ be the subgraph induced by $B^s_1$ and $B^s_3$ and $\gamma_{s}$ (so $H_s$ is connected and the subgraphs $H_s$ and $H_{s'}$ are vertex disjoint for any $s \neq s'$). Since $B^s_1$ has been definitively cleaned at round $j_{s}$, has at least two vertices and by Lemma~\ref{lem:independent}, there exists a vertex $x'_s \in V(B^s_1) \cap \bigcup_{1\leq q\leq j_{2i}} S_q$.

Note that $B^s_3$ is connected and has at least two vertices with at least one in $V_r$. Therefore, by Lemma~\ref{lem:independent}, at least one vertex of $B^s_3$ must be shot before the branch is definitively cleaned. Thus, since $B^s_3$ is definitely cleaned at round $z_s\geq j'_{s}$, but not definitely cleaned at a previous round, $S_{z_s}\cap V(B^s_3)\neq \emptyset$. 
Moreover, since $\mathcal{S}_{2i}$ is parsimonious, any vertex $v\in S_{z_s}\cap V(B^s_3)$ is such that $v\in Z_{z_s-1}$. It follows that $(N(v)\cap Z_{z_s-2})\setminus S_{z_s-1}\neq \emptyset$. Let $w_s\in (N(v)\cap Z_{z_s-2})\setminus S_{z_s-1}$ and note that $w_s \in N[V(B^s_3)]=V(B^s_3) \cup \{\gamma_{s}\}$. Since $\mathcal{S}_{2i}$ is monotone, we get that $w_s\notin  S_j$ for every $j< z_s$ and that $N(w_s)\not\subseteq S_j$ for every $j < z_s$, i.e. $w_s$ has not been shot before round $z_s$. Note also that if $w_s=\gamma_{s}$ , then the neighbour of $\gamma_{s}$ in $B_1^{s}$ is in $Z_{z_s-1}$, a contradiction since $z_s> j_{1}$ and $B_1^{s}$ is definitively cleaned at round $j_s\leq j_{1}$. Hence, $w_s \neq \gamma_{s}$ and so $w_s \in V(B^s_3)$.

Similarly, there exists a vertex $w'_s\in N(w_s)\cap Z_{z_s-3}\setminus S_{z_s-2}$ such that $w'_s$ has not been cleaned before round $z_s-1$. Hence, we have two adjacent vertices $w_s$ and $w'_s$ in $N[B^s_3]$ that have never been shot before the round $z_s-1$. Thus, since $\{w_s,w'_s\} \cap V_r \neq \emptyset$, there exists a rabbit trajectory $(\dots,w_s,w'_s,w_s,w'_s,\dots)$ consisting in oscillating between $w_s$ and $w'_s$ which implies that $\{w_s,w'_s\} \cap Z_j \neq \emptyset$ for all $j<z_s$. In particular, since $j_{1}-1 <z_s$, there exists a vertex $y'_s\in N[B^s_3]\cap Z_{j_{1}-1}$. 

For every $1 \leq s \leq 2i$, let  $x_s$ and $y_s$ be two vertices in $V(H_s)$ such that $x_s \in V(H_s) \cap \bigcup_{1\leq q\leq j_{1}} S_q$, $y_s\in V(H_s) \cap Z_{j_{1}-1}$ and the distance between $x_s$ and $y_s$ 
is minimised. 
Let $\mathcal {P}=\{s \mid 1 \leq s \leq 2i, y_s\in S_{j_{1}} \cup S_{j_{1}+1}$ or $x_s\in S_{j_{1}} \cup S_{j_{1}+1}\}$, i.e. $\cal P$ is the set of indices $s$ such that at least one of $y_s$ or $x_s$ is shot during the round $j_{1}$ or $j_{1}+1$. Since ${\cal S}_{2i}$ uses at most $i-1$ hunters, we have that $|\mathcal {P} |\leq 2(i-1)$. Thus, let $1 \leq s^* \leq 2i$ such that $s^* \notin {\cal P}$, i.e., $x_{s^*}, y_{s^*} \notin S_{j_{1}} \cup S_{j_{1}+1}$. It follows from Lemma~\ref{lem:PseudoDistance-1} that $x_{s^*}y_{s^*}\in E(T_{2i,d})$. Since $y_{s^*} \in Z_{j_{2i}-1}\setminus S_{j_{1}}$, we have that that $x_{s^*} \in Z_{j_{1}}$. However, $x_{s^*} \in \bigcup_{1\leq q\leq j_{1}} S_q \setminus S_{j_{1}+1}$ and so, $x_{s^*}$ is recontaminated, contradicting the monotonicity of ${\cal S}_{2i}$.
\end{proof}

\section{Kernelization by vertex cover}\label{sec:kernel}

Let us first remind some of the basic definitions regarding parameterised complexity.
An instance of a parameterised version $\Pi_p$ of a decision problem $\Pi$ is a pair $(I,t)$, where $I$ is an instance of $\Pi$ and $t$ is a non-negative integer, called a \textit{parameter}, associated with $I$. 
We say that $\Pi_p$ is \textit{fixed-parameter tractable} (\FPT) if there exists an algorithm (called as \textit{FPT algorithm}) that, given an instance $(I,t)$ of $\Pi_p$, solves it in time $f(t)\cdot |I|^{\mathcal{O}(1)}$, where
$f$ is any computable function of $t$.


\begin{definition} [Equivalent Instances] \label{equivalent}
Let $\Pi_{1}$ and $\Pi_{2}$ be two parameterised problems. Two instances, $(I, t)\in \Pi_{1}$ and $(I', t')\in \Pi_{2}$, are \textit{equivalent} when $(I, t)$ is a Yes-instance if and only if $(I',t')$ is a Yes-instance.
\end{definition}

A parameterised (decision) problem $\Pi_p$ admits a \textit{kernel} of size $f(t)$, for some function $f$ that depends only on $t$, if the following is true: there exists an algorithm (called a \textit{kernelization algorithm}) that, given as input an instance $(I,t)$ of $\Pi_p$, runs in $(|I|+t)^{\mathcal{O}(1)}$ time and outputs an equivalent instance $(I',t')$ of $\Pi_p$ such that $|I'|\leq f(t)$ and $t'\leq t$. If the function $f$ is polynomial, then the problem is said to admit a \textit{polynomial kernel}. It is well-known that a decidable parameterised problem is \FPT~if and only if it admits a kernel~\cite{bookParameterized}.

Recall that a \textit{vertex cover} of a graph $G$ is any set $U\subseteq V(G)$ such that for every edge $uv\in E(G)$, $U\cap \{u,v\}\neq \emptyset$. The order of a minimum size vertex cover of $G$ is usually referred to as the \textit{vertex cover number} of $G$ and denoted by $vc(G)$.
In what follows,  we  consider the \textsc{Hunters and Rabbit} Problem parameterised by the vertex cover number. That is, an instance $((G,k),t)$ is defined by an input $(G,k)$ where the problem aims at deciding whether $h(G)\leq k$ and the parameter $t$ is any upper bound on $vc(G)$.

First, we have the following observation.
\begin{proposition}\label{prop:boundVC&h}
For any connected graph $G$, $h(G) \leq mh(G) \leq vc(G)$.
\end{proposition}
\begin{proof}
Let $U$ be a vertex cover in $G$ and $I$ be the independent set $V(G) \setminus U$. The hunter player can win simply by shooting all the vertices of $U$ twice. If the rabbit starts at a vertex $u\in U$, then it gets shot in the first round. Otherwise, the rabbit was on a vertex $v \in I$, and then it has to move to a vertex in $U$ (since $I$ is an independent set) that is, $Z_{1}= U$ and then, the rabbit is shot by a hunter in the next round. Finally, note that this strategy is also monotone.
\end{proof}

Let $U$ be a vertex cover of size $t\geq vc(G)$ of $G$ and $I$ be the independent set $V(G) \setminus U$. For each subset $S\subseteq U$, we define the following equivalence class: $\mathcal{C}_S = \{ v \mid v\in I \ \text{and} \ N(v) = S\}$. Next, we have the following crucial lemma.

\begin{lemma}\label{lem:safeRule}
Let $G=(V,E)$ be a connected graph, $U\subseteq V$ be a vertex cover of $G$, $k \geq 1$ and let $S\subseteq U$ be such that $|\mathcal{C}_S| > k+1$. Let $\mathcal{C}_S=\{v_1, \dots, v_q\}$. 
Then, $h(G) \leq k$ (resp., $mh(G)\leq k$) if and only if $h(G[V\setminus \{v_{k+2}, \dots , v_q \}]) \leq k$ (resp., $mh(G[V\setminus \{v_{k+2}, \dots , v_q \}]) \leq k$).  
\end{lemma}
\begin{proof}
By Lemma~\ref{lem:subgraph}, $h(G[V\setminus \{v_{k+2}, \dots , v_q \}]) \leq h(G)$. Similarly, due to Lemma~\ref{lem:MonotoneSubgraph}, $mh(G[V\setminus \{v_{k+2}, \dots , v_q \}]) \leq mh(G)$. So, it only remains to prove that, if $h(G)>k$ (resp., $mh(G)>k$), then $h(G[V\setminus \{v_{k+2}, \dots , v_q \}]) > k$ (resp., $mh(G[V\setminus \{v_{k+2}, \dots , v_q \}]) > k$). Let $H=G[V\setminus \{v_{k+2}, \dots , v_q \}]$ and let $X = \{v_1,\dots, v_{k+1} \}$ (i.e., $X= V(H)\cap {\cal C}_S$).


In the following we show that if $h(G) >k$ (resp., $mh(G)>k$), then $h(H) >k$ (resp., $mh(H)>k$). To this end, we establish that if the rabbit has a winning strategy in $G$ against $k$ hunters, then the rabbit has a winning strategy in $H$ against $k$ hunters. 

\medskip
\noindent \textbf{(1) $\boldsymbol {h(G)>k \implies h(H)>k}$:} Let $\mathcal{S}=(S_1, S_2,\dots,S_{\ell})$ be any hunter strategy (not necessarily winning) in $H$ using at most $k$ hunters. Then, $\cal S$ is a hunter strategy (not necessarily winning) in $G$ using at most $k$ hunters. Since $h(G) > k$, there exists a rabbit-trajectory $\mathcal{R}' = (r'_0, r'_1, \dots,r'_{\ell-1})$ in $G$ such that $r'_i \notin S_{i+1}$ for every $0\leq i <\ell$. Let ${\cal R}=(r_0,\dots,r_{\ell-1})$ be such that, for every $0 \leq i < \ell$, let $r_i=r'_i$ if $r'_i \in V(H)$ and, otherwise, let $r_i$ be any vertex of $X \setminus S_{i+1}$ (such a vertex exists since $|S_{i+1}|\leq k$ and $|X|>k$). Note that $r'_i \neq r_i$ only if $r'_i \notin V(H)$ and therefore $r'_i \in {\cal C}_S$. This implies that, if $r'_i \notin V(H)$, then $r'_{i-1},r'_{i+1} \in S \subset V(H)$ (since $N(r'_i)=S$). Therefore, $r_{i-1}=r'_{i-1}$ and $r_{i+1}=r'_{i+1}$ and $r_{i-1},r_{i+1} \in N_H(r_i)$ (since $r_i \in X$ and so $N_G(r_i)=N_H(r_i)=S$). Therefore, $\cal R$ is a rabbit trajectory in $H$ and, by construction, $r_i \notin S_{i+1}$ for every $0 \leq i < \ell$. Hence, $\mathcal{S}$ is not a winning hunter strategy. Therefore, $h(H)>k$.

\medskip
\noindent\textbf{(2) $\boldsymbol {mh(G)>k \implies mh(H)>k}$:} 
Let $\mathcal{S}=(S_1, S_2,\dots,S_{\ell})$ be any hunter strategy (not necessarily winning) in $H$ using at most $k$ hunters. Then, $\cal S$ is a hunter strategy (not necessarily winning) in $G$ using at most $k$ hunters. Since $mh(G)>k$, for every 
hunter strategy $\mathcal{S}$ in $G$ using at most $k$ hunters, there is a rabbit trajectory $\mathcal{R}'$ that either is a winning rabbit trajectory (the rabbit never gets shot) or recontaminates a vertex (rabbit may be shot at a later round). Now, let $\mathcal{R}$ be built from $\mathcal{R}'$ similarly to the previous case. If the rabbit never gets shot in $\mathcal{R}'$, then due to the arguments presented in $(1)$, the rabbit evades getting shot in $\mathcal{R}$ as well. Hence, we assume that the rabbit gets shot in $\mathcal{R}'$, during, say, a round $p$, but recontaminates a vertex, say $x$, during a round $p'<p$. 
Since only vertices of $H$ can be shot in the hunter strategy $\mathcal{S}$, only the vertices of $H$ can be recontaminated by $\mathcal{R}'$ (recall Lemma~\ref{lem:nonMonotone}). 
Hence $x\in V(H)$. Therefore, $x$ gets recontaminated by $\mathcal{R}$ in $H$ as well. Thus, $mh(H)>k$.
\end{proof}

Finally, we present our kernelization result.
\begin{theorem}\label{T:kernel}
The problem that takes an $n$-node connected graph $G$ and an integer $k\geq 1$ as inputs and decides whether $h(G) \leq k$ (resp., $mh(G)\leq k$), parameterised by any upper bound $t$ on $vc(G)$, admits a kernel of size at most $4^t(t+1)+2t$.
Moreover, this problem can be solved in \FPT~time $(4^t(t+1)+2t)^{t+1}\cdot n^{\mathcal{O}(1)}$.
\end{theorem}
\begin{proof}
The kernelization proceeds as follows. First, if $k>t$, then answer that $h(G)\leq mh(G) \leq k$ (this is correct by Proposition~\ref{prop:boundVC&h}). Otherwise, let $U$ be a vertex cover of size at most $2t$ of $G$ (which can be computed in time $\mathcal{O}(tn)$ by classical $2$-approximation for vertex cover using maximal matching~\cite{bookApprox}). Let $H$ be the graph obtained from $G$ as follows. For every $S \subseteq U$, if $|{\cal C}_S|>k+1$, then remove $|{\cal C}_S|-(k+1)$ vertices from ${\cal C}_S$. By Lemma~\ref{lem:safeRule} (applied iteratively for each $S \subseteq U$), $h(G)\leq k$ (resp., $mh(G)\leq k$) if and only if $h(H)\leq k$ (resp., $mh(H)\leq k$). Moreover, $|V(H)| = |U| + \sum_{S \subseteq U} |{\cal C}_S \cap V(H)| \leq 2t + 2^{2t}(k+1)\leq 4^t(t+1)+2t$ (the last inequality holds by Proposition~\ref{prop:boundVC&h}). Hence, the above algorithm is the desired kernelization algorithm.

Finally, applying the XP-algorithm~\cite{AbramovskayaFGP16}, it can be decided in time $|V(H)|^{k+1}$ whether $h(H) \leq k$. Since, by Proposition~\ref{prop:boundVC&h}, $k \leq t$, this gives the \FPT~algorithm that decides whether $h(G) \leq k$ (resp., $mh(G)\leq k$) in time  $(4^t(t+1)+2t)^{t+1}\cdot n^{\mathcal{O}(1)}$.
\end{proof}

\section{Some Future Directions}\label{sec:futureDirections}

In this paper, we studied the \textsc{Hunters and Rabbit} game by defining the notion of monotonicity for this game. Using this notion of monotonicity, we characterised the monotone hunter number for various classes of graphs. Moreover, we established that, unlike several graph searching games, the monotonicity helps in this game, i.e., the $h(G)$ can be arbitrary smaller than $mh(G)$. 

There are still several challenging open questions in this area. The most important among them is the computational complexity of \textsc{Hunters and Rabbit}. Although our results establish that computing $mh(G)$ is \textsf{NP}-hard, the computational complexity of computing/deciding $h(G)$ remains open, even if $G$ is restricted to be a tree graph.

We also established that both \textsc{Hunters and Rabbit}, as well as its monotone variant, are \FPT~parameterised by $vc(G)$ by designing  exponential kernels. It is not difficult to see that both of these variants admit AND Composition parameterised by the solution size (by taking the disjoint union of the instances). Thus, since computing $mh(G)$ is $\mathsf{NP}$-hard and $pw(G) \leq mh(G)\leq pw(G)+1$, it is unlikely for \textsc{Monotone Hunters and Rabbit} parameterised by $k+pw(G)$ to admit a polynomial compression. Note that the same cannot be argued about \textsc{Hunters and Rabbit} since it is not yet proved to be \textsf{NP}-hard. Moreover, since $mh(G)$ is closely related to $pw(G)$ and pathwidth admits a polynomial kernel with respect to $vc(G)$~\cite{chapelle2017treewidth}, it might be interesting to see if deciding $mh(G)\leq k$ (resp., $h(G) \leq k$) also admits a polynomial kernel when parameterised by $vc(G)$. Moreover, another interesting research direction is to study the parameterised complexity of both these games by considering  parameters such as solution size, treewidth, and pathwidth.

Finally, we propose some open questions concerning the computation of $h(G)$ for various graph classes including trees, cographs, and interval graphs. Specifically, it will be interesting to design a polynomial time algorithm, similar to Algorithm~\ref{alg:Tree}, to compute $h(T)$ for a tree $T$, a question that was already proposed in~\cite{AbramovskayaFGP16}. The natural way that one could tackle this question is through the notion of monotonicity, which we defined and studied in this paper. Unfortunately, Theorem~\ref{theo:gapMonotone&nonM} implies that such an approach will not work. This means that a positive answer to this question (if any) would require the introduction of new tools and techniques. 
Moreover, it would be interesting to know the monotone hunter number of grids.




\section{Acknowledgements}
This work is partially funded by the project UCA JEDI (ANR-15-IDEX-01) and the EUR DS4H Investments in the Future (ANR-17-EURE-004) and the ANR Digraphs and the ERC grant titled PARAPATH and the IFCAM project ``Applications of graph homomorphisms'' (MA/IFCAM/18/39).

\bibliographystyle{plainurl}
\bibliography{biblio}

\newpage

\appendix
\section{Appendix}

\propreachBip*
\begin{proof}
This clearly holds if $p=0$ (i.e.,  when $v\in V_r$) since $Z_0=V_r$. If $p=1$ (i.e.,  when $v\in V_w$), there exists a rabbit trajectory $(r_0=u \in N(v)\cap V_r\setminus S_1, r_1=v)$ and so $v \in Z_1$. Hence, we assume that $p>1$. 
    
The rabbit can follow the following strategy depending on whether $p$ is odd or even:
\begin{enumerate}
    \item $p$ is odd (and so $v\in V_w$): The rabbit can follow the following trajectory: $(r_0 = u,r_1=x, \ldots, r_{p-1} = u, r_{p}=v)$ where, for $q<p$, $r_{q} = u$ if $q$ is even and $r_q = x$ if $q$ is odd. 
    \item $p$ is even (and so, $v\in V_r$): The rabbit can follow the following trajectory: $(r_0 = x,r_1=u, \ldots, r_{p-1} = u, r_{p}=v)$ where, for $q<p$, $r_{q} = x$ if $q$ is even and $r_q = u$ if $q$ is odd.
\end{enumerate}
In both cases, for every $0\leq j <p$, $r_j\notin S_{j+1}$ since $p\leq i$ and $x,u \notin \bigcup_{j<i}S_j$. Therefore, $v\in Z_p$.
\end{proof}

\lemreachBip*
\begin{proof}
Let us recall that for any hunter strategy in $G$ with respect to $V_r$, if $q$ is even, then $Z_q\subseteq V_r$ and $Z_q\subseteq V_w$ otherwise. 
Thus , if $p=0$ or if $i=2$, $Z_i\subseteq Z_p$. Also, if $i=1$, then $Z_i=Z_p$. Hence, let us assume that $p\geq 1$ and $i>2$. Let $v \in Z_i$. Since $v\in Z_i$, there exists a rabbit trajectory $R=(r_0,\dots, r_{i-2}=x,r_{i-1}=u, r_i=v)$ such that, for any $0\leq j<i$, $r_j\notin S_{j+1}$. By definition of a rabbit trajectory, $u \in N(v)$ and $x\in N(u)$. Moreover, by monotonicity of ${\cal S}$, since $u\in Z_{i-1}\setminus S_i$ (resp. $x\in Z_{i-2}\setminus S_{i-1}$), $u\notin \bigcup_{q\leq i} S_q$ (resp. $x\notin \bigcup_{q\leq i-1} S_q$). By Proposition~\ref{prop:reachBip}, $v \in Z_{p}$ for each $p\leq i$ such that $p$ and $i$ has the same parity.
\end{proof}

\lemnonMonotoneBip*
\begin{proof}
Towards a contradiction, assume that the statement of the lemma is false, i.e., for every vertex $v\in V$ and every $1\leq i \leq \ell$, if $v\in Z_{i-1} \setminus S_i$ then $v\notin \bigcup_{p< i}S_p$. 

Since $\mathcal{S}$ is non-monotone and winning, there exists a vertex $u$ such that $u$ is cleared at a round $1\leq q\leq \ell -2$, and then recontaminated at a round $j> q$ (i.e., $u\in Z_j\setminus S_{j+1}$). Moreover, by our assumption, $u$ is cleared by shooting each contaminated vertex in $N(u)$ at round $q$, i.e., $Z_{q-1} \cap N(u) \subseteq S_q$.
W.l.o.g., let us assume that $u\in V_r$. Therefore, $N(u)\subseteq V_w$ and $q-1$ is odd. 

Let us show that $N(u) \subseteq \bigcup_{p\leq q} S_p$. Let us assume that there exists a vertex $x \in N(u)$ such that $x\notin \bigcup_{p< q} S_p$.  Since both $u,x\notin \bigcup_{p< q} S_q$, $u\in N(x)$, $q-1$ is odd and $x\in V_w$, by Proposition~\ref{prop:reachBip}, we get that $x\in Z_{q-1}$. Therefore, $x\in Z_{q-1} \cap N(u) \subseteq S_q$. Hence, $N(u) \subseteq \bigcup_{p\leq q} S_p$.

Since $v \in Z_j$ then there exists $w \in (N(u)\cap Z_{j-1})\setminus S_j$ and $w\in \bigcup_{p< j}S_p$, i.e., $w$ satisfies the statement of the lemma, a contradiction. 
\end{proof}

\lemMonotoneSubgraphBip*
\begin{proof}
Let $\mathcal{S}=(S_1,\dots,S_\ell)$ be a monotone winning hunter strategy for $G$ with respect to $V_r$.

If $|V(H)|=1$, the result clearly holds since $mh(H)=0$. Hence, let us assume that $|V(H)|>1$.
Let $m$ be the minimum integer such that $S_m \cap V(H) \neq \emptyset$ and let $u \in S_m \cap V(H)$ (by Lemma~\ref{lem:independent}, such an integer $m$ exists because $|V(H)|>1$ and $H$ is connected, which implies that $V_r\cap V(H)\neq \emptyset$). Let $\mathcal{S}'=(S'_1,\dots,S'_\ell)$ be the hunter strategy, such that for every $1\leq i\leq \ell$,
\[   
S'_i = 
     \begin{cases}
       S_i\cap V(H), &\quad\text{if } S_i\cap V(H)\neq \emptyset\\
       \{u\}, &\quad\text{otherwise}\\
     \end{cases}
\]
 First, we have the following claim.
\begin{claim}\label{C:L1Bip}
    For every $0 \leq i\leq  \ell$ and for any vertex $v\in V(H)$, if $v\in Z_i({\cal S}')$, then $v\in Z_i({\cal S})$.
\end{claim}
\begin{proofclaim}
    Let $\mathcal{R}=(r_0,\dots, r_i=v)$ be a rabbit trajectory in $H$ such that for any $0\leq j < i$, $r_j r_{j+1}\in E(H)$, $r_j\notin S'_{j+1}$ and $r_i=v$ (such a trajectory exists since $v\in Z_i({\cal S}')$). By construction of $\mathcal{S}'$, for any $1\leq j\leq \ell$, $S_j\cap V(H)\subseteq S'_j$. Therefore, $\mathcal{R}$ is also a rabbit trajectory in $G$ with $r_j \notin S_{j+1}$, for all $0 \leq j < i$. Thus, $v\in Z_i({\cal S})$. 
\end{proofclaim}

Let us show that $\mathcal{S}'$ is a monotone winning hunter strategy in $H$ with respect to $V_r\cap V(H)$.
First, we show that $\mathcal{S}'$ is indeed a winning hunter strategy in $H$ with respect to $V_r\cap V(H)$. Towards a contradiction, assume that $\mathcal{S}'$ is not a winning strategy in $H$ w.r.t. $V_r\cap V(H)$. This implies that $Z_\ell ({\cal S}')\neq \emptyset$. Hence, Claim~\ref{C:L1Bip} implies that $Z_\ell ({\cal S})\neq \emptyset$, contradicting the fact that $\mathcal{S}$ is a winning hunter strategy in $G$ with respect to $V_r$.

Thus, $\mathcal{S}'$ is a winning strategy in $H$ with respect to $V_r\cap V(H)$. Next, we establish that $\mathcal{S}'$ is indeed monotone. Towards a contradiction, let us assume that $\mathcal{S}'$ is non-monotone. Hence, by Lemma~\ref{lem:nonMonotoneBip}, there exist $v \in V(H)$ and $1 \leq q <i \leq \ell$ such that $v \in S'_q$ and $v \in Z_i({\cal S'})\setminus S'_{i+1}$. By Claim~\ref{C:L1Bip} and because $v\in Z_i({\cal S}')$, $v \in Z_i({\cal S})$.
Since $S_p\cap V(H)\subseteq S'_p$ for any $1\leq p\leq \ell$ and because $v\notin S'_{i+1}$, $v\notin S_{i+1}$. 

If $v=u$, $i+1>m$ (since $u \in S'_p$ for all $1 \leq p \leq m$) and so $v \in S_m$ and in $Z_i({\cal S}) \setminus S_{i+1}$, contradicting the monotonicity of $\cal S$.

Otherwise, $v \neq u$. By construction of $\mathcal{S}'$, $S'_p \setminus \{u\} \subseteq S_p$ for all $1 \leq p \leq \ell$. Hence,  $v \in S_q$ and $v \in Z_i({\cal S}) \setminus S_{i+1}$, contradicting the monotonicity of $\cal S$.

Finally, the fact that $h(\mathcal{S}')\leq \max_{1\leq i\leq \ell} |S_i\cap V(H)|\leq h({\cal S})$ completes the proof.
\end{proof}

\lemmonotoneParsimoniousBip*
\begin{proof}
Let $\mathcal{S}=(S_1,\dots, S_\ell)$ be a monotone winning hunter strategy with respect to $V_r$ using at most $k\geq mh_{V_r}(G)$ hunters such that $\ell$ is minimum. If ${\cal S}$ is parsimonious, we are done. Otherwise, among such strategies, Let us consider ${\cal S}$ that maximizes the first round $1 \leq j < \ell$ that makes $\mathcal{S}$ not parsimonious. There are several cases to be considered.
\begin{itemize}
\item
Let $\mathcal{Z}({\cal S})=(Z_0({\cal S}),\dots,Z_{\ell}({\cal S}))$ be the set of contaminated vertices for each round of $\mathcal{S}$. If there exists an integer $\ell' < \ell$ such that $Z_{\ell'}({\cal S}) = \emptyset$, then $\mathcal{S}=(S_1,\dots, S_{\ell'})$ is also a winning hunter strategy with respect to $V_r$ using at most $k$ hunters, contradicting the minimality of $\ell$. 

Hence, we may assume that $Z_{i}({\cal S}) \neq \emptyset$ for every $0 \leq i < \ell$. 
\item
Let $1 \leq j \leq  \ell$ be the smallest integer such that $S_j \setminus Z_{j-1}({\cal S}) \neq \emptyset$ (if no such integer exists, then ${\cal S}$ is parsimonious and we are done). If $S_j \cap Z_{j-1}({\cal S}) \neq \emptyset$, replace $S_j$ by $S_j \cap Z_{j-1}({\cal S})$. This leads to a winning monotone hunter strategy ${\cal S}'$  (indeed, $Z_h({\cal S})=Z_h({\cal S}')$ for all $1 \leq h \leq \ell$) contradicting the maximality of $j$.

Hence, we may assume that $S_j \cap Z_{j-1}({\cal S}) = \emptyset$. Note that this implies that $j<\ell$ (since otherwise, $\cal S$ would not be winning).

\item 
If any, let $0<i$ be the minimum integer such that $S_{j+2i} \cap Z_{j+2i-1}({\cal S}) \neq \emptyset$. Let $v \in S_{j+2i} \cap Z_{j+2i-1}({\cal S})$. Since $b\in Z_{j+2i-1}({\cal S})$,  by Lemma~\ref{lem:reach2Bip}, $v\in Z_{j-1+2i'}$ for every $0\leq i'\leq i$.

Then, for every $0 \leq i' <i$, replace $S_{j+2i'}$ by $\{v\}$. Let us prove that this leads to a monotone hunter strategy contradicting the maximality of $j$. 

Let ${\cal S'}$ be the strategy obtained by the above modifications. First, note that, for any  $0\leq h< j$, $S_h=S'_h$ and so $Z_h({\cal S})= Z_h({\cal S}')$. By definition, $Z_j({\cal S})=\{x\in V\mid\exists y\in Z_{j-1}({\cal S})\setminus S_j \wedge (xy\in E)\}$ and, since $S_j\cap Z_{j-1}=\emptyset$, we get $Z_j({\cal S})=\{x\in V\mid\exists y\in Z_{j-1}({\cal S}) \wedge (xy\in E)\}$. On the other hand, $Z_j({\cal S}')=\{x\in V\mid\exists y\in Z_{j-1}({\cal S'})\setminus S'_j \wedge (xy\in E)\} = \{x\in V\mid\exists y\in Z_{j-1}({\cal S})\setminus \{v\} \wedge (xy\in E)\}$ since $Z_{j-1}({\cal S}')=Z_{j-1}({\cal S})$ and $S'_j=\{v\}$. Hence, $Z_j({\cal S}')\subseteq Z_j({\cal S})$. By induction on $j\leq i' \leq \ell$ and using the same arguments, we get that $Z_{i'}({\cal S}') \subseteq Z_{i'}({\cal S})$ for every $j\leq i'\leq \ell$.
Thus, ${\cal S}'$ is a winning hunter strategyin $G$ with respect to $W\subseteq V_r$ using at most $k\geq mh_{V_r}(G)$ hunters (because $Z_{\ell}({\cal S}')\subseteq Z_{\ell}({\cal S})=\emptyset$).
It remains to show that ${\cal S}'$ is monotone. 

For purpose of contradiction, let us assume that ${\cal S}'$ is non-monotone. By Lemma~\ref{lem:nonMonotoneBip}, there exists a vertex $x$  and $1  < m \leq \ell$ such that $x\in Z_{m-1}({\cal S}')\setminus S'_m$ and $x \in \bigcup_{h<m}S'_h$. 

If $x \neq v$, then by definition of ${\cal S}'$ (for every $1 \leq r \leq \ell$, either $S'_r=S_r$ or $S'_r=\{v\}$) and because $Z_r({\cal S}')\subseteq Z_r({\cal S})$ for all $1 \leq r \leq \ell$, we get that $x \in \bigcup_{h<m}S_h$ and $x\in Z_{m-1}({\cal S})\setminus S_m$ which contradicts the monotonicity of ${\cal S}$. Hence, let us assume that $x=v$. 

Recall that we proved that $v\in Z_{j-1}({\cal S})\setminus S_j$. Therefore, by monotonicity of $\cal S$ and by the definition of $i$ above, $v \notin \bigcup_{r<j+2i}S_r$ which implies that $v \notin \bigcup_{r<j}S'_r$. Since $v \in S'_{j+2i'}$ for all $0 \leq i' \leq i$ and, because $v \in Z_{j+2i-1}$ and by parity, $v \notin Z_{j+2i'}({\cal S}')$ for all $0 \leq i' \leq i$, we get that $m>j+2i$. This means that $v \in Z_{m-1}({\cal S})\setminus S_m$ (because $Z_r({\cal S}')\subseteq Z_r({\cal S})$ for all $1 \leq r \leq \ell$ and $S'_r=S_r$ for all $r\geq m>j+2i$) and $v \in S_{j+2i}$, which contradicts the monotonicity of $\cal S$. 

Hence, we may assume that $S_{j+2i'} \cap Z_{j+2i'-1}({\cal S}) = \emptyset$ for all $0 \leq i'$ such that $j+2i' \leq \ell$.

\item 
If $Z_{j+1}({\cal S})=Z_{j-1}({\cal S})$, then remove $S_j$ and $S_{j+1}$ from ${\cal S}$. Let ${\cal S}'$ be the obtained strategy. We have that $Z_r({\cal S})=Z_r({\cal S}')$ and $S_r=S'_r$ for all $r <j$ and $Z_{r+2}({\cal S})=Z_r({\cal S}')$ and $S_{r+2}=S'_r$ for all $j \leq r \leq \ell-2$. Hence, ${\cal S}'$ is winning since $Z_{\ell-2}({\cal S'})=\emptyset$. Moreover, if ${\cal S}'$ is non-monotone, then, by Lemma~\ref{lem:nonMonotone}, there exists $x$  and $1  < m \leq \ell$ such that $x\in Z_{m-1}({\cal S}')\setminus S'_m$ and $x \in \bigcup_{h<m}S'_h$. If $m< j$, this implies that $x\in Z_{m-1}({\cal S})\setminus S_m$ and $x \in \bigcup_{h<m}S_h$ contradicting the monotonicity of $\cal S$. Otherwise ($m\geq j$), $x\in Z_{m+1}({\cal S})\setminus S_{m+2}$ and $x \in \bigcup_{h<m+2}S_h$, also contradicting the monotonicity of $\cal S$.

Hence, we may assume that $Z_{j+1}({\cal S})\neq Z_{j-1}({\cal S})$.
\item
Note first that by Lemma~\ref{lem:reach2Bip}, $Z_{j+1}({\cal S})\subseteq Z_{j-1}({\cal S})$. Thus, $Z_{j+1}({\cal S})\setminus Z_{j-1}({\cal S})= \emptyset$. Hence, since $Z_{j+1}({\cal S})\neq Z_{j-1}({\cal S})$, we get that $Z_{j-1}({\cal S})\setminus Z_{j+1}({\cal S})\neq \emptyset$.

Hence, there exists $v \in Z_{j-1}({\cal S})\setminus Z_{j+1}({\cal S})$. Let ${\cal S}'$ be obtained by replacing $S_j$ by $\{v\}$. By arguments similar to the ones of the third item of this proof, we can prove that ${\cal S}'$ is a monotone hunter strategy contradicting the maximality of $j$.
\end{itemize}
This completes the proof.
\end{proof}

\lemPseudoMonotoneProperties*
\begin{proof}
First, for purpose of contradiction, let us assume that $v\in S_i\cap S_j$ and $v\notin S_{i+2}$ with $i<j$. Note that, since $G$ is bipartite, that ${\cal S}$ is a parsimonious strategy with respect to $V_r$ and $S_i \cap S_j \neq \emptyset$, thus $i$ and $j$ must have the same parity and so, $j\geq i+4$ (because $v\notin S_{i+2}$). Since $\mathcal{S}$ is parsimonious, $v \in Z_{i-1} \cap Z_{j-1}$. Thus, by Lemma~\ref{lem:reach2Bip}, $v\in Z_{i+1}\setminus S_{i+2}$ and $v\in S_i$, contradicting the monotonicity of ${\cal S}$.
Hence, the first statement holds.

Let us now prove the second statement. Let us assume now that $v\in V_r$ (the case $v \in V_w$ can be handled similarly). If $v\notin Z_{i-1}$ for any odd $1\leq i< \ell$, then $v\notin Z_{j-1}$ for any odd $i\leq j\leq \ell$ by Lemma~\ref{lem:reach2Bip}. Therefore, since ${\cal S}$ is parsimonious, $v\notin S_j$ for any odd $i\leq j\leq \ell$. Moreover, since $G$ is bipartite and that ${\cal S}$ is a parsimonious strategy with respect to $V_r$, for any even $1\leq j\leq \ell$, $v\notin S_j$. 
\end{proof}
\end{document}